%% file: linear-constraints.tex
\documentclass[acmsmall,usenames,dvipsnames,natbib=false,acmthm=false,screen]{acmart}

\usepackage[backend=biber,datamodel=acmdatamodel, style=acmauthoryear,uniquename=false]{biblatex}
\usepackage{amsmath}
\usepackage{amsthm}
\usepackage{subcaption}
\usepackage{hyperref}
\hypersetup{
    colorlinks,
    linkcolor={red!50!black},
    citecolor={blue!50!black},
    urlcolor={blue!80!black}
  }
\usepackage[plain]{fancyref}
\usepackage[capitalize]{cleveref}
\usepackage{mathpartir}
\usepackage{newunicodechar}
\input{newunicodedefs}

\usepackage[supertabular,implicitLineBreakHack]{ottalt}

  \input{ott.tex}

  \renewottcommands[ott]

\makeatletter
\@ifundefined{lhs2tex.lhs2tex.sty.read}%
  {\@namedef{lhs2tex.lhs2tex.sty.read}{}%
   \newcommand\SkipToFmtEnd{}%
   \newcommand\EndFmtInput{}%
   \long\def\SkipToFmtEnd#1\EndFmtInput{}%
  }\SkipToFmtEnd

\newcommand\ReadOnlyOnce[1]{\@ifundefined{#1}{\@namedef{#1}{}}\SkipToFmtEnd}
\usepackage{amstext}
\usepackage{amssymb}
\usepackage{stmaryrd}
\DeclareFontFamily{OT1}{cmtex}{}
\DeclareFontShape{OT1}{cmtex}{m}{n}
  {<5><6><7><8>cmtex8
   <9>cmtex9
   <10><10.95><12><14.4><17.28><20.74><24.88>cmtex10}{}
\DeclareFontShape{OT1}{cmtex}{m}{it}
  {<-> ssub * cmtt/m/it}{}

\DeclareFontShape{OT1}{cmtt}{bx}{n}
  {<5><6><7><8>cmtt8
   <9>cmbtt9
   <10><10.95><12><14.4><17.28><20.74><24.88>cmbtt10}{}
\DeclareFontShape{OT1}{cmtex}{bx}{n}
  {<-> ssub * cmtt/bx/n}{}

\newcommand{\Conid}[1]{\mathit{#1}}
\newcommand{\Varid}[1]{\mathit{#1}}
\newcommand{\anonymous}{\kern0.06em \vbox{\hrule\@width.5em}}

\renewcommand{\leq}{\leqslant}

\usepackage{polytable}

\@ifundefined{mathindent}%
  {\newdimen\mathindent\mathindent\leftmargini}%
  {}%

\def\resethooks{%
  \global\let\SaveRestoreHook\empty
  \global\let\ColumnHook\empty}
\newcommand*{\savecolumns}[1][default]%
  {\g@addto@macro\SaveRestoreHook{\savecolumns[#1]}}
\newcommand*{\restorecolumns}[1][default]%
  {\g@addto@macro\SaveRestoreHook{\restorecolumns[#1]}}
\newcommand*{\aligncolumn}[2]%
  {\g@addto@macro\ColumnHook{\column{#1}{#2}}}

\resethooks

\newcommand{\onelinecommentchars}{\quad-{}- }
\newcommand{\commentbeginchars}{\enskip\{-}
\newcommand{\commentendchars}{-\}\enskip}

\newcommand{\visiblecomments}{%
  \let\onelinecomment=\onelinecommentchars
  \let\commentbegin=\commentbeginchars
  \let\commentend=\commentendchars}

\newcommand{\invisiblecomments}{%
  \let\onelinecomment=\empty
  \let\commentbegin=\empty
  \let\commentend=\empty}

\visiblecomments

\newlength{\blanklineskip}
\newcommand{\hsindent}[1]{\quad}
\let\hspre\empty
\let\hspost\empty

\EndFmtInput
\makeatother
\ReadOnlyOnce{polycode.fmt}%
\makeatletter

\newcommand{\hsnewpar}[1]%
  {{\parskip=0pt\parindent=0pt\par\vskip #1\noindent}}

\newcommand{\hscodestyle}{}

\newcommand{\sethscode}[1]%
  {\expandafter\let\expandafter\hscode\csname #1\endcsname
   \expandafter\let\expandafter\endhscode\csname end#1\endcsname}

  {\par\noindent
   \advance\leftskip\mathindent
   \hscodestyle
   \let\\=\@normalcr
   \let\hspre\(\let\hspost\)%
   \pboxed}%
  {\endpboxed\)%
   \par\noindent
   \ignorespacesafterend}

  {\hsnewpar\abovedisplayskip
   \advance\leftskip\mathindent
   \hscodestyle
   \let\hspre\(\let\hspost\)%
   \pboxed}%
  {\endpboxed%
   \hsnewpar\belowdisplayskip
   \ignorespacesafterend}

  {\hsnewpar\abovedisplayskip
   \advance\leftskip\mathindent
   \hscodestyle
   \let\\=\@normalcr
   \(\pboxed}%
  {\endpboxed\)%
   \hsnewpar\belowdisplayskip
   \ignorespacesafterend}

\newcommand{\plainhs}{\sethscode{plainhscode}}

\plainhs

  {\hsnewpar\abovedisplayskip
   \advance\leftskip\mathindent
   \hscodestyle
   \let\\=\@normalcr
   \(\parray}%
  {\endparray\)%
   \hsnewpar\belowdisplayskip
   \ignorespacesafterend}

  {\parray}{\endparray}

  {\(\parray}{\endparray\)}

\def\codeframewidth{\arrayrulewidth}
\RequirePackage{calc}

  {\parskip=\abovedisplayskip\par\noindent
   \hscodestyle
   \arrayrulewidth=\codeframewidth
   \tabular{@{}|p{\linewidth-2\arraycolsep-2\arrayrulewidth-2pt}|@{}}%
   \hline\framedhslinecorrect\\{-1.5ex}%
   \let\endoflinesave=\\
   \let\\=\@normalcr
   \(\pboxed}%
  {\endpboxed\)%
   \framedhslinecorrect\endoflinesave{.5ex}\hline
   \endtabular
   \parskip=\belowdisplayskip\par\noindent
   \ignorespacesafterend}

\newcommand{\framedhslinecorrect}[2]%
  {#1[#2]}

  {\(\def\column##1##2{}%
   \let\>\undefined\let\<\undefined\let\\\undefined
   \newcommand\>[1][]{}\newcommand\<[1][]{}\newcommand\\[1][]{}%
   \def\fromto##1##2##3{##3}%
   }{\) }%

  {\let\orighscode=\hscode
   \let\origendhscode=\endhscode
   \def\endhscode{\def\hscode{\endgroup\def\@currenvir{hscode}\\}\begingroup}
   \orighscode\def\hscode{\endgroup\def\@currenvir{hscode}}}%
  {\origendhscode
   \global\let\hscode=\orighscode
   \global\let\endhscode=\origendhscode}%

\makeatother
\EndFmtInput
    \def\noeditingmarks{}

  \usepackage{xargs}
  \usepackage[colorinlistoftodos,prependcaption,textsize=tiny]{todonotes}
  \ifx\noeditingmarks\undefined
      \newcommand{\note}[1]{{\color{blue}{\begin{itemize} \item {#1} \end{itemize}}}}
      \newenvironment{alt}{\color{red}}{}

      \newcommandx{\jp}[2][1=]{\todo[linecolor=purple,backgroundcolor=purple!25,bordercolor=purple,#1]{#2}}
      \newcommandx{\csongor}[2][1=]{\todo[linecolor=blue,backgroundcolor=blue!25,bordercolor=purple,#1]{#2}}
      \newcommandx{\rae}[2][1=]{\todo[linecolor=magenta,backgroundcolor=magenta!25,bordercolor=magenta,#1]{RAE: #2}}
      \newcommandx{\nw}[2][1=]{\todo[linecolor=green,backgroundcolor=green!25,bordercolor=green,#1]{NW: #2}}

      \newcommandx{\unsure}[2][1=]{\todo[linecolor=red,backgroundcolor=red!25,bordercolor=red,#1]{#2}}
      \newcommandx{\info}[2][1=]{\todo[linecolor=green,backgroundcolor=green!25,bordercolor=green,#1]{#2}}
      \newcommandx{\change}[2][1=]{\todo[linecolor=blue,backgroundcolor=blue!25,bordercolor=blue,#1]{#2}}
      \newcommandx{\inconsistent}[2][1=]{\todo[linecolor=blue,backgroundcolor=blue!25,bordercolor=red,#1]{#2}}
      \newcommandx{\critical}[2][1=]{\todo[linecolor=blue,backgroundcolor=blue!25,bordercolor=red,#1]{#2}}
      \newcommand{\improvement}[1]{\todo[linecolor=pink,backgroundcolor=pink!25,bordercolor=pink]{#1}}
      \newcommandx{\resolved}[2][1=]{\todo[linecolor=OliveGreen,backgroundcolor=OliveGreen!25,bordercolor=OliveGreen,#1]{#2}} 
  \else

      \newcommand{\note}[1]{}

      \newcommand{\unsure}[2][1=]{}
      \newcommand{\info}[2][1=]{}
      \newcommand{\change}[2]{}
      \newcommand{\inconsistent}[2]{}
      \newcommand{\critical}[2]{}
      \newcommand{\improvement}[1]{}
      \newcommand{\resolved}[2]{}

      \newcommand{\csongor}[2][1=]{}
      \newcommand{\jp}[2][1=]{}
      \newcommandx{\rae}[2][1=]{}
      \newcommandx{\nw}[2][1=]{}
  \fi

  \newcommand{\constraintcolour}{\color{RoyalBlue}}
  \newcommand{\constraintfont}[1]{{\constraintcolour#1}}
  \newcommand{\multiplicitycolour}{\color{RoyalBlue}}
  \newcommand{\multiplicityfont}[1]{{\multiplicitycolour#1}}

  \newcommand{\vdashi}{%
      \mathrel{%
          \vdash\hspace*{-4pt}%
          \raisebox{0.9pt}{\scalebox{.66}{\(\blacktriangleright\)}}%
      }%
  }
  \newcommand{\vdashs}{⊢_{\mathsf{s}}}
  \newcommand{\vdashsimp}{⊢_{\mathsf{s}}^{\mathsf{a}}}
  \newcommand{\scale}{\constraintfont{\cdot}}
  \newcommand{\constraintop}[1]{\mathop{\constraintfont{#1}}}
  \newcommand{\aand}{\constraintop{\&}}
  \DeclareMathOperator*{\bigaand}{\vcenter{\hbox{\Large\&}}}
  \newcommand{\lollycirc}{\raisebox{-0.255ex}{\scalebox{1.4}{$\circ$}}}
  \newcommand{\Lolly}{\constraintop{=\kern-1.1ex \lollycirc}}
  \newcommand{\FatArrow}{\constraintop{\Rightarrow}}
  \newcommand{\RLolly}{\mathop{\constraintfont \circledless}}
  
  \newcommand{\qtensor}{\constraintop{\otimes}}
  \newcommand{\dsterm}[2]{\llbracket #2 \rrbracket_{#1}}
  \newcommand{\dstype}[1]{\llbracket #1 \rrbracket}
  \newcommand{\dsevidence}[1]{\llbracket #1 \rrbracket^{\mathbf{ev}}}

  \newcommand{\keyword}[1]{\mathbf{#1}}
  \newcommand{\klet}{\keyword{let}}
  \newcommand{\kcase}{\keyword{case}}

  \newcommand{\packbox}{\raisebox{-0.08ex}{\scalebox{0.88}{$\square$}}}
  \newcommand{\kin}{\keyword{in}}
  
  \newcommand{\kunpack}{\klet\packbox}

  \newcommand{\bnfeq}{\mathrel{\Coloneqq}}
  \newcommand{\bnfor}{\mathrel{\mid}}

  \theoremstyle{acmplain}
  \newtheorem{theorem}{Theorem}[section]
  \newtheorem{lemma}[theorem]{Lemma}
  \newtheorem{corollary}[theorem]{Corollary}
  \newtheorem{property}[theorem]{Property}

  \theoremstyle{acmdefinition}
  \newtheorem{definition}[theorem]{Definition}

\setcopyright{rightsretained}
\acmPrice{}
\acmDOI{10.1145/3547626}
\acmYear{2022}
\copyrightyear{2022}
\acmSubmissionID{icfp22main-p15-p}
\acmJournal{PACMPL}
\acmVolume{6}
\acmNumber{ICFP}
\acmArticle{95}
\acmMonth{8}

\begin{document}

\title{Linearly Qualified Types}
\subtitle{Generic inference for capabilities and uniqueness}

\author{Arnaud Spiwack}
\orcid{0000-0002-5985-2086}
\affiliation{
  \institution{Tweag}
  \city{Paris}
  \country{France}
}
\email{arnaud.spiwack@tweag.io}
\author{Csongor Kiss}
\orcid{0000-0002-0195-2420}
\affiliation{
  \institution{Imperial College London}
  \city{London}
  \country{United Kingdom}
}
\email{csongor.kiss14@imperial.ac.uk}
\author{Jean-Philippe Bernardy}
\orcid{0000-0002-8469-5617}
\affiliation{
  \institution{University of Gothenburg}
  \city{Gothenburg}
  \country{Sweden}
}
\email{jean-philippe.bernardy@gu.se}
\author{Nicolas Wu}
\orcid{0000-0002-4161-985X}
\affiliation{
  \institution{Imperial College London}
  \city{London}
  \country{United Kingdom}
}
\email{n.wu@imperial.ac.uk}
\author{Richard A.~Eisenberg}
\orcid{0000-0002-7669-9781}
\affiliation{
  \institution{Tweag}
  \city{Paris}
  \country{France}
}
\email{rae@richarde.dev}



\begin{abstract}
  A linear parameter must be consumed exactly once in the body of its
  function. When declaring resources such as file handles and manually
  managed memory as linear arguments, a linear type system can verify
  that these resources are used safely.
  However, writing code with explicit linear arguments requires bureaucracy.
  This paper presents \emph{linear constraints}, a front-end feature for
  linear typing that
  decreases the bureaucracy of working with linear types.
  Linear constraints are implicit linear arguments that are
  filled in automatically by the compiler.
  We present linear constraints as a qualified type system,
  together with an inference algorithm which extends
  \textsc{ghc}'s existing constraint solver algorithm. Soundness of
  linear constraints is ensured by the fact that they desugar into
  Linear Haskell.
\end{abstract}

\maketitle

\newcommand{\maybesmall}{\small}

\begin{acks}
  Jean-Philippe Bernardy is supported by grant
  \grantnum{clasp}{2014-39} from the \grantsponsor{clasp}{Swedish
    Research Council}{https://www.vr.se}, which funds the Centre for
  Linguistic Theory and Studies in Probability (CLASP) in the
  Department of Philosophy, Linguistics, and Theory of Science at the
  University of Gothenburg.
  Nicolas Wu is supported by
  \grantsponsor{EPSRC}{EPSRC}{https://www.ukri.org/councils/epsrc/}
  Grant \grantnum{EPSRC}{EP/S028129/1}.

\end{acks}

\section{Introduction}
\label{sec:introduction}

Linear type systems have seen a renaissance in recent years in
various programming communities. Rust's ownership system guarantees
memory safety for systems programmers, Haskell's \textsc{ghc}~9.0 includes
support for linear types, and even
dependently typed programmers can now use linear types with Idris 2.
All of these systems are vastly different in ergonomics and
scope. Rust uses dedicated syntax and code generation
to support management of resources, while Linear Haskell is a
type system change without any other impact on the compiler, such as in the
code generator or runtime system.\rae{But actually we'd \emph{like} to have
an impact on the runtime system to allow update-in-place.} Linear Haskell
is designed to be general purpose, but
using its linear arguments to emulate Rust's ownership model is a tedious exercise. It
requires
the programmer to carefully thread resource tokens.

To get a sense of the power and the tedium of using linear types, consider the
following function:
\begin{hscode}\SaveRestoreHook
\column{B}{@{}>{\hspre}l<{\hspost}@{}}%
\column{3}{@{}>{\hspre}l<{\hspost}@{}}%
\column{8}{@{}>{\hspre}l<{\hspost}@{}}%
\column{19}{@{}>{\hspre}l<{\hspost}@{}}%
\column{E}{@{}>{\hspre}l<{\hspost}@{}}%
\>[B]{}\Varid{read2AndDiscard}\mathbin{::}\Conid{MArray}\;\Varid{a}⊸(\Conid{Ur}\;\Varid{a},\Conid{Ur}\;\Varid{a}){}\<[E]%
\\
\>[B]{}\Varid{read2AndDiscard}\;\Varid{arr0}\mathrel{=}{}\<[E]%
\\
\>[B]{}\hsindent{3}{}\<[3]%
\>[3]{}\mathbf{let}\;{}\<[8]%
\>[8]{}(\Varid{arr1},\Varid{x}){}\<[19]%
\>[19]{}\mathrel{=}\Varid{read}\;\Varid{arr0}\;\mathrm{0}{}\<[E]%
\\
\>[8]{}(\Varid{arr2},\Varid{y}){}\<[19]%
\>[19]{}\mathrel{=}\Varid{read}\;\Varid{arr1}\;\mathrm{1}{}\<[E]%
\\
\>[8]{}(){}\<[19]%
\>[19]{}\mathrel{=}\Varid{free}\;\Varid{arr2}{}\<[E]%
\\
\>[B]{}\hsindent{3}{}\<[3]%
\>[3]{}\mathbf{in}\;(\Varid{x},\Varid{y}){}\<[E]%
\ColumnHook
\end{hscode}\resethooks
This function reads the first two elements of an array and returns them after
deallocating the array.
Linearity enables the array library to ensure that there is only one reference
to the array, and therefore it can be mutated in-place without violating
referential transparency. Let us stress that this uniqueness property
is an invariant of the array library, not an intrinsic property of
linear functions.

After the array has been freed, it is no longer
possible to read or write to it.
Notice that the \ensuremath{\Varid{read}} function consumes the array and returns a
fresh array, to be used in future operations. Operationally, the array remains
the same, but each operation assigns a new name to it, thus facilitating tracking
references statically. Finally, \ensuremath{\Varid{free}} consumes the array without returning a
new one, statically guaranteeing that it can no longer be used.
The values \ensuremath{\Varid{x}} and \ensuremath{\Varid{y}} read from the array are returned; their
types include elements wrapped by the \ensuremath{\Conid{Ur}}
(pronounced ``unrestricted'',  and corresponding to the
``!'' operator of \citet{girard-linear-logic}) type, allowing them to be used
arbitrarily many times. This works because \ensuremath{\Varid{read2AndDiscard}} takes a restricted-use array
containing unrestricted elements.
In a non-linear language, one would have to forgo referential transparency to
handle mutable operations either by using a monadic interface or allowing
arbitrary effects.
Compare the above function with what one would write in a non-linear, impure
language:
\begin{hscode}\SaveRestoreHook
\column{B}{@{}>{\hspre}l<{\hspost}@{}}%
\column{3}{@{}>{\hspre}l<{\hspost}@{}}%
\column{9}{@{}>{\hspre}l<{\hspost}@{}}%
\column{14}{@{}>{\hspre}l<{\hspost}@{}}%
\column{E}{@{}>{\hspre}l<{\hspost}@{}}%
\>[B]{}\Varid{read2AndDiscard}\mathbin{::}\Conid{MArray}\;\Varid{a}\to (\Varid{a},\Varid{a}){}\<[E]%
\\
\>[B]{}\Varid{read2AndDiscard}\;\Varid{arr}\mathrel{=}{}\<[E]%
\\
\>[B]{}\hsindent{3}{}\<[3]%
\>[3]{}\mathbf{let}\;{}\<[9]%
\>[9]{}\Varid{x}{}\<[14]%
\>[14]{}\mathrel{=}\Varid{read}\;\Varid{arr}\;\mathrm{0}{}\<[E]%
\\
\>[9]{}\Varid{y}{}\<[14]%
\>[14]{}\mathrel{=}\Varid{read}\;\Varid{arr}\;\mathrm{1}{}\<[E]%
\\
\>[9]{}(){}\<[14]%
\>[14]{}\mathrel{=}\Varid{unsafeFree}\;\Varid{arr}{}\<[E]%
\\
\>[B]{}\hsindent{3}{}\<[3]%
\>[3]{}\mathbf{in}\;(\Varid{x},\Varid{y}){}\<[E]%
\ColumnHook
\end{hscode}\resethooks
This non-linear version does not guarantee that there is a unique reference to
the array, so freeing the array is a potentially unsafe operation.
However, it is simpler because there is less bureaucracy to manage: we are
clearly interacting with the \emph{same} array throughout, and this version makes
that apparent.
We see here a clear tension between extra safety and clarity of code---one
we wish, as language designers, to avoid. 
%
How can we get the compiler to see that the array is used safely
without explicit threading?

Following well-known ideas~\citep{DBLP:conf/popl/CraryWM99,DBLP:conf/esop/SmithWM00,DBLP:conf/tic/WalkerM00}, our approach is to let arrays be unrestricted, but
associate linear capabilities (such as \ensuremath{\constraintfont{\Conid{Read}}}, \ensuremath{\constraintfont{\Conid{Write}}}) to them.
In fact, we show in this paper that such linear capabilities
are the natural analogue of Haskell's type class
constraints to the setting of linear types.
We call these new constraints \emph{linear constraints}.
Like class constraints,
linear constraints are propagated implicitly by the compiler.
Like linear arguments, they can safely be used to track resources such as arrays
or file handles. Thus, linear constraints are the combination of these two
concepts, which have been studied independently elsewhere
\citep{OutsideIn,LinearHaskell,hh-ll,resource-management-for-ll-proof-search}.

With our extension, we can write a new pure version of \ensuremath{\Varid{read2AndDiscard}} which does
not require explicit threading of the array:
\begin{hscode}\SaveRestoreHook
\column{B}{@{}>{\hspre}l<{\hspost}@{}}%
\column{3}{@{}>{\hspre}l<{\hspost}@{}}%
\column{9}{@{}>{\hspre}l<{\hspost}@{}}%
\column{18}{@{}>{\hspre}c<{\hspost}@{}}%
\column{18E}{@{}l@{}}%
\column{21}{@{}>{\hspre}l<{\hspost}@{}}%
\column{E}{@{}>{\hspre}l<{\hspost}@{}}%
\>[B]{}\Varid{read2AndDiscard}\mathbin{::}\constraintfont{(\Conid{Read}\;\Varid{n},\Conid{Write}\;\Varid{n})}\Lolly \Conid{UArray}\;\Varid{a}\;\Varid{n}\to (\Conid{Ur}\;\Varid{a},\Conid{Ur}\;\Varid{a}){}\<[E]%
\\
\>[B]{}\Varid{read2AndDiscard}\;\Varid{arr}\mathrel{=}{}\<[E]%
\\
\>[B]{}\hsindent{3}{}\<[3]%
\>[3]{}\mathbf{let}\;{}\<[9]%
\>[9]{}\packbox\Varid{x}{}\<[18]%
\>[18]{}\mathrel{=}{}\<[18E]%
\>[21]{}\Varid{read}\;\Varid{arr}\;\mathrm{0}{}\<[E]%
\\
\>[9]{}\packbox\Varid{y}{}\<[18]%
\>[18]{}\mathrel{=}{}\<[18E]%
\>[21]{}\Varid{read}\;\Varid{arr}\;\mathrm{1}{}\<[E]%
\\
\>[9]{}\packbox(){}\<[18]%
\>[18]{}\mathrel{=}{}\<[18E]%
\>[21]{}\Varid{free}\;\Varid{arr}{}\<[E]%
\\
\>[B]{}\hsindent{3}{}\<[3]%
\>[3]{}\mathbf{in}\;(\Varid{x},\Varid{y}){}\<[E]%
\ColumnHook
\end{hscode}\resethooks
The only changes from the impure version are that this version explicitly requires having
read and write access to the array,
and pattern-matching against $\packbox$ (read ``pack'') is necessary in order to access the linear constraint
packed in the result of \ensuremath{\Varid{read}} and \ensuremath{\Varid{free}}. (\Fref{sec:implicit-existentials}
suggests how we can get rid of $\packbox$, too.)
Crucially, the resource representing the ownership of the
array is a linear constraint and is separate from the array itself, which no
longer needs to be threaded manually.

Our contributions are as follows:
\begin{itemize}
\item A system of qualified types that allows a constraint assumption
  to be given a multiplicity (linear or unrestricted). Linear assumptions are used precisely
  once in the body of a definition
  (\Fref{sec:qualified-type-system}). This system supports examples
  that have motivated the design of several resource-aware systems,
  such as ownership \textit{à la} Rust (\Fref{sec:memory-ownership}), or
  capabilities in the style of Mezzo~\cite{mezzo-permissions} or
  \textsc{ats}~\cite{AtsLinearViews}; accordingly, our system points
  towards a possible unification of these lines of research.

\item Applications of this qualified type system to allow writing
  \begin{itemize}
  \item resource-aware algorithms without explicit threading (\Fref{sec:memory-ownership}); and
  \item functions whose result can only be used linearly (\Fref{sec:Unique-constraint})
  \end{itemize}

\item An inference algorithm that respects the multiplicity of
  assumptions. We prove that this algorithm is sound with respect to
  our type system~(\Fref{sec:type-inference}). It consists of
  \begin{itemize}
  \item a constraint generation
    algorithm~(\cref{sec:constraint-generation}). The language of
    generated constraints tracks multiplicities.
  \item a solver~(\cref{sec:constraint-solver}) for the generated
    constraints, which restricts proof-search algorithms for linear
    logic in order to be \emph{guess
      free}~\cite[Section~6.4]{OutsideIn}. A guess-free algorithm
    ensures that constraint inference is predictable and insensitive
    to small changes in the source program; it is necessarily
    incomplete.
  \end{itemize}
\end{itemize}
Our language is given semantics by desugaring into a core language based
on that of \citet{LinearHaskell}.
Our design is intended to work well with other features of Haskell and
\textsc{ghc} extensions. Indeed, we have a prototype implementation (\cref{sec:implementation}).

\section{Background: Linear Haskell}
\label{sec:linear-types}

This section, mostly cribbed from \citet[Section 2.1]{LinearHaskell},
describes our baseline approach, as released in \textsc{ghc}~9.0.
Linear Haskell adds a new type of functions,
dubbed \emph{linear functions}, and written \ensuremath{\Varid{a}\mathbin{⊸}\Varid{b}}.\footnote{The linear function
  type and its notation come from linear
  logic~\cite{girard-linear-logic}, to which the phrase \emph{linear
    types} refers. All the various design of linear typing in the
  literature amount to adding such a linear function type, but details
  can vary wildly. See~\citet[Section 6]{LinearHaskell} for an analysis
  of alternative approaches.} A linear function consumes its
  argument exactly once. Linear Haskell defines it as follows: 

\begin{quote}
\ensuremath{\Varid{f}\mathbin{::}\Varid{a}\mathbin{⊸}\Varid{b}} guarantees that if \ensuremath{(\Varid{f}\;\Varid{u})} is consumed exactly once,
then the argument \ensuremath{\Varid{u}} is consumed exactly once.
\end{quote}
To make sense of this statement we need to know what ``consumed exactly once'' means.
Our definition is based on the type of the value concerned:
\begin{definition}[Consume exactly once]~ \label{def:consume}
\begin{itemize}
\item To consume a value of atomic base type (like \ensuremath{\Conid{Int}}) exactly once, just evaluate it.
\item To consume a function exactly once, apply it to one argument, and then consume its result exactly once.
\item To consume a pair exactly once, pattern-match on it, and then consume each component exactly once.
\item In general, to consume a value of an algebraic datatype exactly once, pattern-match on it,
  and then consume all its linear components exactly once.
\end{itemize}
\end{definition}
\subsection{Multiplicities}
\label{sec:multiplicities}
The usual arrow type \ensuremath{\Varid{a}\to \Varid{b}} can be recovered using \ensuremath{\Conid{Ur}}, as \ensuremath{\Conid{Ur}\;\Varid{a}\mathbin{⊸}\Varid{b}}, but Linear
Haskell provides a first-class treatment of \ensuremath{\Varid{a}\to \Varid{b}}, thus ensuring
backwards compatibility with Haskell. In practice, the type-checker
records the \emph{multiplicity} of every introduced variable: \(1\) for linear
arguments and \(ω\) for unrestricted ones. This way, one can give a unified treatment
of both arrow types~\citep{linearity-and-pi-calculus}.
$$
\begin{array}{lcll}
    \multiplicityfont{ \pi }  ,   \multiplicityfont{ \rho }   & \bnfeq &   \multiplicityfont{ \ottsym{1} }   \bnfor   \multiplicityfont{ \omega }   & \text{Multiplicities}
\end{array}
$$

We stress that a multiplicity of \(1\) restricts \emph{how the variable can be used}. It does not
restrict \emph{which values can be substituted for it}.
In particular, a linear function cannot assume that it is given the
unique pointer to its argument.  For example, if \ensuremath{\Varid{f}\mathbin{::}\Varid{a}\mathbin{⊸}\Varid{b}}, then
the following is fine:
\begin{hscode}\SaveRestoreHook
\column{B}{@{}>{\hspre}l<{\hspost}@{}}%
\column{E}{@{}>{\hspre}l<{\hspost}@{}}%
\>[B]{}\Varid{g}\mathbin{::}\Varid{a}\to (\Varid{b},\Varid{b}){}\<[E]%
\\
\>[B]{}\Varid{g}\;\Varid{x}\mathrel{=}(\Varid{f}\;\Varid{x},\Varid{f}\;\Varid{x}){}\<[E]%
\ColumnHook
\end{hscode}\resethooks
The type of \ensuremath{\Varid{g}} makes no guarantees about how it uses \ensuremath{\Varid{x}}.
In particular, \ensuremath{\Varid{g}} can pass \ensuremath{\Varid{x}} to \ensuremath{\Varid{f}}.

Pattern matching on a value of type \ensuremath{\Conid{Ur}\;\Varid{a}} yields a payload of multiplicity
$  \multiplicityfont{ \omega }  $, even when the scrutinee has multiplicity $  \multiplicityfont{ \ottsym{1} }  $. In general, given a
multiplicity set, the desired
(sub)structural rules can be obtained by endowing multiplicities with the appropriate
semiring structure~\citep{abel_unified_2020}. In
this paper, we use the same multiplicity structure as Linear Haskell:\footnote{Even though linear Haskell additionally supports multiplicity
polymorphism, we do not support multiplicity polymorphism on
constraint arguments.  Linear Haskell takes advantage of multiplicity
polymorphism to avoid duplication of higher-order functions. The
prototypical example is \ensuremath{\Varid{map}\mathbin{::}(\Varid{a}\;\mathop{\to_{\multiplicityfont{m}}}\;\Varid{b})\to [\mskip1.5mu \Varid{a}\mskip1.5mu]\;\mathop{\to_{\multiplicityfont{m}}}\;[\mskip1.5mu \Varid{b}\mskip1.5mu]}, where \ensuremath{\mathop{\to_{\multiplicityfont{m}}}} is the notation for a function arrow of
multiplicity \ensuremath{\multiplicityfont{\Varid{m}}}.  First-order functions, on the other
hand, do not need multiplicity polymorphism, because linear functions
can be $\eta$-expanded into unrestricted functions as explained in
\cref{sec:linear-types}. Higher-order functions whose arguments are
themselves constrained functions are rare, so we do not yet see the
need to extend multiplicity polymorphism to apply to constraints.
Futhermore, it is not clear how to extend the constraint solver of
\cref{sec:constraint-solver} to support multiplicity-polymorphic
constraints.}${}^,$\footnote{Here and in the rest of the paper we adopt the convention that
equations defining a function by pattern matching are marked with a
$\left\{\right.$ to their left.}
$$
\begin{array}{c@{\qquad\qquad}c}
\left\{
  \begin{array}{lcl}
      \multiplicityfont{ \pi  \ottsym{+}  \rho }   & = &   \multiplicityfont{ \omega }  
  \end{array}
\right.
&
\left\{
  \begin{array}{lcl}
      \multiplicityfont{  \ottsym{1} {⋅} \pi  }   & = &   \multiplicityfont{ \pi }   \\
      \multiplicityfont{  \omega {⋅} \pi  }   & = &   \multiplicityfont{ \omega }  
  \end{array}
\right.
\end{array}
$$

\subsection{Shortcomings of Linear Haskell that we address}
The \ensuremath{\Varid{read}} function in \cref{sec:introduction} consumes the array it
operates on. Therefore, the same array can no longer be used in
further operations: doing so would result in a type error.  To resolve
this, a new name for the same array is produced by each operation.

From the perspective of the programmer, this is unwanted boilerplate.
Minimizing such boilerplate is the main aim of this paper.  Our
approach is to let the array be non-linear, and let its
capabilities (\emph{i.e.} having read or write access) be \emph{linear
constraints}. Once these capabilities are consumed, the array can no
longer be read from or written to without triggering a compile time
error.

A further drawback of today's Linear Haskell is that
the programmer cannot restrict how a linear function is used.
For example, suppose we want to use linear types to create
a pure interface to arrays that supports in-place mutation;
the interface is safe only if we guarantee that arrays cannot
be aliased. Because the result of a hypothetical \ensuremath{\Varid{newArray}}
function can be stored in an unrestricted variable (of multiplicity
$  \multiplicityfont{ \omega }  $), the linearity system cannot prevent its aliasing.
Instead, \citet[Fig.~2]{LinearHaskell} use a continuation-passing
style to enforce non-aliasing.


\section{Working With Linear Constraints}
\label{sec:what-it-looks-like}
\jp{Why is this here? This current presentation is stuttering. IMO. It should go back somewhere in sec. 4.}
\begin{figure}%
  \maybesmall
  \centering
  \begin{subfigure}{.3\linewidth}%
    \noindent%
\begin{hscode}\SaveRestoreHook
\column{B}{@{}>{\hspre}l<{\hspost}@{}}%
\column{6}{@{}>{\hspre}l<{\hspost}@{}}%
\column{E}{@{}>{\hspre}l<{\hspost}@{}}%
\>[B]{}\phantom{(\mathbf{type}\;\constraintfont{\Conid{RW}\;\Varid{n}}\mathrel{=}\constraintfont{(\Conid{Read}\;\Varid{n},\Conid{Write}\;\Varid{n})})}{}\<[E]%
\\
\>[B]{}\Varid{new}{}\<[6]%
\>[6]{}\mathbin{::}\Conid{Int}\to (\Conid{MArray}\;\Varid{a}⊸\Conid{Ur}\;\Varid{r})⊸\Conid{Ur}\;\Varid{r}{}\<[E]%
\\
\>[B]{}\Varid{write}\mathbin{::}\Conid{MArray}\;\Varid{a}⊸\Conid{Int}\to \Varid{a}\to \Conid{MArray}\;\Varid{a}{}\<[E]%
\\
\>[B]{}\Varid{read}\mathbin{::}\Conid{MArray}\;\Varid{a}⊸\Conid{Int}\to (\Conid{MArray}\;\Varid{a},\Conid{Ur}\;\Varid{a}){}\<[E]%
\\
\>[B]{}\Varid{free}\mathbin{::}\Conid{MArray}\;\Varid{a}⊸(){}\<[E]%
\ColumnHook
\end{hscode}\resethooks
\caption{Linear Types}
\label{fig:linear-interface}
  \end{subfigure}
  \hfill
  \begin{subfigure}{.55\linewidth}
\begin{hscode}\SaveRestoreHook
\column{B}{@{}>{\hspre}l<{\hspost}@{}}%
\column{6}{@{}>{\hspre}l<{\hspost}@{}}%
\column{E}{@{}>{\hspre}l<{\hspost}@{}}%
\>[B]{}\mathbf{type}\;\constraintfont{\Conid{RW}\;\Varid{n}}\mathrel{=}\constraintfont{(\Conid{Read}\;\Varid{n},\Conid{Write}\;\Varid{n})}{}\<[E]%
\\
\>[B]{}\Varid{new}{}\<[6]%
\>[6]{}\mathbin{::}\constraintfont{\constraintfont{\Conid{Linearly}}}\Lolly \Conid{Int}\to \exists\;\Varid{n}.\Conid{UArray}\;\Varid{a}\;\Varid{n}\RLolly\constraintfont{\Conid{RW}\;\Varid{n}}{}\<[E]%
\\
\>[B]{}\Varid{write}\mathbin{::}\constraintfont{\Conid{RW}\;\Varid{n}}\Lolly \Conid{UArray}\;\Varid{a}\;\Varid{n}\to \Conid{Int}\to \Varid{a}\to ()\RLolly\constraintfont{\Conid{RW}\;\Varid{n}}{}\<[E]%
\\
\>[B]{}\Varid{read}\mathbin{::}\constraintfont{\Conid{Read}\;\Varid{n}}\Lolly \Conid{UArray}\;\Varid{a}\;\Varid{n}\to \Conid{Int}\to \Conid{Ur}\;\Varid{a}\RLolly\constraintfont{\Conid{Read}\;\Varid{n}}{}\<[E]%
\\
\>[B]{}\Varid{free}\mathbin{::}\constraintfont{\Conid{RW}\;\Varid{n}}\Lolly \Conid{UArray}\;\Varid{a}\;\Varid{n}\to (){}\<[E]%
\ColumnHook
\end{hscode}\resethooks
\caption{Linear Constraints}
\label{fig:constraints-interface}
  \end{subfigure}
\caption{Interfaces for mutable arrays}
\end{figure}

Consider the Haskell function \ensuremath{\Varid{show}}:
\begin{hscode}\SaveRestoreHook
\column{B}{@{}>{\hspre}l<{\hspost}@{}}%
\column{E}{@{}>{\hspre}l<{\hspost}@{}}%
\>[B]{}\Varid{show}\mathbin{::}\constraintfont{\Conid{Show}\;\Varid{a}}\FatArrow \Varid{a}\to \Conid{String}{}\<[E]%
\ColumnHook
\end{hscode}\resethooks
In addition to the function arrow \ensuremath{\to }, common to all functional
programming languages, the type of this function features a constraint arrow \ensuremath{\FatArrow }.
Everything to the
left of a constraint arrow is called a \emph{constraint}, and will be
highlighted in \constraintfont{blue} throughout the paper. Here
\ensuremath{\constraintfont{\Conid{Show}\;\Varid{a}}} is a
class constraint.

Constraints are handled implicitly by
the typechecker. That is, if we want to \ensuremath{\Varid{show}} the integer \ensuremath{\Varid{n}\mathbin{::}\Conid{Int}} we would write \ensuremath{\Varid{show}\;\Varid{n}}, and the typechecker is responsible for proving that \ensuremath{\constraintfont{\Conid{Show}\;\Conid{Int}}} holds, without
intervention from the programmer.

For our \ensuremath{\Varid{read2AndDiscard}} example, the \ensuremath{\constraintfont{(\Conid{Read}\;\Varid{n},\Conid{Write}\;\Varid{n})}} (abbreviated as \ensuremath{\constraintfont{\Conid{RW}\;\Varid{n}}})
constraint represents read and write
access to the array tagged with the type variable \ensuremath{\Varid{n}}. (The full \textsc{api} under consideration appears
in \cref{fig:constraints-interface}.)
That is, the constraint \ensuremath{\constraintfont{\Conid{RW}\;\Varid{n}}}
is provable if and only if the array tagged with \ensuremath{\Varid{n}} is readable and writable.
This constraint is linear: it must be consumed (that is, used as an assumption
in a function call) exactly once.
In order to manage linearity implicitly, this paper introduces a
linear constraint arrow (\ensuremath{\Lolly }), much like Linear Haskell introduces a linear
function arrow (\ensuremath{⊸}). Constraints to the left of a linear constraint
arrow are \emph{linear constraints}.
Using the linear constraint \ensuremath{\constraintfont{\Conid{RW}\;\Varid{n}}}, we can give
the following type to \ensuremath{\Varid{free}}:

\begin{hscode}\SaveRestoreHook
\column{B}{@{}>{\hspre}l<{\hspost}@{}}%
\column{E}{@{}>{\hspre}l<{\hspost}@{}}%
\>[B]{}\Varid{free}\mathbin{::}\constraintfont{\Conid{RW}\;\Varid{n}}\Lolly \Conid{UArray}\;\Varid{a}\;\Varid{n}\to (){}\<[E]%
\ColumnHook
\end{hscode}\resethooks

There are a few things to notice:
\begin{itemize}
\item We have introduced a new type variable \ensuremath{\Varid{n}}. In contrast,
  the version in \Cref{fig:linear-interface} without linear
  constraints  has type \ensuremath{\Varid{free}\mathbin{::}\Conid{MArray}\;\Varid{a}⊸()}.
  The type variable \ensuremath{\Varid{n}} is a type-level tag used to identify the array.
\item The run-time variable representing the array can now be used multiple
  times. Instead of restricting the use of this variable,
  the linear constraint \ensuremath{\constraintfont{\Conid{RW}\;\Varid{n}}} controls access to the
  array.
\item If we have a single, linear, \ensuremath{\constraintfont{\Conid{RW}\;\Varid{n}}}
  available, then after \ensuremath{\Varid{free}} there will not be any \ensuremath{\constraintfont{\Conid{RW}\;\Varid{n}}}
  left to use, thus preventing the array from being used after freeing.
  This is precisely what we were trying to achieve.
\end{itemize}

The above deals with freeing an array and ensuring that it cannot be used afterwards.
However, we still need to explain how a constraint \ensuremath{\constraintfont{\Conid{RW}\;\Varid{n}}} can come
into scope.
The type of \ensuremath{\Varid{new}} with linear constraints is as follows:
\begin{hscode}\SaveRestoreHook
\column{B}{@{}>{\hspre}l<{\hspost}@{}}%
\column{6}{@{}>{\hspre}l<{\hspost}@{}}%
\column{E}{@{}>{\hspre}l<{\hspost}@{}}%
\>[B]{}\Varid{new}{}\<[6]%
\>[6]{}\mathbin{::}\constraintfont{\constraintfont{\Conid{Linearly}}}\Lolly \Conid{Int}\to \exists\;\Varid{n}.\Conid{UArray}\;\Varid{a}\;\Varid{n}\RLolly\constraintfont{\Conid{RW}\;\Varid{n}}{}\<[E]%
\ColumnHook
\end{hscode}\resethooks

This type, too, illustrates several new aspects:
\begin{itemize}
\item The \ensuremath{\constraintfont{\constraintfont{\Conid{Linearly}}}} constraint is a linear constraint, though
it takes no type parameter. It restricts the result of \ensuremath{\Varid{new}}
to be used linearly, meaning that any variable that stores the result of \ensuremath{\Varid{new}} must
have multiplicity $  \multiplicityfont{ \ottsym{1} }  $. \ensuremath{\constraintfont{\constraintfont{\Conid{Linearly}}}} is explained more fully in \Fref{sec:Unique-constraint}.

\item Because \ensuremath{\Conid{UArray}} is now parameterised by a type-level tag \ensuremath{\Varid{n}}, \ensuremath{\Varid{new}} must
return a \ensuremath{\Conid{UArray}} with a fresh such \ensuremath{\Varid{n}}. This is achieved through returning an
existentially-quantified type\footnote{There is a variety of ways that existential types can
  be worked into a language. The existentials that we use here may be understood
  as a generalisation of those presented by \citet[Chapter
  24]{tapl}.
  However, \citet{existentials} work out an
  approach that makes linear constraints even easier to use, as we
  discuss in \Fref{sec:implicit-existentials}. In order to separate
  concerns, we do not build our formalism on \citet{existentials},
  instead modeling our existentials on the more widely known formulation described by \citet{tapl}.
  We additionally freely
omit the \ensuremath{\exists\;\Varid{a}_{1}\mathbin{...}\Varid{a}_{\Varid{n}}.} or \ensuremath{\RLolly\constraintfont{\Conid{Q}}} parts when they are
empty.
  The idea of using existentials to return linear capabilities can be attributed to \citet{DBLP:conf/esop/FluetMA06}.
} packing the type variable \ensuremath{\Varid{n}}. Such types are introduced with the $\packbox$ constructor.

\item Not only do we need a fresh type variable \ensuremath{\Varid{n}}, but we also need to introduce
the linear constraint \ensuremath{\constraintfont{\Conid{RW}\;\Varid{n}}} for use in subsequent calls to \ensuremath{\Varid{read}} and \ensuremath{\Varid{free}}.
Our existentials also allow packing a constraint, thanks to the \ensuremath{\RLolly} operator.
\end{itemize}

With all these features working together, we see that \ensuremath{\Varid{new}} returns a non-duplicable
\ensuremath{\Conid{UArray}} tagged with \ensuremath{\Varid{n}}, accessible only when the \ensuremath{\constraintfont{\Conid{RW}\;\Varid{n}}} constraint is available.

We must also ensure that \ensuremath{\Varid{read}} can both promise to operate only on a readable array
and that the array remains readable afterwards. That is, \ensuremath{\Varid{read}} must both consume
a linear constraint \ensuremath{\constraintfont{\Conid{Read}\;\Varid{n}}} and also produce a fresh linear constraint
\ensuremath{\constraintfont{\Conid{Read}\;\Varid{n}}}, as we see in \cref{fig:constraints-interface}, and repeated
here:
\begin{hscode}\SaveRestoreHook
\column{B}{@{}>{\hspre}l<{\hspost}@{}}%
\column{E}{@{}>{\hspre}l<{\hspost}@{}}%
\>[B]{}\Varid{read}\mathbin{::}\constraintfont{\Conid{Read}\;\Varid{n}}\Lolly \Conid{UArray}\;\Varid{a}\;\Varid{n}\to \Conid{Int}\to \Conid{Ur}\;\Varid{a}\RLolly\constraintfont{\Conid{Read}\;\Varid{n}}{}\<[E]%
\ColumnHook
\end{hscode}\resethooks
We have now seen all the ingredients needed to write
the \ensuremath{\Varid{read2AndDiscard}} example as in \cref{sec:introduction}.

\subsection{Minimal Examples}
\label{sec:examples}

To get a sense of how the features that we introduce should behave, we now
look at some simple examples.\jp{This seems either way too late or redundant. How can anyone understand the previous section without already understanding this? I suggest either to 1. move this into the Linear Haskell section (using regular arguments). or 2. delete.} Using constraints to represent
limited resources allows the typechecker to reject certain classes
of ill-behaved programs. Accordingly, the following examples show the
different reasons a program might be rejected.

In what follows, we will be using a constraint \ensuremath{\constraintfont{\Conid{C}}} that is consumed by the \ensuremath{\Varid{useC}}
function.
\begin{hscode}\SaveRestoreHook
\column{B}{@{}>{\hspre}l<{\hspost}@{}}%
\column{E}{@{}>{\hspre}l<{\hspost}@{}}%
\>[B]{}\Varid{useC}\mathbin{::}\constraintfont{\Conid{C}}\Lolly \Conid{Int}{}\<[E]%
\ColumnHook
\end{hscode}\resethooks
The type of \ensuremath{\Varid{useC}} indicates that it consumes the linear resource \ensuremath{\Conid{C}} exactly once.

\subsubsection{Dithering}

We reject this program:
\begin{hscode}\SaveRestoreHook
\column{B}{@{}>{\hspre}l<{\hspost}@{}}%
\column{E}{@{}>{\hspre}l<{\hspost}@{}}%
\>[B]{}\Varid{dithering}\mathbin{::}\constraintfont{\Conid{C}}\Lolly \Conid{Bool}\to \Conid{Int}{}\<[E]%
\\
\>[B]{}\Varid{dithering}\;\Varid{x}\mathrel{=}\mathbf{if}\;\Varid{x}\;\mathbf{then}\;\Varid{useC}\;\mathbf{else}\;\mathrm{10}{}\<[E]%
\ColumnHook
\end{hscode}\resethooks
The problem with \ensuremath{\Varid{dithering}} is that it does not unconditionally consume \ensuremath{\constraintfont{\Conid{C}}}:
 the branch where \ensuremath{\Varid{x}\equiv \Conid{True}} uses the resource \ensuremath{\Conid{C}},
whereas the other branch does not.

\subsubsection{Neglecting}

Now consider the type of the linear version of \ensuremath{\Varid{const}}:
\begin{hscode}\SaveRestoreHook
\column{B}{@{}>{\hspre}l<{\hspost}@{}}%
\column{E}{@{}>{\hspre}l<{\hspost}@{}}%
\>[B]{}\Varid{const}\mathbin{::}\Varid{a}⊸\Varid{b}\to \Varid{a}{}\<[E]%
\ColumnHook
\end{hscode}\resethooks
This function uses its first argument linearly, and ignores the second. Thus,
the the second arrow is unrestricted.
One way to improperly use the linear \ensuremath{\Varid{const}} is by neglecting a linear variable:
\begin{hscode}\SaveRestoreHook
\column{B}{@{}>{\hspre}l<{\hspost}@{}}%
\column{E}{@{}>{\hspre}l<{\hspost}@{}}%
\>[B]{}\Varid{neglecting}\mathbin{::}\constraintfont{\Conid{C}}\Lolly \Conid{Int}{}\<[E]%
\\
\>[B]{}\Varid{neglecting}\mathrel{=}\Varid{const}\;\mathrm{10}\;\Varid{useC}{}\<[E]%
\ColumnHook
\end{hscode}\resethooks
The problem with \ensuremath{\Varid{neglecting}} is that, although \ensuremath{\Varid{useC}} is mentioned in this program,
it is never consumed: \ensuremath{\Varid{const}} does not use its second argument.
The constraint \ensuremath{\constraintfont{\Conid{C}}} is not consumed exactly once, and
thus this program is rejected.
The rule is that a linear constraint can only be consumed (linearly)
in a linear context. For example,
\begin{hscode}\SaveRestoreHook
\column{B}{@{}>{\hspre}l<{\hspost}@{}}%
\column{E}{@{}>{\hspre}l<{\hspost}@{}}%
\>[B]{}\Varid{notNeglecting}\mathbin{::}\constraintfont{\Conid{C}}\Lolly \Conid{Int}{}\<[E]%
\\
\>[B]{}\Varid{notNeglecting}\mathrel{=}\Varid{const}\;\Varid{useC}\;\mathrm{10}{}\<[E]%
\ColumnHook
\end{hscode}\resethooks
is accepted, because the \ensuremath{\constraintfont{\Conid{C}}} constraint is passed on to \ensuremath{\Varid{useC}} which itself
appears as an argument to a linear function (whose result is itself consumed linearly).

\subsubsection{Overusing}
\label{sec:overusing}

Finally, the following program is rejected because it uses \ensuremath{\Conid{C}} twice:
\begin{hscode}\SaveRestoreHook
\column{B}{@{}>{\hspre}l<{\hspost}@{}}%
\column{E}{@{}>{\hspre}l<{\hspost}@{}}%
\>[B]{}\Varid{overusing}\mathbin{::}\constraintfont{\Conid{C}}\Lolly (\Conid{Int},\Conid{Int}){}\<[E]%
\\
\>[B]{}\Varid{overusing}\mathrel{=}(\Varid{useC},\Varid{useC}){}\<[E]%
\ColumnHook
\end{hscode}\resethooks

\subsection{Restricting to a linear context with \ensuremath{\constraintfont{\constraintfont{\Conid{Linearly}}}}}
\label{sec:Unique-constraint}
A linear function makes a promise about how it is going to use its argument, but
linearity imposes no restrictions on how a function -- or its result -- is going to be used.
The caller may use the linear function's result unrestrictedly.  This poses a
challenge for providing a type-safe interface for libraries that rely on having a unique pointer to
some resource, such as safe mutable arrays, because the obvious definition of a
constructor function can immediately be misused, violating the assumption of uniqueness.
\begin{hscode}\SaveRestoreHook
\column{B}{@{}>{\hspre}l<{\hspost}@{}}%
\column{E}{@{}>{\hspre}l<{\hspost}@{}}%
\>[B]{}\Varid{new}\mathbin{::}\Conid{Int}\to \Conid{MArray}\;\Varid{a}{}\<[E]%
\\[\blanklineskip]%
\>[B]{}\Varid{bad}\mathrel{=}\mathbf{let}\;\Varid{arr}\mathrel{=}\Varid{new}\;\mathrm{5}\;\mathbf{in}\;(\Varid{arr},\Varid{arr}){}\<[E]%
\\
\>[B]{}\Varid{badToo}\mathrel{=}\Conid{Ur}\;(\Varid{new}\;\mathrm{5}){}\<[E]%
\ColumnHook
\end{hscode}\resethooks
However, with linear constraints, we can overcome this problem by putting the
special \ensuremath{\constraintfont{\constraintfont{\Conid{Linearly}}}} constraint on \ensuremath{\Varid{new}}:
\begin{hscode}\SaveRestoreHook
\column{B}{@{}>{\hspre}l<{\hspost}@{}}%
\column{E}{@{}>{\hspre}l<{\hspost}@{}}%
\>[B]{}\Varid{new}\mathbin{::}\constraintfont{\constraintfont{\Conid{Linearly}}}\Lolly \Conid{Int}\to \Conid{MArray}\;\Varid{a}{}\<[E]%
\ColumnHook
\end{hscode}\resethooks
Suppose we have assumed the \ensuremath{\constraintfont{\constraintfont{\Conid{Linearly}}}} constraint linearly; that is,
we must use the \ensuremath{\constraintfont{\constraintfont{\Conid{Linearly}}}} assumption precisely once. Now, our definition
for \ensuremath{\Varid{bad}} is rejected: either we infer \ensuremath{\Varid{arr}} to have multiplicity~$  \multiplicityfont{ \omega }  $,
in which case its definition uses \ensuremath{\constraintfont{\constraintfont{\Conid{Linearly}}}} $  \multiplicityfont{ \omega }  $ times; or
we infer \ensuremath{\Varid{arr}} to have multiplicity~$  \multiplicityfont{ \ottsym{1} }  $, in which case its use (twice) violates
the linearity restriction. Likewise, the use of \ensuremath{\Conid{Ur}} in \ensuremath{\Varid{badToo}} requires using
the \ensuremath{\constraintfont{\constraintfont{\Conid{Linearly}}}} assumption $  \multiplicityfont{ \omega }  $ times.

This is promising so far, but several problems remain:

\begin{description}
\item[Duplicating \ensuremath{\constraintfont{\constraintfont{\Conid{Linearly}}}}] What if we want to create multiple
arrays, each of which having a unique pointer? If \ensuremath{\constraintfont{\constraintfont{\Conid{Linearly}}}} is assumed
linearly, then \ensuremath{\mathbf{let}\;\Varid{arr1}\mathrel{=}\Varid{new}\;\mathrm{5};\Varid{arr2}\mathrel{=}\Varid{new}\;\mathrm{6}} will fail, as it uses our
\ensuremath{\constraintfont{\constraintfont{\Conid{Linearly}}}} assumption twice. We thus stipulate that \ensuremath{\constraintfont{\constraintfont{\Conid{Linearly}}}}
must itself be duplicable: from one assumption of \ensuremath{\constraintfont{\constraintfont{\Conid{Linearly}}}}, we
must be able to satisfy any arbitrary fixed number of demands on that constraint.
By ``arbitrary fixed number'', we mean to say that we can duplicate \ensuremath{\constraintfont{\constraintfont{\Conid{Linearly}}}}
a finite number of times, but we may not use an assumption of \ensuremath{\constraintfont{\constraintfont{\Conid{Linearly}}}} with
multiplicity 1 to satisfy \ensuremath{\constraintfont{\constraintfont{\Conid{Linearly}}}} at multiplicity $  \multiplicityfont{ \omega }  $.

\item[Discarding \ensuremath{\constraintfont{\constraintfont{\Conid{Linearly}}}}] Similarly to allowing
  duplication, we must allow discarding, in case a function allocates
  no arrays at all. Accordingly, we allow a linear assumption of
  \ensuremath{\constraintfont{\constraintfont{\Conid{Linearly}}}} to be accepted even if the constraint is never
  used.

\item[Initial assumption of \ensuremath{\constraintfont{\constraintfont{\Conid{Linearly}}}}] For this approach to work,
we must have an assumption of \ensuremath{\constraintfont{\constraintfont{\Conid{Linearly}}}} of multiplicity 1. We can
achieve this via the following primitive:
\begin{hscode}\SaveRestoreHook
\column{B}{@{}>{\hspre}l<{\hspost}@{}}%
\column{E}{@{}>{\hspre}l<{\hspost}@{}}%
\>[B]{}\Varid{linearly}\mathbin{::}(\constraintfont{\Conid{Linearly}}\Lolly \Conid{Ur}\;\Varid{r})⊸\Conid{Ur}\;\Varid{r}{}\<[E]%
\ColumnHook
\end{hscode}\resethooks
The argument to \ensuremath{\Varid{linearly}} will be a continuation that assumes \ensuremath{\constraintfont{\constraintfont{\Conid{Linearly}}}}
with multiplicity 1. Because \ensuremath{\Varid{linearly}} returns an unrestricted value, no restricted
values from the continuation can escape the scope of the \ensuremath{\constraintfont{\constraintfont{\Conid{Linearly}}}}
assumption. Thus, the continuation has exactly the condition we need: a
linear assumption of \ensuremath{\constraintfont{\constraintfont{\Conid{Linearly}}}}.
\end{description}

The pattern of using a continuation in \ensuremath{\Varid{linearly}} mirrors the use
of that technique by \citet[Fig.~2]{LinearHaskell}. But
\ensuremath{\Varid{linearly}} is, now, the only place where we need a continuation: once
we have our linear \ensuremath{\constraintfont{\constraintfont{\Conid{Linearly}}}} assumption, we can use it to produce
new values that must be unique.

With just these simple ingredients -- a duplicable, discardable constraint
that can be assumed linearly -- we can write \textsc{api}s that require
uniqueness without heavy use of continuations.

\section{Application: memory ownership}
\label{sec:memory-ownership}
Let us now turn back to the more substantial example introduced in
\cref{sec:introduction}: manual memory management.  In functional programming
languages like Haskell, memory deallocation is normally the responsibility of a
garbage collector. However, garbage collection is
not always desirable, either due to its (unpredictable) runtime costs, or because
pointers exist between separately-managed memory spaces
(for example when calling foreign functions~\cite{linear-inline-java}).
In either case, one must then
resort to explicit memory allocation and deallocation. This task is
error prone: one can easily forget a deallocation (causing a memory
leak) or deallocate several times (corrupting data). In this section we show how
to build a 
memory management \textsc{api} as a \emph{library} using linear
constraints. The library is a generalisation of the array library
introduced in \cref{sec:introduction}.

\subsection{Capability constraints}
\label{sec:atomic-references}

Our approach, inspired by Rust, is
to represent \emph{ownership} of a memory location, and more specifically,
whether the reference is mutable or read-only.
We use the linear constraints \ensuremath{\constraintfont{\Conid{Read}\;\Varid{n}}} and \ensuremath{\constraintfont{\Conid{Write}\;\Varid{n}}},
guarding read access and write access to a reference respectively.
Because of linearity, these constraints
must be consumed, so the \textsc{api} can guarantee that memory
is deallocated correctly.
In \ensuremath{\constraintfont{\Conid{Read}\;\Varid{n}}}, \ensuremath{\Varid{n}} is a type variable (of a special kind \ensuremath{\Conid{Location}})
which represents a memory location. Locations mediate
the relationship between references and ownership constraints.

\begin{minipage}{0.5\linewidth}\begin{hscode}\SaveRestoreHook
\column{B}{@{}>{\hspre}l<{\hspost}@{}}%
\column{3}{@{}>{\hspre}l<{\hspost}@{}}%
\column{E}{@{}>{\hspre}l<{\hspost}@{}}%
\>[3]{}\mathbf{class}\;\constraintfont{\Conid{Read}\;(\Varid{n}\mathbin{::}\Conid{Location})}{}\<[E]%
\ColumnHook
\end{hscode}\resethooks
\end{minipage}
\begin{minipage}{0.5\linewidth}\begin{hscode}\SaveRestoreHook
\column{B}{@{}>{\hspre}l<{\hspost}@{}}%
\column{3}{@{}>{\hspre}l<{\hspost}@{}}%
\column{E}{@{}>{\hspre}l<{\hspost}@{}}%
\>[3]{}\mathbf{class}\;\constraintfont{\Conid{Write}\;(\Varid{n}\mathbin{::}\Conid{Location})}{}\<[E]%
\ColumnHook
\end{hscode}\resethooks
\end{minipage}
\\
To ensure referential transparency,
writes can be done only when we are sure that no other part of the program has
read access to the reference.
Therefore, writing also requires the read capability. Thus we
systematically use \ensuremath{\constraintfont{\Conid{RW}\;\Varid{n}}}, pairing both the read and write
capabilities.

With these components in place, we can provide an \textsc{api} for mutable
references.
\begin{hscode}\SaveRestoreHook
\column{B}{@{}>{\hspre}l<{\hspost}@{}}%
\column{3}{@{}>{\hspre}l<{\hspost}@{}}%
\column{E}{@{}>{\hspre}l<{\hspost}@{}}%
\>[3]{}\mathbf{data}\;\Conid{AtomRef}\;(\Varid{a}\mathbin{::}\Conid{Type})\;(\Varid{n}\mathbin{::}\Conid{Location}){}\<[E]%
\ColumnHook
\end{hscode}\resethooks
The type \ensuremath{\Conid{AtomRef}} is the type of references to values of a type \ensuremath{\Varid{a}} at
location \ensuremath{\Varid{n}}. Allocation of a reference can be done using the
following function. 
\begin{hscode}\SaveRestoreHook
\column{B}{@{}>{\hspre}l<{\hspost}@{}}%
\column{3}{@{}>{\hspre}l<{\hspost}@{}}%
\column{E}{@{}>{\hspre}l<{\hspost}@{}}%
\>[3]{}\Varid{newRef}\mathbin{::}\constraintfont{\constraintfont{\Conid{Linearly}}}\Lolly \exists\;\Varid{n}.\Conid{AtomRef}\;\Varid{a}\;\Varid{n}\RLolly\constraintfont{\Conid{RW}\;\Varid{n}}{}\<[E]%
\ColumnHook
\end{hscode}\resethooks
The function \ensuremath{\Varid{newRef}} creates a new atomic reference, initialised with
\ensuremath{\bot }; we could also pass in an initial value, but doing so in the more
general case below would add complication and obscure our main goal of demonstrating
linear constraints.

To read a reference, a \ensuremath{\constraintfont{\Conid{Read}}} constraint is demanded, and
then returned back. Writing is similar.
\begin{hscode}\SaveRestoreHook
\column{B}{@{}>{\hspre}l<{\hspost}@{}}%
\column{3}{@{}>{\hspre}l<{\hspost}@{}}%
\column{E}{@{}>{\hspre}l<{\hspost}@{}}%
\>[3]{}\Varid{readRef}\mathbin{::}\constraintfont{\Conid{Read}\;\Varid{n}}\Lolly \Conid{AtomRef}\;\Varid{a}\;\Varid{n}\to \Conid{Ur}\;\Varid{a}\RLolly\constraintfont{\Conid{Read}\;\Varid{n}}{}\<[E]%
\\
\>[3]{}\Varid{writeRef}\mathbin{::}\constraintfont{\Conid{RW}\;\Varid{n}}\Lolly \Conid{AtomRef}\;\Varid{a}\;\Varid{n}\to \Varid{a}\to ()\RLolly\constraintfont{\Conid{RW}\;\Varid{n}}{}\<[E]%
\ColumnHook
\end{hscode}\resethooks
Note that the above primitives do not need to explicitly declare
effects in terms of a monad or another higher-order effect-tracking
device: because the \ensuremath{\constraintfont{\Conid{RW}\;\Varid{n}}} constraint is linear, passing it suffices
to ensure proper sequencing of effects concerning location \ensuremath{\Varid{n}}.

Also note that \ensuremath{\Varid{readRef}} returns an unrestricted \emph{copy} of the element, and
\ensuremath{\Varid{writeRef}} \emph{copies} an unrestricted element into the location. This means
that while \ensuremath{\Conid{AtomRef}}s are mutable, their contents are always immutable structures.




Since there is a unique \ensuremath{\constraintfont{\Conid{RW}\;\Varid{n}}} constraint per reference, we
can also use it to represent ownership of the reference: access to \ensuremath{\constraintfont{\Conid{RW}\;\Varid{n}}} represents responsibility (and obligation) to deallocate \ensuremath{\Varid{n}}:
\begin{hscode}\SaveRestoreHook
\column{B}{@{}>{\hspre}l<{\hspost}@{}}%
\column{3}{@{}>{\hspre}l<{\hspost}@{}}%
\column{E}{@{}>{\hspre}l<{\hspost}@{}}%
\>[3]{}\Varid{freeRef}\mathbin{::}\constraintfont{\Conid{RW}\;\Varid{n}}\Lolly \Conid{AtomRef}\;\Varid{a}\;\Varid{n}\to (){}\<[E]%
\ColumnHook
\end{hscode}\resethooks


\subsection{Arrays}
\label{sec:arrays}

The above toolkit handles references to base types just fine.  But
what about storing references in objects managed by the ownership
system? In \cref{sec:introduction}, we presented an interface for mutable
arrays whose contents are themselves immutable. Our approach
scales beyond that use case, supporting arrays of
references, including arrays of (mutable) arrays.


\begin{hscode}\SaveRestoreHook
\column{B}{@{}>{\hspre}l<{\hspost}@{}}%
\column{3}{@{}>{\hspre}l<{\hspost}@{}}%
\column{E}{@{}>{\hspre}l<{\hspost}@{}}%
\>[3]{}\mathbf{data}\;\Conid{PArray}\;(\Varid{a}\mathbin{::}\Conid{Location}\to \Conid{Type})\;(\Varid{n}\mathbin{::}\Conid{Location}){}\<[E]%
\\
\>[3]{}\Varid{newPArray}\mathbin{::}\constraintfont{\constraintfont{\Conid{Linearly}}}\Lolly \Conid{Int}\to \exists\;\Varid{n}.\Conid{PArray}\;\Varid{a}\;\Varid{n}\RLolly\constraintfont{\Conid{RW}\;\Varid{n}}{}\<[E]%
\ColumnHook
\end{hscode}\resethooks
For this purpose we introduce the type \ensuremath{\Conid{PArray}\;\Varid{a}\;\Varid{n}}, where the kind of
\ensuremath{\Varid{a}} is \ensuremath{\Conid{Location}\to \Conid{Type}}: this way we can easily enforce that each
reference in the array refers to the same location \ensuremath{\Varid{n}}. Both types
\ensuremath{\Conid{AtomRef}\;\Varid{a}} and \ensuremath{\Conid{PArray}\;\Varid{a}} have kind \ensuremath{\Conid{Location}\to \Conid{Type}}, and therefore
one can allocate, and manipulate arrays of arrays with this
\textsc{api}. For example, an array of integers has type
\ensuremath{\Conid{PArray}\;(\Conid{AtomRef}\;\Conid{Int})\;\Varid{n}}, and indeed, the \ensuremath{\Conid{UArray}} type from
\cref{sec:introduction} is a synonym for an array of atomic references.
An array of arrays of integers would has type \ensuremath{\Conid{PArray}\;(\Conid{PArray}\;(\Conid{AtomRef}\;\Conid{Int}))\;\Varid{n}}. Thus,
the framework handles nested mutable structures without any additional
difficulty. 

The actual runtime value of a \ensuremath{\Conid{PArray}} is a pointer to a contiguous block of
memory together with the size of the memory block. This means that the length of the
array can be accessed without having ownership of the array: \ensuremath{\Varid{length}\mathbin{::}\Conid{PArray}\;\Varid{a}\;\Varid{n}\to \Conid{Int}}.
While the \ensuremath{\Conid{PArray}} reference itself is managed by the garbage collector, the
pointer it contains points to manually managed memory.







\subsubsection{Borrowing}

The \ensuremath{\Varid{lendMut}\;\Varid{arr}\;\Varid{i}\;\Varid{k}} primitive lends access to the reference at index \ensuremath{\Varid{i}} in
\ensuremath{\Varid{arr}}, to a continuation function \ensuremath{\Varid{k}} (in Rust terminology, the function
\emph{borrows} an element of the array). Note that the continuation must return the
read-write capability, so that the ownership transfer is indeed temporary. The
type system guarantees that the borrowed reference cannot be shared or deallocated.
%
Indeed, with this \textsc{api}, \ensuremath{\constraintfont{\Conid{RW}\;\Varid{n}}} and \ensuremath{\constraintfont{\Conid{RW}\;\Varid{p}}} are
never simultaneously available.
\begin{hscode}\SaveRestoreHook
\column{B}{@{}>{\hspre}l<{\hspost}@{}}%
\column{3}{@{}>{\hspre}l<{\hspost}@{}}%
\column{12}{@{}>{\hspre}l<{\hspost}@{}}%
\column{E}{@{}>{\hspre}l<{\hspost}@{}}%
\>[3]{}\Varid{lendMut}{}\<[12]%
\>[12]{}\mathbin{::}\constraintfont{\Conid{RW}\;\Varid{n}}\Lolly \Conid{PArray}\;\Varid{a}\;\Varid{n}\to \Conid{Int}\to (\forall\;\Varid{p}.\,\constraintfont{\Conid{RW}\;\Varid{p}}\Lolly \Varid{a}\;\Varid{p}\to \Varid{r}\RLolly\constraintfont{\Conid{RW}\;\Varid{p}})\mathbin{⊸}\Varid{r}\RLolly\constraintfont{\Conid{RW}\;\Varid{n}}{}\<[E]%
\ColumnHook
\end{hscode}\resethooks
Because the elements of an array can be mutable structures (such as
other arrays), reading can be done safely only if we can ensure that
no one else has access to the array while the element is accessed. Otherwise,
the array -- including the element being read -- could be mutated.
Therefore,
gaining simple read access to an element needs to be done using a
scoped \textsc{api} as well:
\begin{hscode}\SaveRestoreHook
\column{B}{@{}>{\hspre}l<{\hspost}@{}}%
\column{3}{@{}>{\hspre}l<{\hspost}@{}}%
\column{9}{@{}>{\hspre}c<{\hspost}@{}}%
\column{9E}{@{}l@{}}%
\column{13}{@{}>{\hspre}l<{\hspost}@{}}%
\column{59}{@{}>{\hspre}l<{\hspost}@{}}%
\column{E}{@{}>{\hspre}l<{\hspost}@{}}%
\>[3]{}\Varid{lend}{}\<[9]%
\>[9]{}\mathbin{::}{}\<[9E]%
\>[13]{}\constraintfont{\Conid{Read}\;\Varid{n}}\Lolly \Conid{PArray}\;\Varid{a}\;\Varid{n}\to \Conid{Int}\to {}\<[59]%
\>[59]{}(\forall\;\Varid{p}.\,\constraintfont{\Conid{Read}\;\Varid{p}}\Lolly \Varid{a}\;\Varid{p}\to \Varid{r}\RLolly\constraintfont{\Conid{Read}\;\Varid{p}})\mathbin{⊸}\Varid{r}\RLolly\constraintfont{\Conid{Read}\;\Varid{n}}{}\<[E]%
\ColumnHook
\end{hscode}\resethooks
For the special case of \ensuremath{\Conid{UArray}}s, a more traditional reading
operation can be implemented, by lending the reference to \ensuremath{\Varid{readRef}}
which creates an unrestricted \emph{copy} of the value. This copy is
under control of the garbage collector, and can escape the scope of
the borrowing freely.
\begin{hscode}\SaveRestoreHook
\column{B}{@{}>{\hspre}l<{\hspost}@{}}%
\column{3}{@{}>{\hspre}l<{\hspost}@{}}%
\column{E}{@{}>{\hspre}l<{\hspost}@{}}%
\>[3]{}\Varid{read}\mathbin{::}\constraintfont{\Conid{Read}\;\Varid{n}}\Lolly \Conid{UArray}\;\Varid{a}\;\Varid{n}\to \Conid{Int}\to \Conid{Ur}\;\Varid{a}\RLolly\constraintfont{\Conid{Read}\;\Varid{n}}{}\<[E]%
\\
\>[3]{}\Varid{read}\;\Varid{arr}\;\Varid{i}\mathrel{=}\Varid{lend}\;\Varid{arr}\;\Varid{i}\;\Varid{readRef}{}\<[E]%
\ColumnHook
\end{hscode}\resethooks







\subsubsection{Slices}

It is also possible to give a safe interface to array
\emph{slices}. A slice represents a part of an array and allows
splitting the ownership of the array into multiple parts,
shared between different consumers.  The ownership system
means that slicing does not require copying.

Splitting consumes all capabilities of an array and returns two new
arrays that represent the contiguous blocks of memory before and
starting at a given index.
\begin{hscode}\SaveRestoreHook
\column{B}{@{}>{\hspre}l<{\hspost}@{}}%
\column{3}{@{}>{\hspre}l<{\hspost}@{}}%
\column{E}{@{}>{\hspre}l<{\hspost}@{}}%
\>[3]{}\Varid{split}\mathbin{::}\constraintfont{\Conid{RW}\;\Varid{n}}\Lolly \Conid{PArray}\;\Varid{a}\;\Varid{n}\to \Conid{Int}\to \exists\;\Varid{l}\;\Varid{r}.\Conid{Ur}\;(\Conid{PArray}\;\Varid{a}\;\Varid{l},\Conid{PArray}\;\Varid{a}\;\Varid{r})\RLolly\constraintfont{(\Conid{RW}\;\Varid{l},\Conid{RW}\;\Varid{r},\Conid{Slices}\;\Varid{n}\;\Varid{l}\;\Varid{r})}{}\<[E]%
\ColumnHook
\end{hscode}\resethooks
In addition to the array capabilities, the output constraints also include
\ensuremath{\constraintfont{\Conid{Slices}\;\Varid{n}\;\Varid{l}\;\Varid{r}}}, witnessing the fact that locations
\ensuremath{\Varid{l}} and \ensuremath{\Varid{r}} are components of \ensuremath{\Varid{n}}, so that they can be joined back
together:
\begin{hscode}\SaveRestoreHook
\column{B}{@{}>{\hspre}l<{\hspost}@{}}%
\column{3}{@{}>{\hspre}l<{\hspost}@{}}%
\column{E}{@{}>{\hspre}l<{\hspost}@{}}%
\>[3]{}\Varid{join}\mathbin{::}\constraintfont{(\Conid{Slices}\;\Varid{n}\;\Varid{l}\;\Varid{r},\Conid{RW}\;\Varid{r},\Conid{RW}\;\Varid{l})}\Lolly \Conid{PArray}\;\Varid{a}\;\Varid{l}\to \Conid{PArray}\;\Varid{a}\;\Varid{r}\to \Conid{Ur}\;(\Conid{PArray}\;\Varid{a}\;\Varid{n})\RLolly\constraintfont{\Conid{RW}\;\Varid{n}}{}\<[E]%
\ColumnHook
\end{hscode}\resethooks
\begin{figure}
\maybesmall
\begin{hscode}\SaveRestoreHook
\column{B}{@{}>{\hspre}l<{\hspost}@{}}%
\column{15}{@{}>{\hspre}l<{\hspost}@{}}%
\column{26}{@{}>{\hspre}c<{\hspost}@{}}%
\column{26E}{@{}l@{}}%
\column{29}{@{}>{\hspre}l<{\hspost}@{}}%
\column{35}{@{}>{\hspre}l<{\hspost}@{}}%
\column{53}{@{}>{\hspre}c<{\hspost}@{}}%
\column{53E}{@{}l@{}}%
\column{56}{@{}>{\hspre}l<{\hspost}@{}}%
\column{58}{@{}>{\hspre}l<{\hspost}@{}}%
\column{60}{@{}>{\hspre}l<{\hspost}@{}}%
\column{65}{@{}>{\hspre}l<{\hspost}@{}}%
\column{84}{@{}>{\hspre}c<{\hspost}@{}}%
\column{84E}{@{}l@{}}%
\column{87}{@{}>{\hspre}l<{\hspost}@{}}%
\column{E}{@{}>{\hspre}l<{\hspost}@{}}%
\>[B]{}\Varid{swap}\mathbin{::}\constraintfont{\Conid{RW}\;\Varid{n}}\Lolly \Conid{PArray}\;\Conid{AtomRef}\;\Varid{n}\to \Conid{Int}\to \Conid{Int}\to ()\RLolly\constraintfont{\Conid{RW}\;\Varid{n}}{}\<[E]%
\\
\>[B]{}\Varid{swap}\;\Varid{arr}\;\Varid{i}\;\Varid{j}{}\<[15]%
\>[15]{}\mid \Varid{i}\equiv \Varid{j}{}\<[26]%
\>[26]{}\mathrel{=}{}\<[26E]%
\>[29]{}\packbox(){}\<[E]%
\\
\>[15]{}\mid \Varid{i}\mathbin{>}\Varid{j}{}\<[26]%
\>[26]{}\mathrel{=}{}\<[26E]%
\>[29]{}\Varid{swap}\;\Varid{arr}\;\Varid{j}\;\Varid{i}{}\<[E]%
\\
\>[15]{}\mid \Varid{i}\mathbin{<}\Varid{j}{}\<[26]%
\>[26]{}\mathrel{=}{}\<[26E]%
\>[29]{}\mathbf{let}\;{}\<[35]%
\>[35]{}\packbox(\Conid{Ur}\;(\Varid{l},\Varid{r})){}\<[53]%
\>[53]{}\mathrel{=}{}\<[53E]%
\>[56]{}\Varid{split}\;\Varid{arr}\;(\Varid{i}\mathbin{+}\mathrm{1}){}\<[E]%
\\
\>[35]{}\packbox(){}\<[53]%
\>[53]{}\mathrel{=}{}\<[53E]%
\>[56]{}\Varid{lendMut}\;\Varid{l}\;\Varid{i}\;(\lambda \Varid a \ensuremath{_i}\to {}\<[E]%
\\
\>[56]{}\hsindent{2}{}\<[58]%
\>[58]{}\mathbf{let}\;\packbox()\mathrel{=}\Varid{lendMut}\;\Varid{r}\;(\Varid{j}\mathbin{-}(\Varid{i}\mathbin{+}\mathrm{1}))\;(\lambda \Varid a \ensuremath{_j}\to {}\<[E]%
\\
\>[58]{}\hsindent{2}{}\<[60]%
\>[60]{}\mathbf{let}\;{}\<[65]%
\>[65]{}\packbox(\Conid{Ur}\;\Varid a \ensuremath{_i}\_val){}\<[84]%
\>[84]{}\mathrel{=}{}\<[84E]%
\>[87]{}\Varid{readRef}\;\Varid a \ensuremath{_i}{}\<[E]%
\\
\>[65]{}\packbox(\Conid{Ur}\;\Varid a \ensuremath{_j}\_val){}\<[84]%
\>[84]{}\mathrel{=}{}\<[84E]%
\>[87]{}\Varid{readRef}\;\Varid a \ensuremath{_j}{}\<[E]%
\\
\>[65]{}\packbox(){}\<[84]%
\>[84]{}\mathrel{=}{}\<[84E]%
\>[87]{}\Varid{writeRef}\;\Varid a \ensuremath{_j}\;\Varid a \ensuremath{_i}\_val{}\<[E]%
\\
\>[65]{}\packbox(){}\<[84]%
\>[84]{}\mathrel{=}{}\<[84E]%
\>[87]{}\Varid{writeRef}\;\Varid a \ensuremath{_i}\;\Varid a \ensuremath{_j}\_val{}\<[E]%
\\
\>[58]{}\hsindent{2}{}\<[60]%
\>[60]{}\mathbf{in}\;\packbox())\;\mathbf{in}\;\packbox()){}\<[E]%
\\
\>[35]{}\packbox(\Conid{Ur}\;\anonymous ){}\<[53]%
\>[53]{}\mathrel{=}{}\<[53E]%
\>[56]{}\Varid{join}\;\Varid{l}\;\Varid{r}{}\<[E]%
\\
\>[29]{}\mathbf{in}\;\packbox(){}\<[E]%
\ColumnHook
\end{hscode}\resethooks
\caption{Swapping two elements of an array}
\label{fig:swap}
\end{figure}

With these building blocks, we can now implement various utility functions on
arrays, such as swapping two elements of an array, which is shown in
\Cref{fig:swap}.  It is
not so simple to implement\footnote{Indeed, Rust's implementation uses
an \emph{unsafe} block.}, because we need two elements of an array
simultaneously, but only one element can be borrowed at a time. To
solve this problem, we split the array into two slices such that the
two indices fall in two different slices. Then simply borrow the
element \ensuremath{\Varid{i}} from the first slice, and \ensuremath{\Varid{j}} from the second slice (using
\ensuremath{\Varid{lendMut}}). Finally, we join the two slices back together.

\subsubsection{In-place quicksort}

As an example of using the machinery defined above, we implement an in-place, pure
quicksort algorithm, given in \Cref{fig:quicksort}.
\begin{figure}
\maybesmall
\noindent
\begin{minipage}[t]{0.45\linewidth}
\begin{hscode}\SaveRestoreHook
\column{B}{@{}>{\hspre}l<{\hspost}@{}}%
\column{3}{@{}>{\hspre}l<{\hspost}@{}}%
\column{10}{@{}>{\hspre}l<{\hspost}@{}}%
\column{15}{@{}>{\hspre}l<{\hspost}@{}}%
\column{35}{@{}>{\hspre}c<{\hspost}@{}}%
\column{35E}{@{}l@{}}%
\column{38}{@{}>{\hspre}l<{\hspost}@{}}%
\column{E}{@{}>{\hspre}l<{\hspost}@{}}%
\>[B]{}\Varid{sort}\mathbin{::}\constraintfont{\Conid{RW}\;\Varid{n}}\Lolly \Conid{UArray}\;\Conid{Int}\;\Varid{n}\to ()\RLolly\constraintfont{\Conid{RW}\;\Varid{n}}{}\<[E]%
\\
\>[B]{}\Varid{sort}\;\Varid{arr}\mathrel{=}\mathbf{let}\;\Varid{len}\mathrel{=}\Varid{length}\;\Varid{arr}\;\mathbf{in}{}\<[E]%
\\
\>[B]{}\hsindent{3}{}\<[3]%
\>[3]{}\mathbf{if}\;\Varid{len}\leq \mathrm{1}\;\mathbf{then}\;\packbox(){}\<[E]%
\\
\>[B]{}\hsindent{3}{}\<[3]%
\>[3]{}\mathbf{else}\;{}\<[10]%
\>[10]{}\mathbf{let}\;{}\<[15]%
\>[15]{}\packbox\Varid{pivotIdx}{}\<[35]%
\>[35]{}\mathrel{=}{}\<[35E]%
\>[38]{}\Varid{partition}\;\Varid{arr}{}\<[E]%
\\
\>[15]{}\packbox(\Conid{Ur}\;(\Varid{l},\Varid{r})){}\<[35]%
\>[35]{}\mathrel{=}{}\<[35E]%
\>[38]{}\Varid{split}\;\Varid{arr}\;\Varid{pivotIdx}{}\<[E]%
\\
\>[15]{}\packbox(){}\<[35]%
\>[35]{}\mathrel{=}{}\<[35E]%
\>[38]{}\Varid{sort}\;\Varid{l}{}\<[E]%
\\
\>[15]{}\packbox(){}\<[35]%
\>[35]{}\mathrel{=}{}\<[35E]%
\>[38]{}\Varid{sort}\;\Varid{r}{}\<[E]%
\\
\>[15]{}\packbox(\Conid{Ur}\;\anonymous ){}\<[35]%
\>[35]{}\mathrel{=}{}\<[35E]%
\>[38]{}\Varid{join}\;\Varid{l}\;\Varid{r}{}\<[E]%
\\
\>[10]{}\mathbf{in}\;\packbox(){}\<[E]%
\ColumnHook
\end{hscode}\resethooks
\end{minipage}
\begin{minipage}[t]{0.4\linewidth}
\begin{hscode}\SaveRestoreHook
\column{B}{@{}>{\hspre}l<{\hspost}@{}}%
\column{3}{@{}>{\hspre}l<{\hspost}@{}}%
\column{8}{@{}>{\hspre}l<{\hspost}@{}}%
\column{10}{@{}>{\hspre}c<{\hspost}@{}}%
\column{10E}{@{}l@{}}%
\column{13}{@{}>{\hspre}l<{\hspost}@{}}%
\column{19}{@{}>{\hspre}l<{\hspost}@{}}%
\column{25}{@{}>{\hspre}c<{\hspost}@{}}%
\column{25E}{@{}l@{}}%
\column{28}{@{}>{\hspre}l<{\hspost}@{}}%
\column{E}{@{}>{\hspre}l<{\hspost}@{}}%
\>[B]{}\Varid{partition}\mathbin{::}\constraintfont{\Conid{RW}\;\Varid{n}}\Lolly \Conid{UArray}\;\Conid{Int}\;\Varid{n}\to \Conid{Int}\RLolly\constraintfont{\Conid{RW}\;\Varid{n}}{}\<[E]%
\\
\>[B]{}\Varid{partition}\;\Varid{arr}\mathrel{=}{}\<[E]%
\\
\>[B]{}\hsindent{3}{}\<[3]%
\>[3]{}\mathbf{let}\;{}\<[8]%
\>[8]{}\Varid{last}{}\<[25]%
\>[25]{}\mathrel{=}{}\<[25E]%
\>[28]{}\Varid{length}\;\Varid{arr}\mathbin{-}\mathrm{1}{}\<[E]%
\\
\>[8]{}\packbox(\Conid{Ur}\;\Varid{pivot}){}\<[25]%
\>[25]{}\mathrel{=}{}\<[25E]%
\>[28]{}\Varid{read}\;\Varid{arr}\;\Varid{last}{}\<[E]%
\\
\>[8]{}\Varid{go}\mathbin{::}\constraintfont{\Conid{RW}\;\Varid{n}}\Lolly \Conid{Int}\to \Conid{Int}\to \Conid{Int}\RLolly\constraintfont{\Conid{RW}\;\Varid{n}}{}\<[E]%
\\
\>[8]{}\Varid{go}\;\Varid{l}\;\Varid{r}{}\<[E]%
\\
\>[8]{}\hsindent{2}{}\<[10]%
\>[10]{}\mid {}\<[10E]%
\>[13]{}\Varid{l}\mathbin{>}\Varid{r}{}\<[E]%
\\
\>[8]{}\hsindent{2}{}\<[10]%
\>[10]{}\mathrel{=}{}\<[10E]%
\>[13]{}\mathbf{let}\;\packbox()\mathrel{=}\Varid{swap}\;\Varid{arr}\;\Varid{last}\;\Varid{l}\;\mathbf{in}\;\packbox\Varid{l}{}\<[E]%
\\
\>[8]{}\hsindent{2}{}\<[10]%
\>[10]{}\mid {}\<[10E]%
\>[13]{}\Varid{otherwise}{}\<[E]%
\\
\>[8]{}\hsindent{2}{}\<[10]%
\>[10]{}\mathrel{=}{}\<[10E]%
\>[13]{}\mathbf{let}\;\packbox(\Conid{Ur}\;\Varid{lVal})\mathrel{=}\Varid{read}\;\Varid{arr}\;\Varid{l}\;\mathbf{in}{}\<[E]%
\\
\>[13]{}\mathbf{if}\;\Varid{lVal}\mathbin{>}\Varid{pivot}{}\<[E]%
\\
\>[13]{}\mathbf{then}\;{}\<[19]%
\>[19]{}\mathbf{let}\;\packbox()\mathrel{=}\Varid{swap}\;\Varid{arr}\;\Varid{l}\;\Varid{r}{}\<[E]%
\\
\>[19]{}\mathbf{in}\;\Varid{go}\;\Varid{l}\;(\Varid{r}\mathbin{-}\mathrm{1}){}\<[E]%
\\
\>[13]{}\mathbf{else}\;{}\<[19]%
\>[19]{}\Varid{go}\;(\Varid{l}\mathbin{+}\mathrm{1})\;\Varid{r}{}\<[E]%
\\
\>[B]{}\hsindent{3}{}\<[3]%
\>[3]{}\mathbf{in}\;\Varid{go}\;\mathrm{0}\;(\Varid{last}\mathbin{-}\mathrm{1}){}\<[E]%
\ColumnHook
\end{hscode}\resethooks
\end{minipage}
\caption{In-place quicksort}
\label{fig:quicksort}
\end{figure}
The \ensuremath{\Varid{partition}} function is responsible for picking a pivot element and
reorganising the array elements such that each element preceding the pivot will
be less than or equal to it, and the elements after will be greater than the
pivot. Once finished, it returns the index of the pivot element; \ensuremath{\Varid{sort}} then
splits the array at the pivot element and recursively operates on the two
slices.

\section{A qualified type system for linear constraints}
\label{sec:qualified-type-system}

We
now present our design for a qualified type system~\cite{QualifiedTypes} that supports
linear constraints. Our design, based on the work of \citet{OutsideIn}, is compatible with Haskell and \textsc{ghc}.

\subsection{Simple Constraints and Entailment}
\label{sec:constraint-domain}

We call constraints such as
\ensuremath{\constraintfont{\Conid{Read}\;\Varid{n}}} or \ensuremath{\constraintfont{\Conid{Write}\;\Varid{n}}} \emph{atomic
  constraints}. The set of atomic constraints is a parameter of our
qualified type system.

\begin{definition}[Atomic constraints]
  The qualified type system is parameterised by a set, whose elements
  are called \emph{atomic constraints}. We use the variable $  \constraintfont{ \ottnt{q} }  $
  to denote atomic constraints.
\end{definition}

Atomic constraints are assembled into \emph{simple constraints}
$  \constraintfont{ \ottnt{Q} }  $, which play the hybrid role of constraint contexts and
(linear) logic formulae.
The following operations work with simple constraints:
\begin{description}
\item[Scaled atomic constraints] $  \constraintfont{   \multiplicityfont{ \pi }  \scale \ottnt{q}  }  $ is a simple constraint,
  where $  \multiplicityfont{ \pi }  $ specifies whether $  \constraintfont{ \ottnt{q} }  $ is to be used linearly
  or not.
\item[Conjunction] Two simple constraints can be paired up
  $  \constraintfont{ \ottnt{Q_{{\mathrm{1}}}}  \qtensor  \ottnt{Q_{{\mathrm{2}}}} }  $. Semantically, this corresponds to the multiplicative
  conjunction of linear logic. Tensor products represent pairs of
  constraints such as \ensuremath{\constraintfont{(\Conid{Read}\;\Varid{n},\Conid{Write}\;\Varid{n})}} from Haskell.
\item[Empty conjunction] Finally we need a neutral element
  $  \constraintfont{  \mathbf{\varepsilon}  }  $ to the tensor product. The empty conjunction is used
  to represent functions which don't require any constraints.
\end{description}
However, we do not define $  \constraintfont{ \ottnt{Q} }  $ inductively, because we
require certain equalities to hold:

\smallskip
\begin{minipage}[c]{0.5\linewidth}
$$
\begin{array}{rcl}
    \constraintfont{ \ottnt{Q_{{\mathrm{1}}}}  \qtensor  \ottnt{Q_{{\mathrm{2}}}} }   & = &   \constraintfont{ \ottnt{Q_{{\mathrm{2}}}}  \qtensor  \ottnt{Q_{{\mathrm{1}}}} }   \\
    \constraintfont{ \ottsym{(}  \ottnt{Q_{{\mathrm{1}}}}  \qtensor  \ottnt{Q_{{\mathrm{2}}}}  \ottsym{)}  \qtensor  \ottnt{Q_{{\mathrm{3}}}} }   & = &   \constraintfont{ \ottnt{Q_{{\mathrm{1}}}}  \qtensor  \ottsym{(}  \ottnt{Q_{{\mathrm{2}}}}  \qtensor  \ottnt{Q_{{\mathrm{3}}}}  \ottsym{)} }  
\end{array}
$$
\end{minipage}
\begin{minipage}[c]{0.4\linewidth}
$$
\begin{array}{rcl}
    \constraintfont{   \multiplicityfont{ \omega }  \scale \ottnt{q}   \qtensor    \multiplicityfont{ \omega }  \scale \ottnt{q}  }   & = &   \constraintfont{   \multiplicityfont{ \omega }  \scale \ottnt{q}  }   \\
    \constraintfont{ \ottnt{Q}  \qtensor   \mathbf{\varepsilon}  }   & = &   \constraintfont{ \ottnt{Q} }  
\end{array}
$$
\end{minipage}
\smallskip

We thus say that a simple constraint is a pair combining a set of unrestricted
constraints $  \constraintfont{ \ottnt{U} }  $ and a multiset of linear constraints $  \constraintfont{ \ottnt{L} }  $.
The linear constraints must
be stored in a multiset, because assuming the same constraint twice is distinct
from assuming it only once.

\begin{definition}[Simple constraints]
$$
\begin{array}{rcll}
    \constraintfont{ \ottnt{U} }   & \bnfeq & \ldots & \text{set of atomic constraints $  \constraintfont{ \ottnt{q} }  $} \\
    \constraintfont{ \ottnt{L} }   & \bnfeq & \ldots & \text{multiset of atomic constraints $  \constraintfont{ \ottnt{q} }  $} \\
    \constraintfont{ \ottnt{Q} }   & \bnfeq &   \constraintfont{ \ottsym{(}  \ottnt{U}  \ottsym{,}  \ottnt{L}  \ottsym{)} }   & \text{simple constraints}
\end{array}
$$
We can now straightforwardly define the operations we need on simple
constraints:
$$
\begin{array}{ccc}
  \begin{array}{r@{\;}c@{\;}l}
      \constraintfont{  \mathbf{\varepsilon}  }   &=&   \constraintfont{ \ottsym{(}  \emptyset  \ottsym{,}  \emptyset  \ottsym{)} }  
  \end{array}
  &
    \left\{
    \begin{array}{r@{\;}c@{\;}l}
        \constraintfont{   \multiplicityfont{ \ottsym{1} }  \scale \ottnt{q}  }   &=&   \constraintfont{ \ottsym{(}  \emptyset  \ottsym{,}  \ottnt{q}  \ottsym{)} }   \\
        \constraintfont{   \multiplicityfont{ \omega }  \scale \ottnt{q}  }   &=&   \constraintfont{ \ottsym{(}  \ottnt{q}  \ottsym{,}  \emptyset  \ottsym{)} }  
    \end{array}
                        \right.
  &
    \begin{array}{r@{\;}c@{\;}l}
        \constraintfont{ \ottsym{(}  \ottnt{U}_{{\mathrm{1}}}  \ottsym{,}  \ottnt{L}_{{\mathrm{1}}}  \ottsym{)}  \qtensor  \ottsym{(}  \ottnt{U}_{{\mathrm{2}}}  \ottsym{,}  \ottnt{L}_{{\mathrm{2}}}  \ottsym{)} }   &=&   \constraintfont{ \ottsym{(}   \ottnt{U}_{{\mathrm{1}}}  \cup  \ottnt{U}_{{\mathrm{2}}}   \ottsym{,}   \ottnt{L}_{{\mathrm{1}}}  \uplus  \ottnt{L}_{{\mathrm{2}}}   \ottsym{)} }  
    \end{array}
\end{array}
$$
\end{definition}

In practice, we do not need to concern ourselves with the
concrete representation of $  \constraintfont{ \ottnt{Q} }  $ as a pair of sets, instead
using the operations defined just above.

The semantics of simple constraints (and, indeed, of atomic
constraints) is given by an \emph{entailment relation}. Just like the
set of atomic constraints, the entailment relation is a parameter of
our system.

\begin{figure}
  \maybesmall
  \begin{enumerate}
  \item $ \constraintfont{ \ottnt{Q} }   \Vdash   \constraintfont{ \ottnt{Q} } $.
  \item If $ \constraintfont{ \ottnt{Q_{{\mathrm{1}}}} }   \Vdash   \constraintfont{ \ottnt{Q_{{\mathrm{2}}}} } $  and $ \constraintfont{ \ottnt{Q}  \qtensor  \ottnt{Q_{{\mathrm{2}}}} }   \Vdash   \constraintfont{ \ottnt{Q_{{\mathrm{3}}}} } $, then $ \constraintfont{ \ottnt{Q}  \qtensor  \ottnt{Q_{{\mathrm{1}}}} }   \Vdash   \constraintfont{ \ottnt{Q_{{\mathrm{3}}}} } $.
  \item If $ \constraintfont{ \ottnt{Q} }   \Vdash   \constraintfont{ \ottnt{Q_{{\mathrm{1}}}}  \qtensor  \ottnt{Q_{{\mathrm{2}}}} } $, then there exist $  \constraintfont{ \ottnt{Q'} }  $, $  \constraintfont{ Q_{\mathcal{D} } }  $, and $  \constraintfont{ \ottnt{Q''} }  $
        such that:
        $$
           \constraintfont{ Q_{\mathcal{D} } }   \in  \constraintfont{\mathcal{D} } , \quad   \constraintfont{ \ottnt{Q} }  =  \constraintfont{ \ottnt{Q'}  \qtensor  Q_{\mathcal{D} }  \qtensor  \ottnt{Q''} }  , \quad  \constraintfont{ \ottnt{Q'}  \qtensor  Q_{\mathcal{D} } }   \Vdash   \constraintfont{ \ottnt{Q_{{\mathrm{1}}}} } \text{, and } \quad
         \constraintfont{ Q_{\mathcal{D} }  \qtensor  \ottnt{Q''} }   \Vdash   \constraintfont{ \ottnt{Q_{{\mathrm{2}}}} } .
        $$
  \item If $ \constraintfont{ \ottnt{Q} }   \Vdash   \constraintfont{  \mathbf{\varepsilon}  } $, then $  \constraintfont{ \ottnt{Q} }   \in  \constraintfont{\mathcal{D} } $.
  \item If $ \constraintfont{ \ottnt{Q_{{\mathrm{1}}}} }   \Vdash   \constraintfont{ \ottnt{Q'_{{\mathrm{1}}}} } $ and $ \constraintfont{ \ottnt{Q_{{\mathrm{2}}}} }   \Vdash   \constraintfont{ \ottnt{Q'_{{\mathrm{2}}}} } $, then $ \constraintfont{ \ottnt{Q_{{\mathrm{1}}}}  \qtensor  \ottnt{Q_{{\mathrm{2}}}} }   \Vdash   \constraintfont{ \ottnt{Q'_{{\mathrm{1}}}}  \qtensor  \ottnt{Q'_{{\mathrm{2}}}} } $.
  \item If $ \constraintfont{ \ottnt{Q} }   \Vdash   \constraintfont{   \multiplicityfont{ \rho }  \scale \ottnt{q}  } $, then $ \constraintfont{   \multiplicityfont{ \pi }  \scale \ottnt{Q}  }   \Vdash   \constraintfont{   \multiplicityfont{ \ottsym{(}   \pi {⋅} \rho   \ottsym{)} }  \scale \ottnt{q}  } $.
  \item If $ \constraintfont{ \ottnt{Q} }   \Vdash   \constraintfont{   \multiplicityfont{ \ottsym{(}   \pi {⋅} \rho   \ottsym{)} }  \scale \ottnt{q}  } $, then there exists $  \constraintfont{ \ottnt{Q'} }  $ such that $  \constraintfont{ \ottnt{Q} }   =   \constraintfont{   \multiplicityfont{ \pi }  \scale \ottnt{Q'}  }  $ and $ \constraintfont{ \ottnt{Q'} }   \Vdash   \constraintfont{   \multiplicityfont{ \rho }  \scale \ottnt{q}  } $.
  \item If $ \constraintfont{ \ottnt{Q_{{\mathrm{1}}}} }   \Vdash   \constraintfont{ \ottnt{Q_{{\mathrm{2}}}} } $, then $ \constraintfont{   \multiplicityfont{ \omega }  \scale \ottnt{Q_{{\mathrm{1}}}}  }   \Vdash   \constraintfont{ \ottnt{Q_{{\mathrm{2}}}} } $.
  \item If $ \constraintfont{ \ottnt{Q_{{\mathrm{1}}}} }   \Vdash   \constraintfont{ \ottnt{Q_{{\mathrm{2}}}} } $, then for all $  \constraintfont{ \ottnt{Q'} }  $, it is the case that $ \constraintfont{   \multiplicityfont{ \omega }  \scale \ottnt{Q'}   \qtensor  \ottnt{Q_{{\mathrm{1}}}} }   \Vdash   \constraintfont{ \ottnt{Q_{{\mathrm{2}}}} } $.
  \item If $  \constraintfont{ \ottnt{q} }   \in  \constraintfont{\mathcal{D} } $, then $ \constraintfont{   \multiplicityfont{ \ottsym{1} }  \scale \ottnt{q}  }   \Vdash   \constraintfont{   \multiplicityfont{ \ottsym{1} }  \scale \ottnt{q}   \qtensor    \multiplicityfont{ \ottsym{1} }  \scale \ottnt{q}  } $.
  \item If $  \constraintfont{ \ottnt{q} }   \in  \constraintfont{\mathcal{D} } $, then $ \constraintfont{   \multiplicityfont{ \ottsym{1} }  \scale \ottnt{q}  }   \Vdash   \constraintfont{  \mathbf{\varepsilon}  } $.
  \item If $  \constraintfont{ \ottnt{Q} }   \in  \constraintfont{\mathcal{D} } $ and $ \constraintfont{ \ottnt{Q'} }   \Vdash   \constraintfont{ \ottnt{Q} } $, then $  \constraintfont{ \ottnt{Q'} }   \in  \constraintfont{\mathcal{D} } $.
  \end{enumerate}
\caption{Requirements for the entailment relation $ \constraintfont{ \ottnt{Q_{{\mathrm{1}}}} }   \Vdash   \constraintfont{ \ottnt{Q_{{\mathrm{2}}}} } $}
\label{fig:entailment-relation}
\end{figure}

\begin{definition}[Entailment relation]
  \label{def:entailment-relation}
  The qualified type system is parameterised by a relation
  $ \constraintfont{ \ottnt{Q_{{\mathrm{1}}}} }   \Vdash   \constraintfont{ \ottnt{Q_{{\mathrm{2}}}} } $ between two simple constraints, as well as by a
  distinguished set $ \constraintfont{\mathcal{D} } $ of duplicable atomic constraints.

  We write, abusing notation, $  \constraintfont{ \ottnt{Q} }   \in  \constraintfont{\mathcal{D} } $ for a simple constraint
  $  \constraintfont{ \ottnt{Q} }  =  \constraintfont{ \ottsym{(}  \ottnt{U}  \ottsym{,}  \ottnt{L}  \ottsym{)} }  $ if for all $  \constraintfont{ \ottnt{q} }  \in  \constraintfont{ \ottnt{L} }  $ we
  have $  \constraintfont{ \ottnt{q} }   \in  \constraintfont{\mathcal{D} } $.

  The entailment relation must obey the laws listed in
  \Fref{fig:entailment-relation}.
\end{definition}
\info{See Fig 3, p14 of OutsideIn\cite{OutsideIn}.}
The set $ \constraintfont{\mathcal{D} } $ is a set of constraints which can be duplicated and
discarded (see~\cref{fig:entailment-relation}). We use $ \constraintfont{\mathcal{D} } $ to
model the \ensuremath{\constraintfont{\constraintfont{\Conid{Linearly}}}} constraint. Crucially, it is \emph{not
  the case} that $ \constraintfont{   \multiplicityfont{ \ottsym{1} }  \scale \ottnt{q}  }   \Vdash   \constraintfont{   \multiplicityfont{ \omega }  \scale \ottnt{q}  } $ for $  \constraintfont{ \ottnt{q} }   \in  \constraintfont{\mathcal{D} } $; such an
entailment is, in fact, prohibited (\cref{lem:q:scaling-inversion}) by
the rules of~\cref{fig:entailment-relation}. While it may seem
counter-intuitive, there is nothing in linear logic mandating that a
formula that can be duplicated and discarded be (equivalent to) an
unrestricted formula. This observation has been exploited, for
instance, to introduce so-called
subexponentials~\cite{subexponentials}. For our use case, it lets the
typechecker dispatch \ensuremath{\constraintfont{\constraintfont{\Conid{Linearly}}}} constraints (using
duplication), but prevents the result of constrained functions to be used
unrestrictedly.

An important feature of simple constraints is that, while scaling
syntactically happens at the level of atomic constraints, these properties
of scaling extend to scaling of arbitrary constraints. Define
$  \constraintfont{   \multiplicityfont{ \pi }  \scale \ottnt{Q}  }  $ as:

$$
                      \left\{
                      \begin{array}{r@{\;}c@{\;}l}
                          \constraintfont{   \multiplicityfont{ \ottsym{1} }  \scale \ottsym{(}  \ottnt{U}  \ottsym{,}  \ottnt{L}  \ottsym{)}  }   &=&   \constraintfont{ \ottsym{(}  \ottnt{U}  \ottsym{,}  \ottnt{L}  \ottsym{)} }   \\
                          \constraintfont{   \multiplicityfont{ \omega }  \scale \ottsym{(}  \ottnt{U}  \ottsym{,}  \ottnt{L}  \ottsym{)}  }   &=&   \constraintfont{ \ottsym{(}   \ottnt{U}  \cup  \ottnt{L}   \ottsym{,}  \emptyset  \ottsym{)} }  
                      \end{array}
                                                    \right.
$$
Then the following properties hold
\begin{lemma}[Scaling]
  \label{lem:q:scaling}
  If $ \constraintfont{ \ottnt{Q_{{\mathrm{1}}}} }   \Vdash   \constraintfont{ \ottnt{Q_{{\mathrm{2}}}} } $, then $ \constraintfont{   \multiplicityfont{ \pi }  \scale \ottnt{Q_{{\mathrm{1}}}}  }   \Vdash   \constraintfont{   \multiplicityfont{ \pi }  \scale \ottnt{Q_{{\mathrm{2}}}}  } $.
\end{lemma}

\begin{lemma}[Inversion of scaling]
  \label{lem:q:scaling-inversion}
  If $ \constraintfont{ \ottnt{Q_{{\mathrm{1}}}} }   \Vdash   \constraintfont{   \multiplicityfont{ \pi }  \scale \ottnt{Q_{{\mathrm{2}}}}  } $, then $  \constraintfont{ \ottnt{Q_{{\mathrm{1}}}} }  =  \constraintfont{   \multiplicityfont{ \pi }  \scale \ottnt{Q'}  }  $ and $ \constraintfont{ \ottnt{Q'} }   \Vdash   \constraintfont{ \ottnt{Q_{{\mathrm{2}}}} } $ for some $  \constraintfont{ \ottnt{Q'} }  $.
\end{lemma}

\begin{corollary}[Linear assumptions]
If $ \constraintfont{ \ottnt{Q_{{\mathrm{1}}}} }   \Vdash   \constraintfont{   \multiplicityfont{ \omega }  \scale \ottnt{Q_{{\mathrm{2}}}}  } $, then $  \constraintfont{ \ottnt{Q_{{\mathrm{1}}}} }  $ contains no linear assumptions.
\end{corollary}

Proofs of these lemmas (and others) appear in Appendix~\ref{sec:appendix:proofs-lemmas};
they can be proved by straightforward use of the properties in \Cref{fig:entailment-relation}.

\subsection{Typing rules}
\label{sec:typing-rules}

With this material in place, we can now present our type system. The
grammar is given in \Cref{fig:declarative:grammar}, which also
includes the definitions of scaling on contexts $   \multiplicityfont{ \pi }  \scale \Gamma  $ and addition
of contexts $ \Gamma_{{\mathrm{1}}}  \ottsym{+}  \Gamma_{{\mathrm{2}}} $. Note that addition on contexts is actually
a partial function, as it requires that, if a variable $ \ottmv{x} $ is bound
in both $ \Gamma_{{\mathrm{1}}} $ and $ \Gamma_{{\mathrm{2}}} $, then $ \ottmv{x} $ is assigned the same
type in both (but perhaps different multiplicities). This partiality
is not a problem in practice, as the required condition for combining
contexts is always satisfied.

\begin{figure}
  \maybesmall
  $$
  \begin{array}{@{}c@{}}
   \ottmv{a} ,  \ottmv{b}  \quad \text{Type vars} \qquad
   \ottmv{x} ,  \ottmv{y}  \quad \text{Expression vars} \qquad
  \ottmv{T} \quad \text{Type constructors} \qquad
   \ottmv{K}  \quad \text{Data constructors} \\
  \begin{array}[b]{lcll}
     \sigma  & \bnfeq &   \forall   \overline{\ottmv{a} } .   \constraintfont{ \ottnt{Q} }   \Lolly  \tau   & \text{Type schemes} \\
     \tau ,  \upsilon  & \bnfeq &  \ottmv{a}  \bnfor   \exists   \overline{\ottmv{a} } .  \tau  \RLolly   \constraintfont{ \ottnt{Q} }    \bnfor   \tau_{{\mathrm{1}}}  \to_{  \multiplicityfont{ \pi }  }  \tau_{{\mathrm{2}}}  
                            \bnfor  \ottmv{T} \, \overline{\tau}  & \text{Types} \\
     \Gamma ,  \Delta  & \bnfeq &  ∙  \bnfor  \Gamma  \ottsym{,}   \ottmv{x} {:}_{  \multiplicityfont{ \pi }  } \sigma   &
                                                              \text{Contexts} \\
     \ottnt{e}  & \bnfeq &  \ottmv{x}  \bnfor  \ottmv{K}  \bnfor  \lambda  \ottmv{x}  \ottsym{.}  \ottnt{e}  \bnfor \ottnt{e_{{\mathrm{1}}}} \, \ottnt{e_{{\mathrm{2}}}} \bnfor  \packbox \, \ottnt{e}  & \text{Expressions}\\
                 &\bnfor &  \klet\ \packbox  \ottmv{x}  =  \ottnt{e_{{\mathrm{1}}}}  \  \ottkw{in}  \  \ottnt{e_{{\mathrm{2}}}}  \bnfor   \kcase_  \multiplicityfont{ \pi }   \, \ottnt{e} \, \ottkw{of} \, \ottsym{\{}  \overline{\ottmv{K}_i\ \overline{\ottmv{x}_i } \to \ottnt{e}_i }  \ottsym{\}}  &\\
                 &\bnfor &  \klet_  \multiplicityfont{ \pi }   \, \ottmv{x}  \ottsym{=}  \ottnt{e_{{\mathrm{1}}}} \, \ottkw{in} \, \ottnt{e_{{\mathrm{2}}}} \bnfor   \klet_  \multiplicityfont{ \pi }   \, \ottmv{x}  \ottsym{:}  \sigma  \ottsym{=}  \ottnt{e_{{\mathrm{1}}}} \, \ottkw{in} \, \ottnt{e_{{\mathrm{2}}}}  &
  \end{array} \\[1ex]
  \end{array}
  $$
  Context scaling $   \multiplicityfont{ \pi }  \scale \Gamma  $ and addition of contexts $ \Gamma_{{\mathrm{1}}}  \ottsym{+}  \Gamma_{{\mathrm{2}}} $ is defined as follows:
  $$
  \begin{array}{cc}
  \left\{
  \begin{array}{r@{\;}c@{\;}l}
     \multiplicityfont{ \pi }  \scale ∙   &=&  ∙  \\
     \multiplicityfont{ \pi }  \scale \ottsym{(}  \Gamma  \ottsym{,}   \ottmv{x} {:}_{  \multiplicityfont{ \rho }  } \sigma   \ottsym{)}   &=&    \multiplicityfont{ \pi }  \scale \Gamma   \ottsym{,}   \ottmv{x} {:}_{  \multiplicityfont{ \ottsym{(}   \pi {⋅} \rho   \ottsym{)} }  } \sigma  
  \end{array}
  \right.
  &
  \left\{
  \begin{array}{r@{\;}c@{\;}lll}
   \ottsym{(}  \Gamma_{{\mathrm{1}}}  \ottsym{,}   \ottmv{x} {:}_{  \multiplicityfont{ \pi }  } \sigma   \ottsym{)}  \ottsym{+}  \Gamma_{{\mathrm{2}}}  &=&  \Gamma_{{\mathrm{1}}}  \ottsym{+}  \Gamma'_{{\mathrm{2}}}  \ottsym{,}   \ottmv{x} {:}_{  \multiplicityfont{ \ottsym{(}  \pi  \ottsym{+}  \rho  \ottsym{)} }  } \sigma   & \text{where} & \Gamma_{{\mathrm{2}}}  \ottsym{=}   \ottsym{\{}   \ottmv{x} {:}_{  \multiplicityfont{ \rho }  } \sigma   \ottsym{\}}  \cup  \Gamma'_{{\mathrm{2}}}  \\
  &&&&  \ottmv{x}  \notin  \Gamma'_{{\mathrm{2}}}  \\
   \ottsym{(}  \Gamma_{{\mathrm{1}}}  \ottsym{,}   \ottmv{x} {:}_{  \multiplicityfont{ \pi }  } \sigma   \ottsym{)}  \ottsym{+}  \Gamma_{{\mathrm{2}}}  &=&  \Gamma_{{\mathrm{1}}}  \ottsym{+}  \Gamma_{{\mathrm{2}}}  \ottsym{,}   \ottmv{x} {:}_{  \multiplicityfont{ \pi }  } \sigma   & \text{where} &  \ottmv{x}  \notin  \Gamma_{{\mathrm{2}}}  \\
   ∙  \ottsym{+}  \Gamma_{{\mathrm{2}}}  &=&  \Gamma_{{\mathrm{2}}} 
  \end{array}
  \right.
  \end{array}
  $$
  \caption{Grammar of the qualified type system}
  \label{fig:declarative:grammar}
\end{figure}

\begin{figure}
  \centering
  \drules[E]{$ \constraintfont{ \ottnt{Q} }   \ottsym{;}  \Gamma  \vdash  \ottnt{e}  \ottsym{:}  \tau$}{Expression
    typing}{Var,Abs,App,Pack,Unpack,Let,LetSig,Case,Sub}
  \caption{Qualified type system}
  \label{fig:typing-rules}
\end{figure}

The typing rules are in \Cref{fig:typing-rules}.
A qualified type system~\cite{QualifiedTypes} such as ours introduces a
judgement of the form $ \constraintfont{ \ottnt{Q} }   \ottsym{;}  \Gamma  \vdash  \ottnt{e}  \ottsym{:}  \tau$, where $ \Gamma $ is a standard
type context, and $  \constraintfont{ \ottnt{Q} }  $ is a constraint we have assumed to be true.
$  \constraintfont{ \ottnt{Q} }  $ behaves
much like $ \Gamma $, which will be instrumental for
desugaring in \cref{sec:desugaring}; the main difference is
that $ \Gamma $ is addressed explicitly, whereas $  \constraintfont{ \ottnt{Q} }  $
is used implicitly in \rref{E-Var}.

Because constraints are used implicitly, if there are several instances
of the same $  \constraintfont{   \multiplicityfont{ \ottsym{1} }  \scale \ottnt{q}  }  $, it is non-deterministic which one is used in
which instance of \rref*{E-Var}. As a consequence, we must require
that any two instances of $  \constraintfont{   \multiplicityfont{ \ottsym{1} }  \scale \ottnt{q}  }  $ in a constraint $  \constraintfont{ \ottnt{Q} }  $ have
the same computational content (see \cref{sec:desugaring}). How do we
reconcile this non-determinism with the use of linear constraints, in
\cref{sec:arrays} to thread mutations? We certainly don't want type
inference to non-deterministically reorder a \ensuremath{\Varid{readRef}} and a
\ensuremath{\Varid{writeRef}}! The solution is that the \textsc{api} is arranged so that
only a single instance of \ensuremath{\constraintfont{\Conid{RW}\;\Varid{n}}} is ever
provided. Therefore there is a single possible threading of the reads
and writes. In contrast there will often be several instances of
\ensuremath{\constraintfont{\Conid{Linearly}}} in scope.

The type system of \Cref{fig:typing-rules} is purely
declarative: note, for example, that \rref{E-App} does not describe
how to break the typing assumptions into constraints
$  \constraintfont{ \ottnt{Q_{{\mathrm{1}}}} }  $/$  \constraintfont{ \ottnt{Q_{{\mathrm{2}}}} }  $ and contexts $ \Gamma_{{\mathrm{1}}} $/$ \Gamma_{{\mathrm{2}}} $. We will see
how to infer constraints in \cref{sec:type-inference}. Yet,
this system is our ground truth: a system with a simple enough
definition that programmers can reason about typing. As is standard in
the qualified-type literature (since the original paper~\cite{QualifiedTypes}), we do not
directly give a dynamic
semantics to this language; instead, we will give it meaning via
desugaring to a simpler core language in \cref{sec:desugaring}.

We survey several distinctive features of our qualified type system below:

\info{See Fig 10, p25 of OutsideIn\cite{OutsideIn}.}
\paragraph{Linear functions.}
The type of linear functions is written $  \ottmv{a}  \to_{  \multiplicityfont{ \ottsym{1} }  }  \ottmv{b}  $.
  Despite our focus on linear constraints,
  we still need linearity in ordinary arguments.
  Indeed, the linearity of arrows interacts in interesting
  ways with linear constraints: If $\ottmv{f}  \ottsym{:}   \ottmv{a}  \to_{  \multiplicityfont{ \omega }  }  \ottmv{b} $ and
  $\ottmv{x}  \ottsym{:}   \constraintfont{   \multiplicityfont{ \ottsym{1} }  \scale \ottnt{q}  }   \Lolly  \ottmv{a}$, then calling $ \ottmv{f} \, \ottmv{x} $ would actually use $  \constraintfont{ \ottnt{q} }  $
  many times. We must make sure it is impossible to derive
  $ \constraintfont{   \multiplicityfont{ \ottsym{1} }  \scale \ottnt{q}  }   \ottsym{;}   \ottmv{f} {:}_{  \multiplicityfont{ \omega }  }  \ottmv{a}  \to_{  \multiplicityfont{ \omega }  }  \ottmv{b}    \ottsym{,}   \ottmv{x} {:}_{  \multiplicityfont{ \omega }  }  \constraintfont{   \multiplicityfont{ \ottsym{1} }  \scale \ottnt{q}  }   \Lolly  \ottmv{a}   \vdash  \ottmv{f} \, \ottmv{x}  \ottsym{:}  \ottmv{b}$.
  Otherwise we could make, for instance, the \ensuremath{\Varid{overusing}} function from
  \cref{sec:overusing}.
  You can check that $ \constraintfont{   \multiplicityfont{ \ottsym{1} }  \scale \ottnt{q}  }   \ottsym{;}   \ottmv{f} {:}_{  \multiplicityfont{ \omega }  }  \ottmv{a}  \to_{  \multiplicityfont{ \omega }  }  \ottmv{b}    \ottsym{,}   \ottmv{x} {:}_{  \multiplicityfont{ \omega }  }  \constraintfont{   \multiplicityfont{ \ottsym{1} }  \scale \ottnt{q}  }   \Lolly  \ottmv{a}   \vdash  \ottmv{f} \, \ottmv{x}  \ottsym{:}  \ottmv{b}$
  indeed does not
  type check, because the scaling of $  \constraintfont{ \ottnt{Q_{{\mathrm{2}}}} }  $ in \rref{E-App} ensures that
  the constraint would be $  \constraintfont{   \multiplicityfont{ \omega }  \scale \ottnt{q}  }  $ instead. On the other hand,
  it is perfectly fine to have $ \constraintfont{   \multiplicityfont{ \ottsym{1} }  \scale \ottnt{q}  }   \ottsym{;}   \ottmv{f} {:}_{  \multiplicityfont{ \omega }  }  \ottmv{a}  \to_{  \multiplicityfont{ \ottsym{1} }  }  \ottmv{b}    \ottsym{,}   \ottmv{x} {:}_{  \multiplicityfont{ \omega }  }  \constraintfont{   \multiplicityfont{ \ottsym{1} }  \scale \ottnt{q}  }   \Lolly  \ottmv{a}   \vdash  \ottmv{f} \, \ottmv{x}  \ottsym{:}  \ottmv{b}$ when $ \ottmv{f} $ is a linear function.

\paragraph{Variables.}
As is standard, the \rref{E-Var} rule works in a context containing more
than just the used binding for $ \ottmv{x} $. However, crucially,
our rule allows only
\emph{unrestricted} variables to be discarded; linear variables \emph{must} be
used. We can see this in the rule by noticing that the context has an unrestricted
component $   \multiplicityfont{ \omega }  \scale \Gamma_{{\mathrm{2}}}  $. The $ \Gamma_{{\mathrm{1}}} $ component might be restricted or might not,
allowing this rule to apply both for restricted and unrestricted $ \ottmv{x} $.

\paragraph{Data constructors.}
Data constructors $ \ottmv{K} $ don't have a dedicated typing
rule. Instead they are typed using the \rref{E-Var}, where they are
treated as if they were unrestricted variables.

\paragraph{Let-bindings.}
Bindings in a \ensuremath{\mathbf{let}} may be for either linear or unrestricted variables.
  We could require all bindings to be linear and to implement unrestricted
  information only using \ensuremath{\Conid{Ur}}, but it is very easy to add a multiplicity
  annotation on \ensuremath{\mathbf{let}}, and so we do.

\paragraph{Local assumptions.}
\Rref{E-Let} includes support for local
  assumptions. We thus have the ability to generalise a subset of
  the constraints needed by $ \ottnt{e_{{\mathrm{1}}}} $ (but not the type variables---no
  \ensuremath{\mathbf{let}}-generalisation here, though it could be added). The inference algorithm of
  \cref{sec:type-inference} will not make use of this
  possibility. 

\paragraph{Existentials.}
 We include $  \exists   \overline{\ottmv{a} } .  \tau  \RLolly   \constraintfont{ \ottnt{Q} }   $, as
  introduced in \cref{sec:what-it-looks-like}, together with
  the $\packbox$ constructor. See rules~\rref*{E-Pack} and
  \rref*{E-Unpack}.
\info{No substitution on $  \constraintfont{ \ottnt{Q_{{\mathrm{1}}}} }  $ in the E-Unpack rule, because there is
  only existential quantification.}

\section{Constraint inference}
\label{sec:type-inference}

The type system of \Cref{fig:typing-rules} gives a declarative description
of what programs are acceptable. We now present the algorithmic counterpart to
this system. Our algorithm is structured, unsurprisingly, around generating and
solving constraints, broadly following the template of
\citet{essence-of-ml-type-inference}.
That is, our algorithm takes a pass over the abstract syntax entered by the
user, generating constraints as it goes. Then, separately, we solve those
constraints (that is, try to satisfy them) in the presence of a set of assumptions,
or we determine that the assumptions do not imply that the constraints hold. In the
latter case, we issue an error to the programmer.

The procedure is responsible for inferring both \emph{types} and \emph{constraints}.
For our type system, type inference can be done independently from constraint
inference. Indeed, we focus on the latter, and defer type inference to
an external oracle (such as~\cite{linear-types-inference}).
That is, we assume an algorithm that produces typing derivations for the
judgement $\Gamma  \vdash  \ottnt{e}  \ottsym{:}  \tau$, ignoring all the constraints. Then, we describe a
constraint generation algorithm that passes over these typing derivations.
%
%
We can make this simplification for two reasons:
\begin{itemize}
\item We do not formalise type equality constraints, and our implementation
  in \textsc{ghc} (\cref{sec:equality-constraints}) takes care to not allow linear equality constraints to influence type inference.
  Indeed, a typical treatment of unification
  would be unsound for linear equalities, because it reuses the same
  equality many times (or none at all). Linear equalities make sense
  (\citet{shulman2018linear} puts linear
  equalities to great use), but they do not seem to lend themselves to
  automation.
\item We do not support, or intend to support, multiplicity
  polymorphism in constraint arrows. That is, the multiplicity of a
  constraint is always syntactically known to be either linear or
  unrestricted. This way, no equality constraints (which might, conceivably,
  relate multiplicity variables) can interfere with
  constraint resolution.
\end{itemize}
%

\subsection{Wanted constraints}
\label{sec:wanteds}

The constraints $  \constraintfont{ \constraintfont{C} }  $ generated in our system have a richer
logical structure than the simple constraints $  \constraintfont{ \ottnt{Q} }  $, above. Following
\textsc{ghc} and echoing \citet{OutsideIn}, we call these \emph{wanted constraints}:
they are constraints which the constraint solver \emph{wants} to prove.
An unproved wanted constraint results in a type error reported to the programmer.
$$
\begin{array}{lcll}
    \constraintfont{ \constraintfont{C} }   & \bnfeq &   \constraintfont{ \ottnt{Q} }   \bnfor   \constraintfont{ \constraintfont{C_{{\mathrm{1}}}}  \qtensor  \constraintfont{C_{{\mathrm{2}}}} }   \bnfor   \constraintfont{ \constraintfont{C_{{\mathrm{1}}}}  \aand  \constraintfont{C_{{\mathrm{2}}}} }   \bnfor   \constraintfont{   \multiplicityfont{ \pi }  \scale( \ottnt{Q}  \Lolly  \constraintfont{C} )  }  &
                                                                \text{Wanted constraints}
\end{array}
$$
A simple constraint is a valid wanted constraint, and we have two forms of
conjunction for wanted constraints:
the new
$  \constraintfont{ \constraintfont{C_{{\mathrm{1}}}}  \aand  \constraintfont{C_{{\mathrm{2}}}} }  $ construction (read $  \constraintfont{ \constraintfont{C_{{\mathrm{1}}}} }  $ \emph{with} $  \constraintfont{ \constraintfont{C_{{\mathrm{2}}}} }  $), alongside
the more typical $  \constraintfont{ \constraintfont{C_{{\mathrm{1}}}}  \qtensor  \constraintfont{C_{{\mathrm{2}}}} }  $. These are
connectives from linear logic: $  \constraintfont{ \constraintfont{C_{{\mathrm{1}}}}  \qtensor  \constraintfont{C_{{\mathrm{2}}}} }  $ is the
\emph{multiplicative} conjunction, and $  \constraintfont{ \constraintfont{C_{{\mathrm{1}}}}  \aand  \constraintfont{C_{{\mathrm{2}}}} }  $ is the \emph{additive}
conjunction. Both connectives are conjunctions, but they differ
in meaning. To satisfy $  \constraintfont{ \constraintfont{C_{{\mathrm{1}}}}  \qtensor  \constraintfont{C_{{\mathrm{2}}}} }  $ one consumes the (linear)
assumptions consumed by satisfying $  \constraintfont{ \constraintfont{C_{{\mathrm{1}}}} }  $ and those consumed by $  \constraintfont{ \constraintfont{C_{{\mathrm{2}}}} }  $;
if an assumed linear constraint is needed to prove both $  \constraintfont{ \constraintfont{C_{{\mathrm{1}}}} }  $ and $  \constraintfont{ \constraintfont{C_{{\mathrm{2}}}} }  $,
then $  \constraintfont{ \constraintfont{C_{{\mathrm{1}}}}  \qtensor  \constraintfont{C_{{\mathrm{2}}}} }  $ will not be provable, because that linear assumption cannot
be used twice. On the
other hand, satisfying $  \constraintfont{ \constraintfont{C_{{\mathrm{1}}}}  \aand  \constraintfont{C_{{\mathrm{2}}}} }  $ requires that satisfying $  \constraintfont{ \constraintfont{C_{{\mathrm{1}}}} }  $
and $  \constraintfont{ \constraintfont{C_{{\mathrm{2}}}} }  $ must each
consume the \emph{same} assumptions, which $  \constraintfont{ \constraintfont{C_{{\mathrm{1}}}}  \aand  \constraintfont{C_{{\mathrm{2}}}} }  $ consumes as well.
Thus, if $  \constraintfont{ \constraintfont{C} }  $ is assumed linearly (and we have no other assumptions),
then $  \constraintfont{ \constraintfont{C}  \qtensor  \constraintfont{C} }  $ is not provable, while $  \constraintfont{ \constraintfont{C}  \aand  \constraintfont{C} }  $ is.
The intuition, here, is that in $  \constraintfont{ \constraintfont{C_{{\mathrm{1}}}}  \aand  \constraintfont{C_{{\mathrm{2}}}} }  $, only
one of $  \constraintfont{ \constraintfont{C_{{\mathrm{1}}}} }  $ or $  \constraintfont{ \constraintfont{C_{{\mathrm{2}}}} }  $ will be eventually used. ``With'' constraints
arise from the branches in a $\kcase$-expression.

The last form of wanted constraint $  \constraintfont{ \constraintfont{C} }  $ is an implication
$  \constraintfont{   \multiplicityfont{ \pi }  \scale( \ottnt{Q}  \Lolly  \constraintfont{C} )  }  $. The more interesting case is $  \constraintfont{   \multiplicityfont{ \omega }  \scale( \ottnt{Q}  \Lolly  \constraintfont{C} )  }  $:
to prove $  \constraintfont{   \multiplicityfont{ \omega }  \scale( \ottnt{Q}  \Lolly  \constraintfont{C} )  }  $, you need to prove $  \constraintfont{ \constraintfont{C} }  $ under the
\emph{linear} assumption $  \constraintfont{ \ottnt{Q} }  $, but without using any other linear
assumptions.

These implications arise when we unpack an existential
package that contains a linear constraint and also when checking a \ensuremath{\mathbf{let}}-binding.
We can define scaling over wanted constraints by recursion as follows, where we
use scaling over simple constraints in the simple-constraint case:
$$
\left\{
  \begin{array}{lcl}
      \constraintfont{   \multiplicityfont{ \pi }  \scale \ottsym{(}  \constraintfont{C_{{\mathrm{1}}}}  \qtensor  \constraintfont{C_{{\mathrm{2}}}}  \ottsym{)}  }   & = &   \constraintfont{   \multiplicityfont{ \pi }  \scale \constraintfont{C_{{\mathrm{1}}}}   \qtensor    \multiplicityfont{ \pi }  \scale \constraintfont{C_{{\mathrm{2}}}}  }   \\
      \constraintfont{   \multiplicityfont{ \ottsym{1} }  \scale \ottsym{(}  \constraintfont{C_{{\mathrm{1}}}}  \aand  \constraintfont{C_{{\mathrm{2}}}}  \ottsym{)}  }   & = &   \constraintfont{ \constraintfont{C_{{\mathrm{1}}}}  \aand  \constraintfont{C_{{\mathrm{2}}}} }   \\
      \constraintfont{   \multiplicityfont{ \omega }  \scale \ottsym{(}  \constraintfont{C_{{\mathrm{1}}}}  \aand  \constraintfont{C_{{\mathrm{2}}}}  \ottsym{)}  }   & = &   \constraintfont{   \multiplicityfont{ \omega }  \scale \constraintfont{C_{{\mathrm{1}}}}   \qtensor    \multiplicityfont{ \omega }  \scale \constraintfont{C_{{\mathrm{2}}}}  }   \\
      \constraintfont{   \multiplicityfont{ \pi }  \scale \ottsym{(}    \multiplicityfont{ \rho }  \scale( \ottnt{Q}  \Lolly  \constraintfont{C} )   \ottsym{)}  }   & = &   \constraintfont{   \multiplicityfont{ \ottsym{(}   \pi {⋅} \rho   \ottsym{)} }  \scale( \ottnt{Q}  \Lolly  \constraintfont{C} )  }  
  \end{array}
\right.
$$
For the most part, scaling of wanted constraints is straightforward. The only
peculiar case is when we scale the additive conjunction $  \constraintfont{ \constraintfont{C_{{\mathrm{1}}}}  \aand  \constraintfont{C_{{\mathrm{2}}}} }  $ by
$  \multiplicityfont{ \omega }  $, the result is a multiplicative conjunction. The intuition here is
that when if we have both $  \constraintfont{   \multiplicityfont{ \omega }  \scale \constraintfont{C_{{\mathrm{1}}}}  }  $ and $  \constraintfont{   \multiplicityfont{ \omega }  \scale \constraintfont{C_{{\mathrm{2}}}}  }  $, then
a choice between $  \constraintfont{ \constraintfont{C_{{\mathrm{1}}}} }  $ and $  \constraintfont{ \constraintfont{C_{{\mathrm{2}}}} }  $ can be made $  \multiplicityfont{ \omega }  $ times.

We define an entailment relation over wanteds in \Cref{fig:wanted:entailment}.
Note that this relation uses only simple constraints $  \constraintfont{ \ottnt{Q} }  $ as assumptions, as
there is no way to assume the more elaborate $  \constraintfont{ \constraintfont{C} }  $\footnote{Allowing the full wanted-constraint syntax
in assumptions is the subject of work by \citet{quantified-constraints}.}.
\begin{figure}
  \maybesmall
  \centering
  \drules[C]{$ \constraintfont{ \ottnt{Q} }   \vdash   \constraintfont{ \constraintfont{C} } $} {Wanted-constraint entailment}
  {Dom,Id,Tensor,With,Impl}
  \caption{Wanted-constraint entailment}
  \label{fig:wanted:entailment}
\end{figure}

Before we move on to constraint generation proper, let us highlight a few
technical, yet essential, lemmas about the wanted-constraint
entailment relation.

\begin{lemma}[Inversion]
  \label{lem:inversion}
  The inference rules of $ \constraintfont{ \ottnt{Q} }   \vdash   \constraintfont{ \constraintfont{C} } $ can be read bottom-up (up to the
  set $ \constraintfont{\mathcal{D} } $) as well
  as top-down, as is required of $ \constraintfont{ \ottnt{Q_{{\mathrm{1}}}} }   \Vdash   \constraintfont{ \ottnt{Q_{{\mathrm{2}}}} } $ in
  \Cref{fig:entailment-relation}. That is:
  \begin{itemize}
  \item If $ \constraintfont{ \ottnt{Q} }   \vdash   \constraintfont{ \constraintfont{C_{{\mathrm{1}}}}  \qtensor  \constraintfont{C_{{\mathrm{2}}}} } $, then there exists $  \constraintfont{ \ottnt{Q_{{\mathrm{1}}}} }  $, $  \constraintfont{ Q_{\mathcal{D} } }  $ and $  \constraintfont{ \ottnt{Q_{{\mathrm{2}}}} }  $
    such that $  \constraintfont{ Q_{\mathcal{D} } }   \in  \constraintfont{\mathcal{D} } $, $ \constraintfont{ \ottnt{Q_{{\mathrm{1}}}}  \qtensor  Q_{\mathcal{D} } }   \vdash   \constraintfont{ \constraintfont{C_{{\mathrm{1}}}} } $, $ \constraintfont{ Q_{\mathcal{D} }  \qtensor  \ottnt{Q_{{\mathrm{2}}}} }   \vdash   \constraintfont{ \constraintfont{C_{{\mathrm{2}}}} } $, and
    $  \constraintfont{ \ottnt{Q} }   =   \constraintfont{ \ottnt{Q_{{\mathrm{1}}}}  \qtensor  Q_{\mathcal{D} }  \qtensor  \ottnt{Q_{{\mathrm{2}}}} }  $.
  \item If $ \constraintfont{ \ottnt{Q} }   \vdash   \constraintfont{ \constraintfont{C_{{\mathrm{1}}}}  \aand  \constraintfont{C_{{\mathrm{2}}}} } $, then $ \constraintfont{ \ottnt{Q} }   \vdash   \constraintfont{ \constraintfont{C_{{\mathrm{1}}}} } $ and $ \constraintfont{ \ottnt{Q} }   \vdash   \constraintfont{ \constraintfont{C_{{\mathrm{2}}}} } $.
  \item If $ \constraintfont{ \ottnt{Q} }   \vdash   \constraintfont{   \multiplicityfont{ \pi }  \scale( \ottnt{Q_{{\mathrm{2}}}}  \Lolly  \constraintfont{C} )  } $, then there exists $  \constraintfont{ \ottnt{Q_{{\mathrm{1}}}} }  $ such
    that $ \constraintfont{ \ottnt{Q_{{\mathrm{1}}}}  \qtensor  \ottnt{Q_{{\mathrm{2}}}} }   \vdash   \constraintfont{ \constraintfont{C} } $ and  $  \constraintfont{ \ottnt{Q} }   =   \constraintfont{   \multiplicityfont{ \pi }  \scale \ottnt{Q_{{\mathrm{1}}}}  }  $
  \end{itemize}
\end{lemma}

\begin{lemma}[Scaling]
  \label{lem:wanted:promote}
  If $ \constraintfont{ \ottnt{Q} }   \vdash   \constraintfont{ \constraintfont{C} } $, then $ \constraintfont{   \multiplicityfont{ \pi }  \scale \ottnt{Q}  }   \vdash   \constraintfont{   \multiplicityfont{ \pi }  \scale \constraintfont{C}  } $.
\end{lemma}

\begin{lemma}[Inversion of scaling]
  \label{lem:wanted:demote}
  If $ \constraintfont{ \ottnt{Q} }   \vdash   \constraintfont{   \multiplicityfont{ \pi }  \scale \constraintfont{C}  } $ then $ \constraintfont{ \ottnt{Q'} }   \vdash   \constraintfont{ \constraintfont{C} } $ and $  \constraintfont{ \ottnt{Q} }   =   \constraintfont{   \multiplicityfont{ \pi }  \scale \ottnt{Q'}  }  $ for some $  \constraintfont{ \ottnt{Q'} }  $.
\end{lemma}

\subsection{Constraint generation}
\label{sec:constraint-generation}
\label{sec:constraint-generation-soundness}

The process of inferring constraints is split into two parts: generating
constraints, which we do in this section, then solving them in
\cref{sec:constraint-solver}. Constraint generation is described by
the judgement $\Gamma  \vdashi  \ottnt{e}  \ottsym{:}  \tau  \leadsto   \constraintfont{ \constraintfont{C} } $ (defined in
\Cref{fig:constraint-generation}) which outputs a constraint $  \constraintfont{ \constraintfont{C} }  $
required to make $\ottnt{e}$ typecheck.
The definition
$\Gamma  \vdashi  \ottnt{e}  \ottsym{:}  \tau  \leadsto   \constraintfont{ \constraintfont{C} } $ is syntax directed, so it can directly be read as an
algorithm, taking as input a \emph{typing derivation} for $\Gamma  \vdash  \ottnt{e}  \ottsym{:}  \tau$
(produced by an external type inference oracle as discussed above). Notably, the
algorithm has access to the context splitting from the (previously computed)
typing derivation, and is
thus indeed syntax directed.
\info{See Fig.13, p39 of OutsideIn~\cite{OutsideIn}}

\info{Not caring about inferences simplifies $\packbox$ quite a bit, we
  are using the pseudo-inferred type to generate constraint. In a real
  system, we would need $\packbox$ to know its type (\emph{e.g.} using
  bidirectional type checking).}
\begin{figure}
  \maybesmall
  \centering
  \drules[G]{$\Gamma  \vdashi  \ottnt{e}  \ottsym{:}  \tau  \leadsto   \constraintfont{ \constraintfont{C} } $}{Constraint generation}{Var, Abs,
    App, Pack, Unpack, Case, Let, LetSig}

  \caption{Constraint generation}
  \label{fig:constraint-generation}
\end{figure}

The rules of \Cref{fig:constraint-generation} constitute a mostly
unsurprising translation of the rules of \Cref{fig:typing-rules},
except for the following points of interest:

\emph{Case expressions.}
Note the use of $ \aand $ in the conclusion of \rref{G-Case}.
We require that each branch of a $\kcase$ expression use the exact
same (linear) assumptions; this is enforced by combining the
emitted constraints with $ \aand $, not $ \qtensor $.
  This can also be understood in terms of the array example of
  \cref{sec:introduction}:
  if an array is freed in one branch of a $\kcase$, we require it to be freed (or freezed) in
the other branches too.
  Otherwise, the array's state will be unknown to the type system
  after the $\kcase$.

\emph{Implications.} The introduction of constraints local to a
  definition (\rref{G-LetSig}) corresponds to
  emitting an implication constraint.

\emph{Unannotated \ensuremath{\mathbf{let}}.}
 However, the~\rref*{G-Let} rule does not produce an implication
  constraint, as we do not model \ensuremath{\mathbf{let}}-generalisation~\cite{let-should-not-be-generalised}.

\vspace{1ex}
The key property of the constraint-generation algorithm is that,
if the generated constraint is solvable, then we can indeed type the
term in the qualified type system of
\cref{sec:qualified-type-system}. That is,
these rules are simply an implementation of our declarative qualified
type system.

\begin{lemma}[Soundness of constraint generation]\label{lem:generation-soundness}
  For all $  \constraintfont{ \ottnt{Q_{\ottmv{g}}} }  $, if $\Gamma  \vdashi  \ottnt{e}  \ottsym{:}  \tau  \leadsto   \constraintfont{ \constraintfont{C} } $ and $ \constraintfont{ \ottnt{Q_{\ottmv{g}}} }   \vdash   \constraintfont{ \constraintfont{C} } $ then
  $ \constraintfont{ \ottnt{Q_{\ottmv{g}}} }   \ottsym{;}  \Gamma  \vdash  \ottnt{e}  \ottsym{:}  \tau$.
\end{lemma}

\subsection{Constraint solving}
\label{sec:constraint-solver}

In this section,
we build a \emph{constraint solver} that proves
that $ \constraintfont{ \ottnt{Q_{\ottmv{g}}} }   \vdash   \constraintfont{ \constraintfont{C} } $ holds, as required by \cref{lem:generation-soundness}.
The constraint solver is represented by the following judgement:
$$
 \constraintfont{ \ottnt{U} }   \ottsym{;}   \constraintfont{ \ottnt{D} }   \ottsym{;}   \constraintfont{ \ottnt{L}_{\ottmv{i}} }   \vdashs   \constraintfont{ \constraintfont{C} }   \leadsto   \constraintfont{ \ottnt{L}_{\ottmv{o}} } 
$$
The judgement
takes in three contexts: $  \constraintfont{ \ottnt{U} }  $, which holds all the unrestricted
atomic constraint assumptions, $  \constraintfont{ \ottnt{D} }  $ which holds the linear atomic
assumptions which are members of $ \constraintfont{\mathcal{D} } $ and $  \constraintfont{ \ottnt{L}_{\ottmv{i}} }  $, which holds the linear
atomic constraint assumptions which aren't members of $ \constraintfont{\mathcal{D} } $.
The linear contexts $  \constraintfont{ \ottnt{D} }  $, $  \constraintfont{ \ottnt{L}_{\ottmv{i}} }  $, and $  \constraintfont{ \ottnt{L}_{\ottmv{o}} }  $ have been described
as multisets (\cref{sec:constraint-domain}), but we treat them
as ordered lists in the more concrete setting here; we will see soon
why this treatment is necessary.

Linearity requires treating constraints as consumable resources. This
is what $  \constraintfont{ \ottnt{L}_{\ottmv{o}} }  $ is for: it contains the hypotheses of
$  \constraintfont{ \ottnt{L}_{\ottmv{i}} }  $ which are not consumed when proving $  \constraintfont{ \constraintfont{C} }  $. As
suggested by the notation, it is an output of the
algorithm. Constraints from $  \constraintfont{ \ottnt{D} }  $ are never outputted in
$  \constraintfont{ \ottnt{L}_{\ottmv{i}} }  $: if constraints from $  \constraintfont{ \ottnt{D} }  $ remain unused, we
weaken them instead.

If the constraint solver finds a solution, then the output linear constraints
must be a subset of the input linear constraints, and the solution must indeed
be entailed from the given assumptions.
\begin{lemma}[Constraint solver soundness]
\label{lem:solver-soundness} If $ \constraintfont{ \ottnt{U} }   \ottsym{;}   \constraintfont{ \ottnt{D} }   \ottsym{;}   \constraintfont{ \ottnt{L}_{\ottmv{i}} }   \vdashs   \constraintfont{ \constraintfont{C} }   \leadsto   \constraintfont{ \ottnt{L}_{\ottmv{o}} } $, then:
\begin{enumerate}
\item $  \constraintfont{ \ottnt{L}_{\ottmv{o}} }  \subseteq   \constraintfont{ \ottnt{L}_{\ottmv{i}} }  $
\item $ \constraintfont{ \ottsym{(}  \ottnt{U}  \ottsym{,}   \ottnt{D}  \uplus  \ottnt{L}_{\ottmv{i}}   \ottsym{)} }   \vdash   \constraintfont{ \constraintfont{C}  \qtensor  \ottsym{(}  \emptyset  \ottsym{,}  \ottnt{L}_{\ottmv{o}}  \ottsym{)} } $
\end{enumerate}
\end{lemma}

To handle simple wanted constraints, we will need  a domain-specific
\emph{atomic-constraint solver} to be the algorithmic counterpart of the
abstract entailment relation of \cref{sec:constraint-domain}. The
main solver will appeal to this atomic-constraint solver when solving atomic
constraints.  The atomic-constraint solver is represented by the following
judgement:
$$
  \constraintfont{ \ottnt{U} }  ;   \constraintfont{ \ottnt{D} }  ;   \constraintfont{ \ottnt{L}_{\ottmv{i}} }    \vdashsimp    \multiplicityfont{ \pi }   \scale   \constraintfont{ \ottnt{q} }    \leadsto    \constraintfont{ \ottnt{L}_{\ottmv{o}} }  
$$

It has a similar structure to the main solver, but only deals with atomic
constraints. Even though the main solver is parameterised by this
atomic-constraint solver, we will give an instantiation in
\cref{sec:simple-constr-solv}.
We require the following property of the atomic-constraint solver:

\begin{property}[Atomic-constraint solver soundness]
\label{prop:atomic-solver-soundness} If $  \constraintfont{ \ottnt{U} }  ;   \constraintfont{ \ottnt{D} }  ;   \constraintfont{ \ottnt{L}_{\ottmv{i}} }    \vdashsimp    \multiplicityfont{ \pi }   \scale   \constraintfont{ \ottnt{q} }    \leadsto    \constraintfont{ \ottnt{L}_{\ottmv{o}} }  $, then:
\begin{enumerate}
\item $  \constraintfont{ \ottnt{L}_{\ottmv{o}} }  \subseteq   \constraintfont{ \ottnt{L}_{\ottmv{i}} }  $
\item $ \constraintfont{ \ottsym{(}  \ottnt{U}  \ottsym{,}   \ottnt{D}  \uplus  \ottnt{L}_{\ottmv{i}}   \ottsym{)} }   \Vdash   \constraintfont{   \multiplicityfont{ \pi }  \scale \ottnt{q}   \qtensor  \ottsym{(}  \emptyset  \ottsym{,}  \ottnt{L}_{\ottmv{o}}  \ottsym{)} } $
\end{enumerate}
\end{property}

\subsubsection{Constraint solver algorithm}

Building on this atomic-constraint solver, we use a linear proof
search algorithm based on the recipe given
by~\citet{resource-management-for-ll-proof-search}. \Cref{fig:constraint-solver}
presents the rules of the constraint solver.

\begin{figure}
  \maybesmall
  \centering
  \drules[S]{$ \constraintfont{ \ottnt{U} }   \ottsym{;}   \constraintfont{ \ottnt{D} }   \ottsym{;}   \constraintfont{ \ottnt{L}_{\ottmv{i}} }   \vdashs   \constraintfont{ \constraintfont{C_{\ottmv{w}}} }   \leadsto   \constraintfont{ \ottnt{L}_{\ottmv{o}} } $}{Constraint solving}{Atom, Mult, ImplOne, Add, ImplMany}
  \caption{Constraint solver}
  \label{fig:constraint-solver}
\end{figure}

\begin{itemize}
  \item The~\rref*{S-Mult} rule proceeds by solving one side of a
conjunction first, then passing the output constraints to the other side.
Both the unrestricted context and the duplicable context are shared between both sides.
  \item The~\rref*{S-Add} rule handles additive conjunction. The linear
constraints are shared between the branches (since additive conjunction is
generated from $\kcase$ expressions, only one of them is actually going to be
executed). Both branches must consume exactly the same
resources.
  \item The~\rref*{S-ImplOne} rules handles linear implications. The
    unrestricted and linear components of the assumption are unioned
    with their respective context when solving the conclusion. Note
    how the linear constraints, in particular, are classified
    according to whether they are members of $ \constraintfont{\mathcal{D} } $ or not. Importantly (see
\cref{sec:simple-constr-solv}), the linear assumptions are
added to the front of the lists.
 The side condition that the output
context is a subset of the input context ensures that the implication
fully consumes its assumption and does not leak it to the ambient
context.
\item The~\rref*{S-ImplMany} rules handles unrestricted implication.
  The conclusion uses its own linear assumption, but none of the other
  linear constraints.  This is because, as per~\rref*{C-Impl},
  unrestricted implications can only use an unrestricted context. In
  particular, crucially, the constraints from $  \constraintfont{ \ottnt{D} }  $, despite
  being duplicable and discardable, are not (and cannot) be
  used to prove an unrestricted implication, as first discussed
  in~\cref{sec:constraint-domain}.
\end{itemize}

\subsubsection{An atomic-constraint solver}
\label{sec:simple-constr-solv}

So far, the atomic-constraint domain has been
an abstract parameter. In this section, though, we offer a concrete
domain which supports our examples.

For the sake of our examples, we need very little: linear
constraints can remain abstract. It is thus sufficient for the entailment relation
(\Cref{fig:simpl-entailment}) to prove $  \constraintfont{ \ottnt{q} }  $ if and only if it
is already assumed---while respecting linearity. That is, with the
exception of a distinguished constraint $  \constraintfont{  \mathcal{L}  }  $, which can be
duplicated and discarded, and we use to model the \ensuremath{\constraintfont{\Conid{Linearly}}}
constraint from~\ref{sec:Unique-constraint}. The set $ \constraintfont{\mathcal{D} } $ is
defined to only contain $  \constraintfont{  \mathcal{L}  }  $, therefore the $  \constraintfont{ \ottnt{D} }  $
context is a sequence of 0 or more $  \constraintfont{  \mathcal{L}  }  $.

\begin{figure}
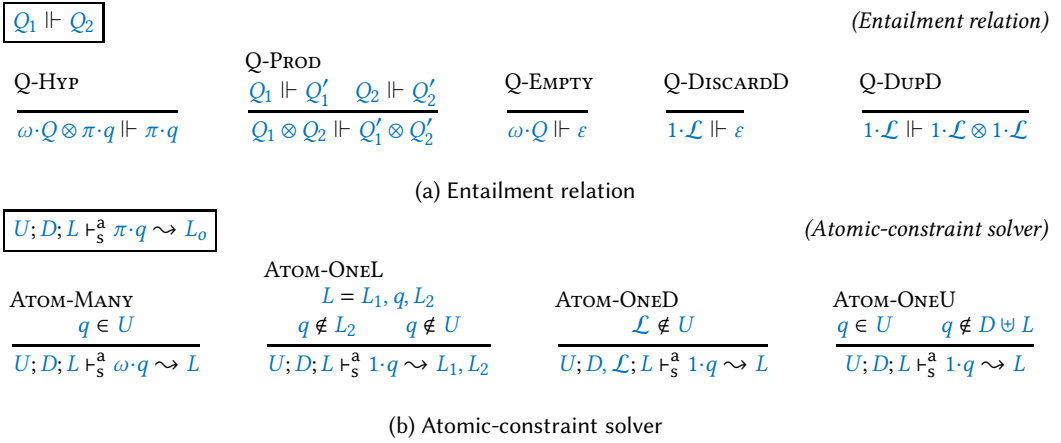

  \maybesmall
  \centering
  \begin{subfigure}{\linewidth}
    \drules[Q]{$ \constraintfont{ \ottnt{Q_{{\mathrm{1}}}} }   \Vdash   \constraintfont{ \ottnt{Q_{{\mathrm{2}}}} } $}{Entailment relation}{Hyp,Prod,Empty,
      DiscardD, DupD}
    \caption{Entailment relation}
    \label{fig:simpl-entailment}
  \end{subfigure}
  \begin{subfigure}{\linewidth}
    \drules[Atom]{$  \constraintfont{ \ottnt{U} }  ;   \constraintfont{ \ottnt{D} }  ;   \constraintfont{ \ottnt{L} }    \vdashsimp    \multiplicityfont{ \pi }   \scale   \constraintfont{ \ottnt{q} }    \leadsto    \constraintfont{ \ottnt{L}_{\ottmv{o}} }  $}{Atomic-constraint solver}{Many,OneL,OneD, OneU}
    \caption{Atomic-constraint solver}
    \label{fig:simpl-solver}
  \end{subfigure}
  \caption{A stripped-down constraint domain}
  \label{fig:predicate-domain}
\end{figure}

The corresponding atomic-constraint solver
(\Cref{fig:simpl-solver}) is more interesting.
It is deterministic: in all circumstances,
only one of the three rules can apply. This means that the
algorithm does not guess, thus never needs to backtrack.
Avoiding guesses is a key property of \textsc{ghc}'s solver~\cite[Section~6.4]{OutsideIn},
one we must maintain if we are to be compatible with \textsc{ghc}.

\Cref{fig:simpl-solver} is also where the fact that the
$  \constraintfont{ \ottnt{L} }  $ are lists comes into play. Indeed, \rref{Atom-OneL}
takes care to use the most recent occurrence of $  \constraintfont{ \ottnt{q} }  $
(remember that \rref{S-ImplOne} adds the new hypotheses on the front of
the list). To understand why, consider the following example:
\begin{hscode}\SaveRestoreHook
\column{B}{@{}>{\hspre}l<{\hspost}@{}}%
\column{7}{@{}>{\hspre}l<{\hspost}@{}}%
\column{9}{@{}>{\hspre}l<{\hspost}@{}}%
\column{16}{@{}>{\hspre}l<{\hspost}@{}}%
\column{21}{@{}>{\hspre}c<{\hspost}@{}}%
\column{21E}{@{}l@{}}%
\column{25}{@{}>{\hspre}l<{\hspost}@{}}%
\column{52}{@{}>{\hspre}l<{\hspost}@{}}%
\column{E}{@{}>{\hspre}l<{\hspost}@{}}%
\>[B]{}\Varid{f}\mathrel{=}{}\<[7]%
\>[7]{}\Varid{linearly}\mathbin{\$}{}\<[E]%
\\
\>[7]{}\hsindent{2}{}\<[9]%
\>[9]{}\mathbf{let}\;{}\<[16]%
\>[16]{}\packbox(\Conid{Ur}\;\Varid{arr})\mathrel{=}\Varid{new}\;\mathrm{10}{}\<[E]%
\\
\>[16]{}\Varid{fr}{}\<[21]%
\>[21]{}\mathbin{::}{}\<[21E]%
\>[25]{}\constraintfont{\Conid{RW}\;\Varid{n}}\Lolly (){}\<[52]%
\>[52]{}\,{}\<[E]%
\\
\>[16]{}\Varid{fr}{}\<[21]%
\>[21]{}\mathrel{=}{}\<[21E]%
\>[25]{}\Varid{free}\;\Varid{arr}{}\<[E]%
\\
\>[16]{}(){}\<[21]%
\>[21]{}\mathrel{=}{}\<[21E]%
\>[25]{}\Varid{fr}\;{}\<[52]%
\>[52]{}\mathbf{in}\;\Conid{Ur}\;(){}\<[E]%
\ColumnHook
\end{hscode}\resethooks
In this example, the programmer meant for
\ensuremath{\Varid{free}} to use the \ensuremath{\constraintfont{\Conid{RW}\;\Varid{n}}} constraint introduced locally in the type
of \ensuremath{\Varid{fr}}. Yet
there are actually two linear \ensuremath{\constraintfont{\Conid{RW}\;\Varid{n}}} constraints: this local one and the
one assumed when unpacking \ensuremath{\Varid{arr}}. The wrong
choice among the constraints will lead the algorithm to fail.
Choosing the first $  \constraintfont{ \ottnt{q} }  $ linear assumption guarantees we get the
most local one.

Another interesting feature of the solver (\Cref{fig:simpl-solver}) is that
no rule solves a linear constraint if it appears both in the
unrestricted and a linear context.
Consider the following (contrived) \textsc{api}:
\begin{hscode}\SaveRestoreHook
\column{B}{@{}>{\hspre}l<{\hspost}@{}}%
\column{E}{@{}>{\hspre}l<{\hspost}@{}}%
\>[B]{}\mathbf{class}\;\constraintfont{\Conid{C}}{}\<[E]%
\ColumnHook
\end{hscode}\resethooks

\vspace{-2ex}\noindent
\begin{minipage}{0.5\linewidth}\begin{hscode}\SaveRestoreHook
\column{B}{@{}>{\hspre}l<{\hspost}@{}}%
\column{3}{@{}>{\hspre}l<{\hspost}@{}}%
\column{E}{@{}>{\hspre}l<{\hspost}@{}}%
\>[3]{}\Varid{giveC}\mathbin{::}(\constraintfont{\Conid{C}}\FatArrow \Conid{Int})\to \Conid{Int}{}\<[E]%
\ColumnHook
\end{hscode}\resethooks
\end{minipage}
\begin{minipage}{0.5\linewidth}\begin{hscode}\SaveRestoreHook
\column{B}{@{}>{\hspre}l<{\hspost}@{}}%
\column{3}{@{}>{\hspre}l<{\hspost}@{}}%
\column{E}{@{}>{\hspre}l<{\hspost}@{}}%
\>[3]{}\Varid{useC}\mathbin{::}\constraintfont{\Conid{C}}\Lolly \Conid{Int}{}\<[E]%
\ColumnHook
\end{hscode}\resethooks
\end{minipage}
\ensuremath{\Varid{giveC}} gives an unrestricted copy of \ensuremath{\constraintfont{\Conid{C}}} to some continuation, while \ensuremath{\Varid{useC}}
uses \ensuremath{\constraintfont{\Conid{C}}} linearly. Now consider two potential consumers of this \textsc{api}:

\noindent
\begin{minipage}{0.5\linewidth}\begin{hscode}\SaveRestoreHook
\column{B}{@{}>{\hspre}l<{\hspost}@{}}%
\column{3}{@{}>{\hspre}l<{\hspost}@{}}%
\column{E}{@{}>{\hspre}l<{\hspost}@{}}%
\>[3]{}\Varid{ambiguous1}\mathbin{::}\constraintfont{\Conid{C}}\Lolly \Conid{Int}{}\<[E]%
\\
\>[3]{}\Varid{ambiguous1}\mathrel{=}\Varid{giveC}\;\Varid{useC}{}\<[E]%
\ColumnHook
\end{hscode}\resethooks
\end{minipage}
\begin{minipage}{0.5\linewidth}\begin{hscode}\SaveRestoreHook
\column{B}{@{}>{\hspre}l<{\hspost}@{}}%
\column{3}{@{}>{\hspre}l<{\hspost}@{}}%
\column{E}{@{}>{\hspre}l<{\hspost}@{}}%
\>[3]{}\Varid{ambiguous2}\mathbin{::}\constraintfont{\Conid{C}}\Lolly (\Conid{Int},\Conid{Int}){}\<[E]%
\\
\>[3]{}\Varid{ambiguous2}\mathrel{=}(\Varid{giveC}\;\Varid{useC},\Varid{useC}){}\<[E]%
\ColumnHook
\end{hscode}\resethooks
\end{minipage}
Looking at \ensuremath{\Varid{ambiguous1}}, the invocation of \ensuremath{\Varid{useC}} has both a linear \ensuremath{\constraintfont{\Conid{C}}} in scope, and
a more local unrestricted \ensuremath{\constraintfont{\Conid{C}}}. The strategy to pick the more local constraint
fails here, because it would leave the linear \ensuremath{\constraintfont{\Conid{C}}} unconsumed. A
tempting refinement might be to always consume the most local \emph{linear}
constraint. That would handle handles \ensuremath{\Varid{ambiguous1}} correctly, but fail on \ensuremath{\Varid{ambiguous2}}. In
the case of the latter, if the invocation of \ensuremath{\Varid{giveC}\;\Varid{useC}} consumes the linear
\ensuremath{\constraintfont{\Conid{C}}}, then the second \ensuremath{\Varid{useC}} invocation will fail. It is possible to
give a type derivation to \ensuremath{\Varid{ambiguous2}} in the qualified type system
of \cref{sec:qualified-type-system} by making the first \ensuremath{\Varid{useC}} consume the
unrestricted \ensuremath{\constraintfont{\Conid{C}}} and the second \ensuremath{\Varid{useC}} consume the linear
\ensuremath{\constraintfont{\Conid{C}}}.
This assignment, however, would require
the constraint solver to guess when solving the constraint from the first \ensuremath{\Varid{useC}}.
Accordingly, in order to both avoid backtracking and to keep type inference
independent of the order terms appear in the program text, \ensuremath{\Varid{bad}} is
rejected.
This introduces incompleteness with respect the entailment
relation. We conjecture that this is the only source of
incompleteness that we introduce beyond what is already in
\textsc{ghc}~\cite[Section~6]{OutsideIn}.

\section{Desugaring}
\label{sec:desugaring}

The semantics of our language is given by desugaring it into
a simpler core language: a variant of the $λ^q$
calculus~\cite{LinearHaskell}. We
define the core language's type system here; its operational semantics
is the same, \emph{mutatis mutandis}, as that of Linear Haskell.

\subsection{The core calculus}
\label{sec:core-calculus}
\label{sec:ds:inferred-constraints}

\begin{figure}
  \maybesmall
  \centering
  $$
  \begin{array}{lcll}
     \sigma  & \bnfeq &   \forall   \overline{\ottmv{a} } .  \tau   & \text{Type schemes} \\
     \tau ,  \upsilon  & \bnfeq & ... \bnfor   \exists   \overline{\ottmv{a} } .  \tau  \otimes  \upsilon   & \text{Types} \\
     \ottnt{e}  & \bnfeq & ... \bnfor  \packbox \, \ottsym{(}  \ottnt{e_{{\mathrm{1}}}}  \ottsym{,}  \ottnt{e_{{\mathrm{2}}}}  \ottsym{)}  \bnfor   \klet\ \packbox (  \ottmv{x}  ,  \ottmv{y}  ) =  \ottnt{e_{{\mathrm{1}}}}  \ \kin \  \ottnt{e_{{\mathrm{2}}}}   & \text{Expressions}
  \end{array}
  $$

  \drules[L]{$\Gamma  \vdash  \ottnt{e}  \ottsym{:}  \tau$}{Core language
    typing}{Pack,Unpack}
  \caption{Core calculus (subset)}
  \label{fig:core-typing-rules}\label{fig:core-grammar}
\end{figure}

The core calculus is a variant of the type system defined in
\cref{sec:qualified-type-system}, but without constraints. That is, the evidence for constraints is passed explicitly in this core calculus.
Following $λ^q$, we assume the existence of the following data types:
\begin{itemize}
\item $ \tau_{{\mathrm{1}}}  \otimes  \tau_{{\mathrm{2}}} $ with sole constructor
  $ ({,})   \ottsym{:}    \forall   \ottmv{a} \, \ottmv{b} .  \ottmv{a}   \to_{  \multiplicityfont{ \ottsym{1} }  }   \ottmv{b}  \to_{  \multiplicityfont{ \ottsym{1} }  }  \ottmv{a}    \otimes  \ottmv{b}$. We will write $\ottsym{(}  \ottnt{e_{{\mathrm{1}}}}  \ottsym{,}  \ottnt{e_{{\mathrm{2}}}}  \ottsym{)}$ for $  ({,})  \, \ottnt{e_{{\mathrm{1}}}} \, \ottnt{e_{{\mathrm{2}}}} $.
\item $  \mathbf{1}  $ with sole constructor $\ottsym{()}  \ottsym{:}   \mathbf{1} $.
\item $ \ottkw{Ur} \, \tau $ with sole constructor $\ottkw{Ur}  \ottsym{:}   \forall   \ottmv{a} .   \ottmv{a}  \to_{  \multiplicityfont{ \omega }  }  \ottkw{Ur} \, \ottmv{a}  $
\end{itemize}
\Cref{fig:core-typing-rules}
highlights the differences from the qualified system:
\begin{itemize}
  \item Type schemes $ \sigma $ do not support qualified types.
  \item Existentially quantified types ($  \exists   \overline{\ottmv{a} } .  \tau  \RLolly   \constraintfont{ \ottnt{Q} }   $) are now represented as an (existentially quantified, linear) pair of values ($  \exists   \overline{\ottmv{a} } .  \tau_{{\mathrm{2}}}  \otimes  \tau_{{\mathrm{1}}}  $).
Accordingly, $\packbox$ operates on pairs.
\end{itemize}
The differences between our core calculus and $λ^q$ are as follows:
\begin{itemize}
\item We do not support multiplicity polymorphism.
\item On the other hand, we do include type polymorphism.
\item Polymorphism is implicit rather than explicit. This is not an
  essential difference, but it simplifies the presentation. We could,
for example, include more details in the terms in order to make type-checking
more obvious; this amounts essentially to an encoding of typing derivations
in the terms\footnote{See, for example, \citet{weirich-icfp17} and their
comparison between an implicit core language D and an explicit one DC.}.
\item We have existential types. These can be realised in regular Haskell as a
  family of datatypes.
\end{itemize}

Using \cref{lem:generation-soundness} together with
\cref{lem:solver-soundness} we know that if
$\Gamma  \vdashi  \ottnt{e}  \ottsym{:}  \tau  \leadsto   \constraintfont{ \constraintfont{C} } $ and $ \constraintfont{ \ottnt{U} }   \ottsym{;}   \constraintfont{ \ottnt{D} }   \ottsym{;}   \constraintfont{ \ottnt{L} }   \vdashs   \constraintfont{ \constraintfont{C} }   \leadsto   \constraintfont{ \emptyset } $, then
$ \constraintfont{ \ottsym{(}  \ottnt{U}  \ottsym{,}   \ottnt{D}  \uplus  \ottnt{L}   \ottsym{)} }   \ottsym{;}  \Gamma  \vdash  \ottnt{e}  \ottsym{:}  \tau$.
It only remains to desugar derivations of $ \constraintfont{ \ottnt{Q} }   \ottsym{;}  \Gamma  \vdash  \ottnt{e}  \ottsym{:}  \tau$ into the
core calculus.

\subsection{From qualified to core}
\label{sec:ds:from-qualified-core}

\subsubsection{Evidence}
In order to desugar derivations of the qualified system to the core calculus,
we pass evidence explicitly\footnote{This technique is also often called
  dictionary-passing style \cite{type-classes-impl} because, in the case of type classes, evidences are
  dictionaries, and because type classes were the original form of constraints
  in Haskell.}.
To do so, we require some more material from
constraints. Namely, we assume a type $  \dsevidence{\constraintfont{ \ottnt{q} } }  $ for each atomic
constraint $  \constraintfont{ \ottnt{q} }  $,
defined in \Cref{fig:evidence}.  The $ \dsevidence{\constraintfont{ \_ } } $ operation
extends to simple constraints as
$  \dsevidence{\constraintfont{ \ottnt{Q} } }  $.
Furthermore, we require that for every $  \constraintfont{ \ottnt{Q_{{\mathrm{1}}}} }  $ and $  \constraintfont{ \ottnt{Q_{{\mathrm{2}}}} }  $
such that $ \constraintfont{ \ottnt{Q_{{\mathrm{1}}}} }   \Vdash   \constraintfont{ \ottnt{Q_{{\mathrm{2}}}} } $, there is a (linear) function
$ \dsevidence{\constraintfont{ \ottnt{Q_{{\mathrm{1}}}} }  \Vdash  \constraintfont{ \ottnt{Q_{{\mathrm{2}}}} } }   \ottsym{:}    \dsevidence{\constraintfont{ \ottnt{Q_{{\mathrm{1}}}} } }   \to_{  \multiplicityfont{ \ottsym{1} }  }   \dsevidence{\constraintfont{ \ottnt{Q_{{\mathrm{2}}}} } }  $.

Let us now define a family of functions $  \dstype{ \_ }  $ to translate
the type schemes, types, contexts, and typing derivations of the qualified system into the
types, type schemes, contexts, and terms of the core calculus.

\subsubsection{Translating types}
Type schemes $ \sigma $ are translated by turning the implicit argument $  \constraintfont{ \ottnt{Q} }  $
into an explicit one of type $  \dsevidence{\constraintfont{ \ottnt{Q} } }  $. Translating types $ \tau $
and contexts $ \Gamma $ proceeds as
expected.

\begin{minipage}{0.5\linewidth}
$$
\left\{
  \begin{array}{lcl}
      \dstype{  \forall   \overline{\ottmv{a} } .   \constraintfont{ \ottnt{Q} }   \Lolly  \tau  }   & = &    \forall   \overline{\ottmv{a} } .   \dsevidence{\constraintfont{ \ottnt{Q} } }    \to_{  \multiplicityfont{ \ottsym{1} }  }   \dstype{ \tau }    \\
  \end{array}
\right.
$$
$$
\left\{
  \begin{array}{lcl}
      \dstype{  \tau_{{\mathrm{1}}}  \to_{  \multiplicityfont{ \pi }  }  \tau_{{\mathrm{2}}}  }   & = &    \dstype{ \tau_{{\mathrm{1}}} }   \to_{  \multiplicityfont{ \pi }  }   \dstype{ \tau_{{\mathrm{2}}} }    \\
      \dstype{  \exists   \overline{\ottmv{a} } .  \tau  \RLolly   \constraintfont{ \ottnt{Q} }   }   & = &   \exists   \overline{\ottmv{a} } .   \dstype{ \tau }   \otimes   \dsevidence{\constraintfont{ \ottnt{Q} } }   
  \end{array}
\right.
$$
\end{minipage}%
\begin{minipage}{0.5\linewidth}
$$
\left\{
  \begin{array}{lcl}
      \dstype{ ∙ }   &=&  ∙  \\
      \dstype{ \Gamma  \ottsym{,}   \ottmv{x} {:}_{  \multiplicityfont{ \pi }  } \tau  }   &=&   \dstype{ \Gamma }   \ottsym{,}   \ottmv{x} {:}_{  \multiplicityfont{ \pi }  }  \dstype{ \tau }   
  \end{array}
\right.
$$
\end{minipage}

\subsubsection{Translating terms}
Given a derivation $ \constraintfont{ \ottnt{Q} }   \ottsym{;}  \Gamma  \vdash  \ottnt{e}  \ottsym{:}  \tau$, we can build an expression
$ \dsterm{ \ottmv{z} }{  \constraintfont{ \ottnt{Q} }   \ottsym{;}  \Gamma  \vdash  \ottnt{e}  \ottsym{:}  \tau } $, such that
$ \dstype{ \Gamma }   \ottsym{,}   \ottmv{z} {:}_{  \multiplicityfont{ \ottsym{1} }  }  \dsevidence{\constraintfont{ \ottnt{Q} } }    \vdash   \dsterm{ \ottmv{z} }{  \constraintfont{ \ottnt{Q} }   \ottsym{;}  \Gamma  \vdash  \ottnt{e}  \ottsym{:}  \tau }   \ottsym{:}   \dstype{ \tau } $ (for some fresh variable
$ \ottmv{z} $). Even though we abbreviate the derivation as only its
concluding judgement, the translation is defined recursively on the
whole typing derivation: in particular, we have access to typing rule
premises in the body of the definition.
We present some of the interesting cases in \Cref{fig:desugaring}.
\begin{figure}
    \maybesmall
\centering
  \begin{subfigure}{0.3\linewidth}%
$$
\left\{
  \begin{array}{lcl}
      \dsevidence{\constraintfont{   \multiplicityfont{ \ottsym{1} }  \scale \ottnt{q}  } }   & = &   \dsevidence{\constraintfont{ \ottnt{q} } }   \\
      \dsevidence{\constraintfont{   \multiplicityfont{ \omega }  \scale \ottnt{q}  } }   & = &  \ottkw{Ur} \, \ottsym{(}   \dsevidence{\constraintfont{ \ottnt{q} } }   \ottsym{)}  \\
      \dsevidence{\constraintfont{  \mathbf{\varepsilon}  } }   & = &   \mathbf{1}   \\
      \dsevidence{\constraintfont{ \ottnt{Q_{{\mathrm{1}}}}  \qtensor  \ottnt{Q_{{\mathrm{2}}}} } }   & = &   \dsevidence{\constraintfont{ \ottnt{Q_{{\mathrm{1}}}} } }   \otimes   \dsevidence{\constraintfont{ \ottnt{Q_{{\mathrm{2}}}} } }  
  \end{array}
\right.
$$
  \caption{Evidence passing}
  \label{fig:evidence}
  \end{subfigure}\hfill
  \begin{subfigure}{0.7\linewidth}%
$$
\left\{
  \;
  \begin{minipage}{0.5\linewidth}
\begin{hscode}\SaveRestoreHook
\column{B}{@{}>{\hspre}l<{\hspost}@{}}%
\column{3}{@{}>{\hspre}l<{\hspost}@{}}%
\column{5}{@{}>{\hspre}l<{\hspost}@{}}%
\column{22}{@{}>{\hspre}c<{\hspost}@{}}%
\column{22E}{@{}l@{}}%
\column{25}{@{}>{\hspre}l<{\hspost}@{}}%
\column{52}{@{}>{\hspre}c<{\hspost}@{}}%
\column{52E}{@{}l@{}}%
\column{E}{@{}>{\hspre}l<{\hspost}@{}}%
\>[B]{}\dsterm{\ottmv{\Varid{z}}}{  \constraintfont{ \ottnt{Q} }  ; \Gamma  \vdash \Varid{x}\mathbin{:} \upsilon [ \overline{\tau} / \overline{\ottmv{a} } ]}\mathrel{=}\Varid{x}\;\Varid{z}{}\<[E]%
\\
\>[B]{}\dsterm{\ottmv{\Varid{z}}}{  \constraintfont{ \ottnt{Q} }   \qtensor   \constraintfont{ \ottnt{Q_{{\mathrm{1}}}} }  [ \overline{\upsilon} / \overline{\ottmv{a} } ]; \Gamma  \vdash \packbox\Varid{e}\mathbin{:}\exists\; \overline{\ottmv{a} } .\tau\RLolly  \constraintfont{ \ottnt{Q_{{\mathrm{1}}}} }  }\mathrel{=}{}\<[E]%
\\
\>[B]{}\hsindent{3}{}\<[3]%
\>[3]{} \kcase_  \multiplicityfont{ \ottsym{1} }  \;\Varid{z}\;\mathbf{of}\;\{\mskip1.5mu (\Varid{z'},\Varid{z''})\to {}\<[E]%
\\
\>[3]{}\hsindent{2}{}\<[5]%
\>[5]{}\packbox(\Varid{z''},\dsterm{\ottmv{\Varid{z'}}}{  \constraintfont{ \ottnt{Q} }  ; \Gamma  \vdash \Varid{e}\mathbin{:}\tau[ \overline{\upsilon} / \overline{\ottmv{a} } ]})\mskip1.5mu\}{}\<[E]%
\\
\>[B]{}\dsterm{\ottmv{\Varid{z}}}{  \constraintfont{ \ottnt{Q_{{\mathrm{1}}}} }   \qtensor   \constraintfont{ \ottnt{Q_{{\mathrm{2}}}} }  ; \Gamma_{{\mathrm{1}}} \ottsym{+} \Gamma_{{\mathrm{2}}}  \vdash \klet\ \packbox \Varid{x}\mathrel{=} \ottnt{e_{{\mathrm{1}}}} \;\mathbf{in}\; \ottnt{e_{{\mathrm{2}}}} \mathbin{:}\tau}{}\<[52]%
\>[52]{}\mathrel{=}{}\<[52E]%
\\
\>[B]{}\hsindent{3}{}\<[3]%
\>[3]{} \kcase_  \multiplicityfont{ \ottsym{1} }  \;\Varid{z}\;\mathbf{of}\;\{\mskip1.5mu ( \ottmv{z_{{\mathrm{1}}}} , \ottmv{z_{{\mathrm{2}}}} )\to {}\<[E]%
\\
\>[3]{}\hsindent{2}{}\<[5]%
\>[5]{}\klet\ \packbox \Varid{z'},\Varid{x}\mathrel{=}\dsterm{\ottmv{ \ottmv{z_{{\mathrm{1}}}} }}{  \constraintfont{ \ottnt{Q_{{\mathrm{1}}}} }  ; \Gamma_{{\mathrm{1}}}  \vdash  \ottnt{e_{{\mathrm{1}}}} \mathbin{:}\exists\; \overline{\ottmv{a} } . \tau_{{\mathrm{1}}} \RLolly  \constraintfont{ \ottnt{Q} }  }\;\mathbf{in}{}\<[E]%
\\
\>[3]{}\hsindent{2}{}\<[5]%
\>[5]{} \klet_  \multiplicityfont{ \ottsym{1} }  \; \ottmv{z_{{\mathrm{2}}}} '\mathrel{=}( \ottmv{z_{{\mathrm{2}}}} ,\Varid{z'})\;\mathbf{in}{}\<[E]%
\\
\>[3]{}\hsindent{2}{}\<[5]%
\>[5]{}\dsterm{\ottmv{ \ottmv{z_{{\mathrm{2}}}} '}}{  \constraintfont{ \ottnt{Q_{{\mathrm{2}}}} }   \qtensor   \constraintfont{ \ottnt{Q} }  ; \Gamma_{{\mathrm{2}}} ,\Varid{x}\mathop{:_{1}} \tau_{{\mathrm{1}}}  \vdash  \ottnt{e_{{\mathrm{2}}}} \mathbin{:}\tau}\mskip1.5mu\}{}\<[E]%
\\
\>[B]{}\dsterm{\ottmv{\Varid{z}}}{  \constraintfont{ \ottnt{Q} }  ; \Gamma  \vdash \Varid{e}\mathbin{:}\tau}{}\<[22]%
\>[22]{}\mathrel{=}{}\<[22E]%
\>[25]{}\mbox{\onelinecomment  \rref{E-Sub}}{}\<[E]%
\\
\>[B]{}\hsindent{3}{}\<[3]%
\>[3]{} \klet_  \multiplicityfont{ \ottsym{1} }  \;\Varid{z'}\mathrel{=}\dsevidence{  \constraintfont{ \ottnt{Q} }   \Vdash   \constraintfont{ \ottnt{Q_{{\mathrm{1}}}} }  }\;\Varid{z}\;\mathbf{in}\;\dsterm{\ottmv{\Varid{z'}}}{  \constraintfont{ \ottnt{Q_{{\mathrm{1}}}} }  ; \Gamma  \vdash \Varid{e}\mathbin{:}\tau}{}\<[E]%
\\
\>[B]{}\mathbin{...}{}\<[E]%
\ColumnHook
\end{hscode}\resethooks
  \end{minipage}
\right.
$$
  \caption{Desugaring (subset)}
  \label{fig:desugaring}
  \end{subfigure}
  \caption{Evidence passing and desugaring}
\end{figure}

The cases correspond to the~\rref*{E-Var},~\rref*{E-Unpack}\footnote{The attentive
reader may note that the case for $\kunpack$ extracts out $  \constraintfont{ \ottnt{Q_{{\mathrm{1}}}} }  $ and $  \constraintfont{ \ottnt{Q_{{\mathrm{2}}}} }  $
from the provided simple constraint. Given that simple constraints $  \constraintfont{ \ottnt{Q} }  $ have no
internal ordering and allow duplicates (in the non-linear component), this splitting
is not well defined. To fix this, an implementation would have to \emph{name} individual
components of $  \constraintfont{ \ottnt{Q} }  $, and then the typing derivation can indicate which constraints
go with which sub-expression. Happily, \textsc{ghc} \emph{already} names its constraints,
and so this approach fits easily in the implementation. We could also augment our formalism
here with these details, but they add clutter with little insight.}, and~\rref*{E-Sub} rules, respectively.
Variables are stored with qualified types in the environment, so they get
translated to functions that take the evidence as argument. Accordingly, the evidence
is inserted by passing $ \ottmv{z} $ as an argument.
Handling \rref*{E-Unpack} requires splitting the context into two: $ \ottnt{e_{{\mathrm{1}}}} $ is desugared as a pair, and the evidence
it contains is passed to $ \ottnt{e_{{\mathrm{2}}}} $. Finally, subsumption summons the function corresponding to the entailment relation $ \constraintfont{ \ottnt{Q} }   \Vdash   \constraintfont{ \ottnt{Q_{{\mathrm{1}}}} } $
and applies it to $ \ottmv{z} $ : $  \dsevidence{\constraintfont{ \ottnt{Q} } }  $ then proceeds to desugar $ \ottnt{e} $ with the resulting evidence for $  \constraintfont{ \ottnt{Q_{{\mathrm{1}}}} }  $.
Crucially, since $  \dsterm{ \ottmv{z} }{ \_ }  $ is defined on \emph{derivations}, we can access the premises used in the rule.
Namely, $  \constraintfont{ \ottnt{Q_{{\mathrm{1}}}} }  $ is available in this last case from the~\rref*{E-Sub} rule's premise.

It is straightforward by induction, to verify that desugaring is correct:
\begin{theorem}[Desugaring]
If $ \constraintfont{ \ottnt{Q} }   \ottsym{;}  \Gamma  \vdash  \ottnt{e}  \ottsym{:}  \tau$, then
$ \dstype{ \Gamma }   \ottsym{,}   \ottmv{z} {:}_{  \multiplicityfont{ \ottsym{1} }  }  \dsevidence{\constraintfont{ \ottnt{Q} } }    \vdash   \dsterm{ \ottmv{z} }{  \constraintfont{ \ottnt{Q} }   \ottsym{;}  \Gamma  \vdash  \ottnt{e}  \ottsym{:}  \tau }   \ottsym{:}   \dstype{ \tau } $, for any fresh
variable $ \ottmv{z} $.
\end{theorem}

Thanks to the desugaring machinery, the semantics of a language with linear
constraints can be understood in terms of a simple core language with linear
types, such as $λ^q$, or indeed, \textsc{ghc} Core.

\section{Integrating into \textsc{ghc}}

One of the guiding principles behind our design was ease of integration with
modern Haskell. In this section we describe some of the particulars of adding
linear constraints to \textsc{ghc}.

\subsection{Implementation}
\label{sec:implementation}

We have written a prototype implementation~\cite{prototype} of linear constraints on top of \textsc{ghc} 9.1, a version that already
ships with the \text{\ttfamily LinearTypes} extension. Function arrows (\ensuremath{\to }) and context arrows
(\ensuremath{\FatArrow }) share the same internal representation in the typechecker, differentiated
only by a boolean flag. Thus, the \text{\ttfamily LinearTypes} implementation effort has already
laid down the bureaucratic ground work of annotating these arrows with
multiplicity information.

The key changes affect constraint generation and constraint solving. Constraints
are now annotated with a multiplicity, according to the context from
which they arise. With \text{\ttfamily LinearTypes}, \textsc{ghc} already scales the usage
of term variables. We simply modified the scaling function to capture all the
generated constraints and re-emit a scaled version of them, which is a fairly local
change.

The constraint solver maintains a set of given constraints (the \emph{inert set}
in \textsc{ghc} jargon), which corresponds to the $  \constraintfont{ \ottnt{U} }  $, $  \constraintfont{ \ottnt{D} }  $, and $  \constraintfont{ \ottnt{L} }  $
contexts in our solver judgements in \cref{sec:constraint-solver}. When
the solver goes under an implication, the assumptions of the implication are
added to set of givens. When a new given is added, we record the \emph{level} of
the implication (how many implications deep the constraint arises from) along
with the constraint. So that in case there are multiple matching
givens, the constraint solver selects the innermost one
(in \cref{sec:constraint-solver} we use an ordered list
for this purpose).

As constraint solving proceeds, the compiler pipeline constructs a
term in a typed language known as \textsc{ghc} Core~\cite{system-fc}.
In Core, type class constraints are turned into explicit evidence (see
\cref{sec:desugaring}). Thanks to being fully annotated, Core has
decidable typechecking, which is used to find and fix bugs in
the compiler (the Haskell type checker finds mistakes in user
programs). Thus, the Core typechecker verifies that the desugaring
procedure produced a linearity-respecting program before code
generation occurs.

\subsection{Interaction with other features}

Since constraints play an important role in \textsc{ghc}'s type system, we must
pay close attention to the interaction of linearity with other language features
related to constraints. Of these, we point out two that require some extra care.
\info{There isn't room to properly explain how we can implement
\ensuremath{\constraintfont{\Conid{Linearly}}} constraints. We don't already speak of
desugaring in the implementation section, so I'd need to give a bit of
context about wrappers or something.

The right approach is to count the number of `Linearly` used. We need
wrappers at the toplevel of definitions, and at each branch of a match
(even if it doesn't introduce an implication, as we may need to adjust
the number of `Linearly` in some branches (presumably by
weakening)). And we do the appropriate amount of
duplications/discards in these wrappers.}

\subsubsection{Superclasses}

Haskell's type classes can have \emph{superclasses}, which place constraints on
all of the instances of that class. For example, the \ensuremath{\Conid{Ord}} class is defined as
\begin{hscode}\SaveRestoreHook
\column{B}{@{}>{\hspre}l<{\hspost}@{}}%
\column{E}{@{}>{\hspre}l<{\hspost}@{}}%
\>[B]{}\mathbf{class}\;\constraintfont{\Conid{Eq}\;\Varid{a}}\FatArrow \constraintfont{\Conid{Ord}\;\Varid{a}}\;\mathbf{where}\mathbin{...}{}\<[E]%
\ColumnHook
\end{hscode}\resethooks
which means that every ordered type must also support equality. Such
superclass declarations extend the entailment relation: if we know that a type
is ordered, we also know that it supports equality. This is troublesome if we
have a linear occurrence of \ensuremath{\constraintfont{\Conid{Ord}\;\Varid{a}}}, because then using this entailment, we could
conclude that a linear constraint (\ensuremath{\constraintfont{\Conid{Ord}\;\Varid{a}}}) implies an unrestricted constraint
(\ensuremath{\constraintfont{\Conid{Eq}\;\Varid{a}}}), which violates \cref{lem:q:scaling-inversion}.

But even linear superclass constraints cause trouble. Consider a version of \ensuremath{\constraintfont{\Conid{Ord}\;\Varid{a}}}
that has \ensuremath{\constraintfont{\Conid{Eq}\;\Varid{a}}} as a linear superclass.
\begin{hscode}\SaveRestoreHook
\column{B}{@{}>{\hspre}l<{\hspost}@{}}%
\column{E}{@{}>{\hspre}l<{\hspost}@{}}%
\>[B]{}\mathbf{class}\;\constraintfont{\Conid{Eq}\;\Varid{a}}\Lolly \constraintfont{\Conid{Ord}\;\Varid{a}}\;\mathbf{where}\mathbin{...}{}\<[E]%
\ColumnHook
\end{hscode}\resethooks
When given a linear \ensuremath{\constraintfont{\Conid{Ord}\;\Varid{a}}}, should we keep it as \ensuremath{\constraintfont{\Conid{Ord}\;\Varid{a}}}, or rewrite to
\ensuremath{\constraintfont{\Conid{Eq}\;\Varid{a}}} using the entailment? Short of backtracking, the constraint solver
needs to make a guess, which \textsc{ghc} never does.

To address both of these issues at once, we make the following rule: the
superclasses of a linear constraint are ignored.

\subsubsection{Equality constraints}
\label{sec:equality-constraints}

In \cref{sec:type-inference} we argued that \emph{type} inference and
\emph{constraint} inference can be performed independently. However, this is not
the case for \textsc{ghc}'s constraint domain, because it supports equality
constraints, which allows unification problems to be deferred, and potentially
be solvable only after solving other constraints first.

To reconcile this with our presentation, we need to ensure that
\emph{unrestricted constraint} inference and \emph{linear constraint} inference
can be performed independently. That is, solving a linear constraint should
never be required for solving an unrestricted constraint. This is ensured by
\cref{lem:q:scaling-inversion}.

They key is to represent unification problems as \emph{unrestricted} equality
constraints, so a given linear equality constraint cannot be used during type
inference.  This way, linear equalities require no
special treatment, and are harmless.

\subsection{Inferring packing and unpacking}
\label{sec:implicit-existentials}
Recent work~\cite{existentials} describes an algorithm (call it \textsc{edwl}, after the
authors' names) that
can infer the location of the pack and unpack annotations (our $\packbox$ and $\kunpack$)
in a program.%
\footnote{Actually, \citet{existentials} use an $\ottkw{open}$ construct instead of $\kunpack$
to access the contents of an existential package, but that distinction does
not affect our usage of existentials with linear constraints.}
In Section~9.2 of that paper,
the authors extend their system to include class constraints,
much as we allow our existential packages to carry linear constraints.

Accordingly, \textsc{edwl} would work well for us here and remove the need for these annotations.
The \textsc{edwl} algorithm is only a small change on the way some types are treated during
bidirectional type-checking. Though the presentation of linear constraints is not
written using a bidirectional algorithm, our implementation in \textsc{ghc}
is indeed bidirectional (as \textsc{ghc}'s existing type inference algorithm
is bidirectional, as described by \citet{practical-type-inference} and
\citet{visible-type-application}) and produces constraints much like we
have presented here, formally. None of this would change in adapting \textsc{edwl}.
Indeed, it would seem that the two extensions are orthogonal in implementation,
though avoiding the need for explicit packing and unpacking would
make linear constraints easier to use.

\section{Related work}
\label{sec:related-work}

\paragraph{OutsideIn}
\label{sec:outsidein}

Our aim is to integrate the present work in \textsc{ghc}, and
accordingly the qualified type system in
\cref{sec:qualified-type-system} and the constraint inference
algorithm in \cref{sec:type-inference} follow a similar
presentation to that of OutsideIn~\cite{OutsideIn}, \textsc{ghc}'s
constraint solver algorithm.  Even though our presentation is
self-contained, we outline some of the differences from that work.

The solver judgement in OutsideIn takes the following form:
\[\mathcal{Q}\ ;\ Q_{\mathit{given}}\ ;\ \overline{\alpha}_{\mathit{tch}} \overset{\mathit{solv}}{\mapsto} C_{\mathit{wanted}} \leadsto Q_{\mathit{residual}}\  ; \ \theta\]
The main differences from our solver judgement in \cref{sec:constraint-solver} are:
\begin{itemize}
  \item OutsideIn's judgement includes top-level axioms schemes separately
($\mathcal{Q}$), which we have omitted for the sake of brevity and are instead
included in $Q_{\mathit{given}}$.
  \item We present the \emph{given} constraints ($Q_{\mathit{given}}$ in OutsideIn) as two separate
constraint sets $  \constraintfont{ \ottnt{U} }  $ and $  \constraintfont{ \ottnt{L} }  $, standing for the unrestricted and linear
parts respectively.
  \item In addition to constraint inference, OutsideIn performs type
inference, requiring additional bookkeeping in the solver judgment. The solver
takes as input a set of \emph{touchable} variables $\overline{\alpha}_{tch}$
which record the type variables that can be unified at any given time, and
produces a type substitution $\theta$ as an output.
As discussed in \cref{sec:type-inference}, we do not perform type
inference, only constraint inference. Therefore, our solver need not return a
type assignment.
  \item Both
OutsideIn and our solver output a set of constraints, $Q_{\mathit{residual}}$ and
$  \constraintfont{ \ottnt{L}_{\ottmv{o}} }  $ respectively. However, the meaning of these contexts is different.
OutsideIn's \emph{residual} constraints $Q_{\mathit{residual}}$
correspond to the part of $C_{\mathit{wanted}}$ that could not be solved from the
assumptions. These residuals are then quantified over in the generalisation step
of the inference algorithm. We omit these residuals, which means that our
algorithm cannot infer qualified types.
Our \emph{output} constraints $  \constraintfont{ \ottnt{L}_{\ottmv{o}} }  $ instead correspond to the part of the
\emph{linear} givens $  \constraintfont{ \ottnt{L}_{\ottmv{i}} }  $ that were not used in the solution for $  \constraintfont{ \constraintfont{C_{\ottmv{w}}} }  $.

\item Finally, while OutsideIn has a single kind of conjunction, our constraint
language requires two: $  \constraintfont{ \ottnt{Q_{{\mathrm{1}}}}  \qtensor  \ottnt{Q_{{\mathrm{2}}}} }  $ and $  \constraintfont{ \ottnt{Q_{{\mathrm{1}}}}  \aand  \ottnt{Q_{{\mathrm{2}}}} }  $. This shows up when
generating constraints for $\kcase$ expressions in the~\rref{G-Case} rule.
OutsideIn accumulates constraints across branches (taking the union of each
branch), whereas we need to make sure that each branch of a $\kcase$-expression
consumes the same constraints.
\end{itemize}

\paragraph{Ownership}

%
Ownership and borrowing are the key features of Rust's safe memory management model.
In \cref{sec:memory-ownership} we show how linear constraints can be used to
implement such an ownership model as a library.
Although linear constraints do not have the
convenience of Rust's syntax, we expect that they will support a
greater variety of abstractions.

Clean is another language with built-in ownership typing. Like Haskell
it is a lazy language. Mutation is performed by returning a new
reference, like in Linear Haskell without linear constraints.

\paragraph{Languages with capabilities}

The idea of using capabilities to enforce high-level resource usage protocols is not new~\citep{DBLP:conf/pldi/DeLineF01},
and as such has been applied in practical programming languages before.
Both Mezzo~\cite{mezzo-permissions} and
\textsc{ats}~\cite{AtsLinearViews} served as inspiration for the
design of linear constraints. Of the two, Mezzo is more specialised,
being entirely built around its system of capabilities.  \textsc{Ats}
is the closest to our system because it appeals explicitly to linear
logic, and because the capabilities (known as \emph{stateful views})
are not tied to a particular use case. However,
\textsc{ats} does not have full inference of capabilities.

Other than that, the two systems have a lot of similarities. They have a
finer-grained capability system than is expressible in Rust (or our
encoding of it in \cref{sec:memory-ownership}) which makes it possible to change
the type of a reference cell upon write (though linear constraints could be used to implement such type-changing references too). They also eschew scoped
borrowing in favour of more traditional read and write capabilities.

Linear constraints are more general than either Mezzo or \textsc{ats},
while maintaining a considerably simpler inference algorithm, and at
the same time supporting a richer set of constraints (such as \textsc{gadt}s). This
simplicity is a benefit of abstracting over the simple-constraint
domain. In fact, it should be possible to see Mezzo or \textsc{ats} as
particular instantiations of the simple-constraint domain, with linear
constraints providing the general inference mechanism.


\paragraph{Linearly typed languages}

Affe~\cite{kindly-bent} is a linearly typed \textsc{ml}-style core
language with mutable references and arrays, augmented with a notion
of borrowing. It has dedicated syntax for the scope of borrows. In
contrast, we represent scopes as functions. Affe is presented as a
fully integrated solution, while linear constraints is a small layer
on top of Linear Haskell.

\paragraph{Logic programming}

There are a lot of commonalities between \textsc{ghc}'s constraint and logic
programs. Traditional type classes can be seen as Horn clause programs, much
like Prolog programs. \textsc{ghc} puts further restrictions in order to
avoid backtracking for speed and predictability.

The recent addition of quantified
constraints~\cite{quantified-constraints} extends type class
resolution to Hereditary Harrop~\cite{hereditary-harrop} programs. A generalisation of the
Hereditary Harrop fragment to linear logic, described by~\citet{hh-ll},
is the foundation of the Lolli language~\cite{hodas-thesis-lolli}.
The authors also coin the notion of \emph{uniform} proof. A fragment where
uniform proofs are complete supports goal-oriented proof search, like
Prolog does.

Completeness of uniform proofs is equivalent to
\cref{lem:inversion}, which, in turn, is used in the proof of the
soundness \cref{lem:generation-soundness}. Therefore our linear
constraints are compatible with quantified constraints: we simply need
to adapt~\cref{lem:inversion}.

It is interesting that goal-oriented search is baked into the
definition of OutsideIn. It's not only used as the constraint solving
strategy, but it seems to required for the soundness of the constraint
generation algorithm. Or, if they are not required, uniform proofs are
at least an effective strategy to prove soundness.



\section{Conclusion}
\label{sec:conclusion}

We showed how a simple linear type system like that of Linear
Haskell can be extended with an inference mechanism which lets the
compiler manage some of the additional complexity of linear types
instead of the programmer. Linear constraints narrow the gap between linearly
typed languages and dedicated linear-like typing disciplines such as Rust's,
Mezzo's, or \textsc{ats}'s.




\printbibliography
\newpage

\appendix

\section{Full descriptions}
\label{sec:appendix:full-descriptions}

In this appendix, we give, for reference, complete descriptions of the
type systems, functions, etc. that we have abbreviated in the main 
body of the article.

\subsection{Core calculus}
\label{sec:appendix:core-calculus}

This is the complete version of the core calculus described in
\cref{sec:core-calculus}. The full grammar is given by
\Cref{fig:full:core-grammar} and the type system by
\Cref{fig:full:core-typing-rules}.

\begin{figure}
  \maybesmall
  \centering
  $$
  \begin{array}{lcll}
     \ottmv{a} ,  \ottmv{b}  & \bnfeq & \ldots & \text{Type variables} \\
     \ottmv{x} ,  \ottmv{y}  & \bnfeq & \ldots & \text{Expression variables} \\
     \ottmv{K}  & \bnfeq & \ldots & \text{Data constructors} \\
     \sigma  & \bnfeq &   \forall   \overline{\ottmv{a} } .  \tau   & \text{Type schemes} \\
     \tau ,  \upsilon  & \bnfeq &  \ottmv{a}  \bnfor   \exists   \overline{\ottmv{a} } .  \tau  \otimes  \upsilon   \bnfor   \tau_{{\mathrm{1}}}  \to_{  \multiplicityfont{ \pi }  }  \tau_{{\mathrm{2}}}  
                            \bnfor  \ottmv{T} \, \overline{\tau}  & \text{Types} \\
     \Gamma ,  \Delta  & \bnfeq &  ∙  \bnfor  \Gamma  \ottsym{,}   \ottmv{x} {:}_{  \multiplicityfont{ \pi }  } \sigma   &
                                                              \text{Contexts} \\
     \ottnt{e}  & \bnfeq &  \ottmv{x}  \bnfor  \ottmv{K}  \bnfor  \lambda  \ottmv{x}  \ottsym{.}  \ottnt{e}  \bnfor \ottnt{e_{{\mathrm{1}}}} \, \ottnt{e_{{\mathrm{2}}}} \bnfor  \packbox \, \ottsym{(}  \ottnt{e_{{\mathrm{1}}}}  \ottsym{,}  \ottnt{e_{{\mathrm{2}}}}  \ottsym{)}  & \text{Expressions}\\
                 &\bnfor &  \klet\ \packbox (  \ottmv{y}  ,  \ottmv{x}  ) =  \ottnt{e_{{\mathrm{1}}}}  \ \kin \  \ottnt{e_{{\mathrm{2}}}}  \bnfor   \kcase_  \multiplicityfont{ \pi }   \, \ottnt{e} \, \ottkw{of} \, \ottsym{\{}  \overline{\ottmv{K}_i\ \overline{\ottmv{x}_i } \to \ottnt{e}_i }  \ottsym{\}}  &\\
                 &\bnfor &  \klet_  \multiplicityfont{ \pi }   \, \ottmv{x}  \ottsym{=}  \ottnt{e_{{\mathrm{1}}}} \, \ottkw{in} \, \ottnt{e_{{\mathrm{2}}}} \bnfor   \klet_  \multiplicityfont{ \pi }   \, \ottmv{x}  \ottsym{:}  \sigma  \ottsym{=}  \ottnt{e_{{\mathrm{1}}}} \, \ottkw{in} \, \ottnt{e_{{\mathrm{2}}}}  &
  \end{array}
  $$
  \caption{Grammar of the core calculus}
  \label{fig:full:core-grammar}
\end{figure}

\begin{figure}
  \maybesmall
  \centering
  \drules[L]{$\Gamma  \vdash  \ottnt{e}  \ottsym{:}  \tau$}{Core language
    typing}{Var,Abs,App,Pack,Unpack,Let,Case}
  \caption{Core calculus type system}
  \label{fig:full:core-typing-rules}
\end{figure}

\subsection{Desugaring}
\label{sec:appendix:desugaring}

The complete definition of the desugaring function from
\cref{sec:desugaring} can be found in
\Cref{fig:full:desugaring}.

For the sake of concision, we allow ourselves to write nested patterns
in $\kcase$ expressions of the core language. Desugaring nested patterns
into atomic $\kcase$ expression is routine.

In the complete description, we use a device which was omitted in the
main body of the article. Namely, we'll need a way to turn every
$  \dsevidence{\constraintfont{   \multiplicityfont{ \omega }  \scale \ottnt{Q}  } }  $ into an $ \ottkw{Ur} \, \ottsym{(}   \dsevidence{\constraintfont{ \ottnt{Q} } }   \ottsym{)} $. For any
$\ottnt{e}  \ottsym{:}   \dsevidence{\constraintfont{   \multiplicityfont{ \omega }  \scale \ottnt{Q}  } } $, we shall write $ \underline{ \ottnt{e} }_{  \constraintfont{ \ottnt{Q} }  }   \ottsym{:}  \ottkw{Ur} \, \ottsym{(}   \dsevidence{\constraintfont{   \multiplicityfont{ \omega }  \scale \ottnt{Q}  } }   \ottsym{)}$. As a shorthand, particularly useful in nested
patterns, we will write $  \kcase_  \multiplicityfont{ \pi }   \, \ottnt{e} \, \ottkw{of} \, \ottsym{\{}   \underline{ \ottmv{x} }_{  \constraintfont{ \ottnt{Q} }  }   \to  \ottnt{e'}  \ottsym{\}} $ for
$  \kcase_  \multiplicityfont{ \pi }   \,  \underline{ \ottnt{e} }_{  \constraintfont{ \ottnt{Q} }  }  \, \ottkw{of} \, \ottsym{\{}  \ottkw{Ur} \, \ottmv{x}  \to  \ottnt{e'}  \ottsym{\}} $.
$$
\left\{
  \begin{array}{lcl}
      \underline{ \ottnt{e} }_{  \constraintfont{  \mathbf{\varepsilon}  }  }  & = &   \kcase_  \multiplicityfont{ \ottsym{1} }   \, \ottnt{e} \, \ottkw{of} \, \ottsym{\{}  \ottsym{()}  \to  \ottkw{Ur} \, \ottsym{()}  \ottsym{\}}  \\
      \underline{ \ottnt{e} }_{  \constraintfont{   \multiplicityfont{ \ottsym{1} }  \scale \ottnt{q}  }  }   & = &  \ottnt{e}  \\
      \underline{ \ottnt{e} }_{  \constraintfont{   \multiplicityfont{ \omega }  \scale \ottnt{q}  }  }   & = &   \kcase_  \multiplicityfont{ \ottsym{1} }   \, \ottnt{e} \, \ottkw{of} \, \ottsym{\{}  \ottkw{Ur} \, \ottmv{x}  \to  \ottkw{Ur} \, \ottsym{(}  \ottkw{Ur} \, \ottmv{x}  \ottsym{)}  \ottsym{\}}  \\
      \underline{ \ottnt{e} }_{  \constraintfont{ \ottnt{Q_{{\mathrm{1}}}}  \qtensor  \ottnt{Q_{{\mathrm{2}}}} }  }   & = &   \kcase_  \multiplicityfont{ \ottsym{1} }   \, \ottnt{e} \, \ottkw{of} \, \ottsym{\{}  \ottsym{(}   \underline{ \ottmv{x} }_{  \constraintfont{ \ottnt{Q_{{\mathrm{1}}}} }  }   \ottsym{,}   \underline{ \ottmv{y} }_{  \constraintfont{ \ottnt{Q_{{\mathrm{2}}}} }  }   \ottsym{)}  \to  \ottkw{Ur} \, \ottsym{(}  \ottmv{x}  \ottsym{,}  \ottmv{y}  \ottsym{)}  \ottsym{\}} 
  \end{array}
\right.
$$
We will omit the $  \constraintfont{ \ottnt{Q} }  $ in $  \underline{ \ottnt{e} }_{  \constraintfont{ \ottnt{Q} }  }  $ and write
$  \underline{ \ottnt{e} }  $ when it can be easily inferred from the context.

\begin{figure}
  \small
  \centering

$$
\left\{\;
\begin{minipage}{0.8\linewidth}
\begin{hscode}\SaveRestoreHook
\column{B}{@{}>{\hspre}l<{\hspost}@{}}%
\column{3}{@{}>{\hspre}l<{\hspost}@{}}%
\column{5}{@{}>{\hspre}l<{\hspost}@{}}%
\column{7}{@{}>{\hspre}l<{\hspost}@{}}%
\column{29}{@{}>{\hspre}c<{\hspost}@{}}%
\column{29E}{@{}l@{}}%
\column{53}{@{}>{\hspre}c<{\hspost}@{}}%
\column{53E}{@{}l@{}}%
\column{67}{@{}>{\hspre}c<{\hspost}@{}}%
\column{67E}{@{}l@{}}%
\column{E}{@{}>{\hspre}l<{\hspost}@{}}%
\>[B]{}\dsterm{\ottmv{\Varid{z}}}{  \constraintfont{ \ottnt{Q} }  ; \Gamma  \vdash \Varid{x}\mathbin{:} \upsilon \;[\mskip1.5mu  \overline{\tau} \mathbin{/} \overline{\ottmv{a} } \mskip1.5mu]}{}\<[29]%
\>[29]{}\mathrel{=}{}\<[29E]%
\\
\>[B]{}\hsindent{5}{}\<[5]%
\>[5]{}\Varid{x}\;\Varid{z}{}\<[E]%
\\
\>[B]{}\dsterm{\ottmv{\Varid{z}}}{  \constraintfont{ \ottnt{Q} }  ; \Gamma  \vdash \lambda \Varid{x}.\Varid{e}\mathbin{:} \tau_{{\mathrm{1}}} \to_{\pi} \tau_{{\mathrm{2}}} }\mathrel{=}{}\<[E]%
\\
\>[B]{}\hsindent{3}{}\<[3]%
\>[3]{}\lambda \Varid{x}.\dsterm{\ottmv{\Varid{z}}}{  \constraintfont{ \ottnt{Q} }  ; \Gamma ,\Varid{x}\mathop{:_{  \multiplicityfont{ \pi }  }} \tau_{{\mathrm{1}}}  \vdash \Varid{e}\mathbin{:} \tau_{{\mathrm{2}}} }{}\<[E]%
\\
\>[B]{}\dsterm{\ottmv{\Varid{z}}}{  \constraintfont{ \ottnt{Q_{{\mathrm{1}}}} }   \qtensor   \constraintfont{ \ottnt{Q_{{\mathrm{2}}}} }  ; \Gamma_{{\mathrm{1}}} \mathbin{+} \Gamma_{{\mathrm{2}}}  \vdash  \ottnt{e_{{\mathrm{1}}}} \; \ottnt{e_{{\mathrm{2}}}} \mathbin{:} \tau }\mathrel{=}{}\<[E]%
\\
\>[B]{}\hsindent{3}{}\<[3]%
\>[3]{} \kcase_  \multiplicityfont{ \ottsym{1} }  \;\Varid{z}\;\mathbf{of}\;\{\mskip1.5mu ( \ottmv{z_{{\mathrm{1}}}} , \ottmv{z_{{\mathrm{2}}}} )\to {}\<[E]%
\\
\>[3]{}\hsindent{2}{}\<[5]%
\>[5]{}(\dsterm{\ottmv{ \ottmv{z_{{\mathrm{1}}}} }}{  \constraintfont{ \ottnt{Q_{{\mathrm{1}}}} }  ; \Gamma_{{\mathrm{1}}}  \vdash  \ottnt{e_{{\mathrm{1}}}} \mathbin{:} \tau_{{\mathrm{1}}} \to_{1} \tau })\;(\dsterm{\ottmv{ \ottmv{z_{{\mathrm{2}}}} }}{  \constraintfont{ \ottnt{Q_{{\mathrm{2}}}} }  ; \Gamma_{{\mathrm{2}}}  \vdash  \ottnt{e_{{\mathrm{2}}}} \mathbin{:} \tau_{{\mathrm{1}}} })\mskip1.5mu\}{}\<[E]%
\\
\>[B]{}\dsterm{\ottmv{\Varid{z}}}{  \constraintfont{ \ottnt{Q_{{\mathrm{1}}}} }   \qtensor   \multiplicityfont{ \omega }  \mathbin{⋅}  \constraintfont{ \ottnt{Q_{{\mathrm{2}}}} }  ; \Gamma_{{\mathrm{1}}} \mathbin{+}  \multiplicityfont{ \omega }  \mathbin{⋅} \Gamma_{{\mathrm{2}}}  \vdash  \ottnt{e_{{\mathrm{1}}}} \; \ottnt{e_{{\mathrm{2}}}} \mathbin{:} \tau }\mathrel{=}{}\<[E]%
\\
\>[B]{}\hsindent{3}{}\<[3]%
\>[3]{} \kcase_  \multiplicityfont{ \ottsym{1} }  \;\Varid{z}\;\mathbf{of}\;\{\mskip1.5mu ( \ottmv{z_{{\mathrm{1}}}} ,\underline{ \ottmv{z_{{\mathrm{2}}}} })\to {}\<[E]%
\\
\>[3]{}\hsindent{2}{}\<[5]%
\>[5]{}(\dsterm{\ottmv{ \ottmv{z_{{\mathrm{1}}}} }}{  \constraintfont{ \ottnt{Q_{{\mathrm{1}}}} }  ; \Gamma_{{\mathrm{1}}}  \vdash  \ottnt{e_{{\mathrm{1}}}} \mathbin{:} \tau_{{\mathrm{1}}} \to_{\omega} \tau })\;(\dsterm{\ottmv{ \ottmv{z_{{\mathrm{2}}}} }}{  \constraintfont{ \ottnt{Q_{{\mathrm{2}}}} }  ; \Gamma_{{\mathrm{2}}}  \vdash  \ottnt{e_{{\mathrm{2}}}} \mathbin{:} \tau_{{\mathrm{1}}} })\mskip1.5mu\}{}\<[E]%
\\
\>[B]{}\dsterm{\ottmv{\Varid{z}}}{  \constraintfont{ \ottnt{Q} }   \qtensor   \constraintfont{ \ottnt{Q_{{\mathrm{1}}}} }  [ \overline{\upsilon} / \overline{\ottmv{a} } ]; \Gamma  \vdash \packbox\Varid{e}\mathbin{:}\exists\; \overline{\ottmv{a} } . \tau \RLolly  \constraintfont{ \ottnt{Q_{{\mathrm{1}}}} }  }\mathrel{=}{}\<[E]%
\\
\>[B]{}\hsindent{3}{}\<[3]%
\>[3]{} \kcase_  \multiplicityfont{ \ottsym{1} }  \;\Varid{z}\;\mathbf{of}\;\{\mskip1.5mu (\Varid{z'},\Varid{z''})\to {}\<[E]%
\\
\>[3]{}\hsindent{2}{}\<[5]%
\>[5]{}\packbox(\Varid{z''},\dsterm{\ottmv{\Varid{z'}}}{  \constraintfont{ \ottnt{Q} }  ; \Gamma  \vdash \Varid{e}\mathbin{:} \tau [ \overline{\upsilon} / \overline{\ottmv{a} } ]})\mskip1.5mu\}{}\<[E]%
\\
\>[B]{}\dsterm{\ottmv{\Varid{z}}}{  \constraintfont{ \ottnt{Q_{{\mathrm{1}}}} }   \qtensor   \constraintfont{ \ottnt{Q_{{\mathrm{2}}}} }  ; \Gamma_{{\mathrm{1}}} \mathbin{+} \Gamma_{{\mathrm{2}}}  \vdash \klet\ \packbox \Varid{x}\mathrel{=} \ottnt{e_{{\mathrm{1}}}} \;\mathbf{in}\; \ottnt{e_{{\mathrm{2}}}} \mathbin{:} \tau }\mathrel{=}{}\<[E]%
\\
\>[B]{}\hsindent{3}{}\<[3]%
\>[3]{} \kcase_  \multiplicityfont{ \ottsym{1} }  \;\Varid{z}\;\mathbf{of}\;\{\mskip1.5mu ( \ottmv{z_{{\mathrm{1}}}} , \ottmv{z_{{\mathrm{2}}}} )\to {}\<[E]%
\\
\>[3]{}\hsindent{2}{}\<[5]%
\>[5]{}\klet\ \packbox \Varid{z'},\Varid{x}\mathrel{=}\dsterm{\ottmv{ \ottmv{z_{{\mathrm{1}}}} }}{  \constraintfont{ \ottnt{Q_{{\mathrm{1}}}} }  ; \Gamma_{{\mathrm{1}}}  \vdash  \ottnt{e_{{\mathrm{1}}}} \mathbin{:}\exists\; \overline{\ottmv{a} } . \tau_{{\mathrm{1}}} \RLolly  \constraintfont{ \ottnt{Q} }  }\;\mathbf{in}{}\<[E]%
\\
\>[3]{}\hsindent{2}{}\<[5]%
\>[5]{} \klet_  \multiplicityfont{ \ottsym{1} }  \; \ottmv{z_{{\mathrm{2}}}} '\mathrel{=}( \ottmv{z_{{\mathrm{2}}}} ,\Varid{z'})\;\mathbf{in}{}\<[E]%
\\
\>[3]{}\hsindent{2}{}\<[5]%
\>[5]{}\dsterm{\ottmv{ \ottmv{z_{{\mathrm{2}}}} '}}{  \constraintfont{ \ottnt{Q_{{\mathrm{2}}}} }   \qtensor   \constraintfont{ \ottnt{Q} }  ; \Gamma_{{\mathrm{2}}} ,\Varid{x}\mathop{:_{1}} \tau_{{\mathrm{1}}}  \vdash  \ottnt{e_{{\mathrm{2}}}} \mathbin{:} \tau }\mskip1.5mu\}{}\<[E]%
\\
\>[B]{}\dsterm{\ottmv{\Varid{z}}}{  \constraintfont{ \ottnt{Q_{{\mathrm{1}}}} }   \qtensor   \constraintfont{ \ottnt{Q_{{\mathrm{2}}}} }  ; \Gamma_{{\mathrm{1}}} \mathbin{+} \Gamma_{{\mathrm{2}}}  \vdash  \klet_  \multiplicityfont{ \ottsym{1} }  \;\Varid{x}\mathrel{=} \ottnt{e_{{\mathrm{1}}}} \;\mathbf{in}\; \ottnt{e_{{\mathrm{2}}}} \mathbin{:} \tau }\mathrel{=}{}\<[E]%
\\
\>[B]{}\hsindent{3}{}\<[3]%
\>[3]{} \kcase_  \multiplicityfont{ \ottsym{1} }  \;\Varid{z}\;\mathbf{of}\;\{\mskip1.5mu ( \ottmv{z_{{\mathrm{1}}}} , \ottmv{z_{{\mathrm{2}}}} )\to {}\<[E]%
\\
\>[3]{}\hsindent{2}{}\<[5]%
\>[5]{} \klet_  \multiplicityfont{ \ottsym{1} }  \;\Varid{x}\mathbin{:}\dsevidence{  \constraintfont{ \ottnt{Q} }  }\to_{1} \tau_{{\mathrm{1}}} \mathrel{=}\dsterm{\ottmv{ \ottmv{z_{{\mathrm{1}}}} }}{  \constraintfont{ \ottnt{Q_{{\mathrm{1}}}} }   \qtensor   \constraintfont{ \ottnt{Q} }  ; \Gamma_{{\mathrm{1}}}  \vdash  \ottnt{e_{{\mathrm{1}}}} \mathbin{:} \tau_{{\mathrm{1}}} }{}\<[E]%
\\
\>[3]{}\hsindent{2}{}\<[5]%
\>[5]{}\mathbf{in}\;\dsterm{\ottmv{ \ottmv{z_{{\mathrm{2}}}} }}{  \constraintfont{ \ottnt{Q_{{\mathrm{2}}}} }  ; \Gamma_{{\mathrm{2}}} ,\Varid{x}\mathop{:_{1}} \tau_{{\mathrm{1}}}  \vdash  \ottnt{e_{{\mathrm{2}}}} \mathbin{:} \tau }\mskip1.5mu\}{}\<[E]%
\\
\>[B]{}\dsterm{\ottmv{\Varid{z}}}{  \multiplicityfont{ \omega }  \mathbin{⋅}  \constraintfont{ \ottnt{Q_{{\mathrm{1}}}} }   \qtensor   \constraintfont{ \ottnt{Q_{{\mathrm{2}}}} }  ;  \multiplicityfont{ \omega }  \mathbin{⋅} \Gamma_{{\mathrm{1}}} \mathbin{+} \Gamma_{{\mathrm{2}}}  \vdash  \klet_  \multiplicityfont{ \omega }  \;\Varid{x}\mathrel{=} \ottnt{e_{{\mathrm{1}}}} \;\mathbf{in}\; \ottnt{e_{{\mathrm{2}}}} \mathbin{:} \tau }\mathrel{=}{}\<[E]%
\\
\>[B]{}\hsindent{3}{}\<[3]%
\>[3]{} \kcase_  \multiplicityfont{ \ottsym{1} }  \;\Varid{z}\;\mathbf{of}\;\{\mskip1.5mu (\underline{ \ottmv{z_{{\mathrm{1}}}} }, \ottmv{z_{{\mathrm{2}}}} )\to {}\<[E]%
\\
\>[3]{}\hsindent{2}{}\<[5]%
\>[5]{} \klet_  \multiplicityfont{ \omega }  \;\Varid{x}\mathbin{:}\dsevidence{  \constraintfont{ \ottnt{Q} }  }\to_{1} \tau_{{\mathrm{1}}} \mathrel{=}\dsterm{\ottmv{ \ottmv{z_{{\mathrm{1}}}} }}{  \constraintfont{ \ottnt{Q_{{\mathrm{1}}}} }   \qtensor   \constraintfont{ \ottnt{Q} }  ; \Gamma_{{\mathrm{1}}}  \vdash  \ottnt{e_{{\mathrm{1}}}} \mathbin{:} \tau_{{\mathrm{1}}} }\;\mathbf{in}{}\<[E]%
\\
\>[3]{}\hsindent{2}{}\<[5]%
\>[5]{}\dsterm{\ottmv{ \ottmv{z_{{\mathrm{2}}}} }}{  \constraintfont{ \ottnt{Q_{{\mathrm{2}}}} }  ; \Gamma_{{\mathrm{2}}} ,\Varid{x}\mathop{:_{  \multiplicityfont{ \omega }  }} \tau_{{\mathrm{1}}}  \vdash  \ottnt{e_{{\mathrm{2}}}} \mathbin{:} \tau }\mskip1.5mu\}{}\<[E]%
\\
\>[B]{}\dsterm{\ottmv{\Varid{z}}}{  \constraintfont{ \ottnt{Q_{{\mathrm{1}}}} }   \qtensor   \constraintfont{ \ottnt{Q_{{\mathrm{2}}}} }  ; \Gamma_{{\mathrm{1}}} \mathbin{+} \Gamma_{{\mathrm{2}}}  \vdash  \klet_  \multiplicityfont{ \ottsym{1} }  \;\Varid{x}\mathbin{:}\forall\; \overline{\ottmv{a} } .  \constraintfont{ \ottnt{Q} }  \Lolly  \tau_{{\mathrm{1}}} \mathrel{=} \ottnt{e_{{\mathrm{1}}}} \;\mathbf{in}\; \ottnt{e_{{\mathrm{2}}}} \mathbin{:} \tau }\mathrel{=}{}\<[E]%
\\
\>[B]{}\hsindent{3}{}\<[3]%
\>[3]{} \kcase_  \multiplicityfont{ \ottsym{1} }  \;\Varid{z}\;\mathbf{of}\;\{\mskip1.5mu ( \ottmv{z_{{\mathrm{1}}}} , \ottmv{z_{{\mathrm{2}}}} )\to {}\<[E]%
\\
\>[3]{}\hsindent{2}{}\<[5]%
\>[5]{} \klet_  \multiplicityfont{ \ottsym{1} }  \;\Varid{x}\mathbin{:}\forall\; \overline{\ottmv{a} } .\dsevidence{  \constraintfont{ \ottnt{Q} }  }\to_{1} \tau_{{\mathrm{1}}} \mathrel{=}\dsterm{\ottmv{ \ottmv{z_{{\mathrm{1}}}} }}{  \constraintfont{ \ottnt{Q_{{\mathrm{1}}}} }   \qtensor   \constraintfont{ \ottnt{Q} }  ; \Gamma_{{\mathrm{1}}}  \vdash  \ottnt{e_{{\mathrm{1}}}} \mathbin{:} \tau_{{\mathrm{1}}} }\;\mathbf{in}{}\<[E]%
\\
\>[3]{}\hsindent{2}{}\<[5]%
\>[5]{}\dsterm{\ottmv{ \ottmv{z_{{\mathrm{2}}}} }}{  \constraintfont{ \ottnt{Q_{{\mathrm{2}}}} }  ; \Gamma_{{\mathrm{2}}} ,\Varid{x}\mathop{:_{1}}\forall\; \overline{\ottmv{a} } .  \constraintfont{ \ottnt{Q} }  \Lolly  \tau_{{\mathrm{1}}}  \vdash  \ottnt{e_{{\mathrm{2}}}} \mathbin{:} \tau }\mskip1.5mu\}{}\<[E]%
\\
\>[B]{}\dsterm{\ottmv{\Varid{z}}}{  \multiplicityfont{ \omega }  \mathbin{⋅}  \constraintfont{ \ottnt{Q_{{\mathrm{1}}}} }   \qtensor   \constraintfont{ \ottnt{Q_{{\mathrm{2}}}} }  ;  \multiplicityfont{ \omega }  \mathbin{⋅} \Gamma_{{\mathrm{1}}} \mathbin{+} \Gamma_{{\mathrm{2}}}  \vdash  \klet_  \multiplicityfont{ \omega }  \;\Varid{x}\mathbin{:}\forall\; \overline{\ottmv{a} } .  \constraintfont{ \ottnt{Q} }  \Lolly  \tau_{{\mathrm{1}}} \mathrel{=} \ottnt{e_{{\mathrm{1}}}} \;\mathbf{in}\; \ottnt{e_{{\mathrm{2}}}} \mathbin{:} \tau }\mathrel{=}{}\<[E]%
\\
\>[B]{}\hsindent{3}{}\<[3]%
\>[3]{} \kcase_  \multiplicityfont{ \ottsym{1} }  \;\Varid{z}\;\mathbf{of}\;\{\mskip1.5mu (\underline{ \ottmv{z_{{\mathrm{1}}}} }, \ottmv{z_{{\mathrm{2}}}} )\to {}\<[E]%
\\
\>[3]{}\hsindent{2}{}\<[5]%
\>[5]{} \klet_  \multiplicityfont{ \omega }  \;\Varid{x}\mathbin{:}\forall\; \overline{\ottmv{a} } .\dsevidence{  \constraintfont{ \ottnt{Q} }  }\to_{1} \tau_{{\mathrm{1}}} \mathrel{=}\dsterm{\ottmv{ \ottmv{z_{{\mathrm{1}}}} }}{  \constraintfont{ \ottnt{Q_{{\mathrm{1}}}} }   \qtensor   \constraintfont{ \ottnt{Q} }  ; \Gamma_{{\mathrm{1}}}  \vdash  \ottnt{e_{{\mathrm{1}}}} \mathbin{:} \tau_{{\mathrm{1}}} }\;\mathbf{in}{}\<[E]%
\\
\>[3]{}\hsindent{2}{}\<[5]%
\>[5]{}\dsterm{\ottmv{ \ottmv{z_{{\mathrm{2}}}} }}{  \constraintfont{ \ottnt{Q_{{\mathrm{2}}}} }  ; \Gamma_{{\mathrm{2}}} ,\Varid{x}\mathop{:_{  \multiplicityfont{ \omega }  }} \tau_{{\mathrm{1}}}  \vdash  \ottnt{e_{{\mathrm{2}}}} \mathbin{:} \tau }\mskip1.5mu\}{}\<[E]%
\\
\>[B]{}\dsterm{\ottmv{\Varid{z}}}{  \multiplicityfont{ \omega }  \mathbin{⋅}  \constraintfont{ \ottnt{Q_{{\mathrm{1}}}} }   \qtensor   \constraintfont{ \ottnt{Q_{{\mathrm{2}}}} }  ;  \multiplicityfont{ \omega }  \mathbin{⋅} \Gamma_{{\mathrm{1}}} \mathbin{+} \Gamma_{{\mathrm{2}}}  \vdash  \kcase_  \multiplicityfont{ \ottsym{1} }  \;\Varid{e}\;\mathbf{of}\;\{\mskip1.5mu \overline{\ottmv{K}_i\ \overline{\ottmv{x}_i } \to \ottnt{e}_i }\mskip1.5mu\}\mathbin{:} \tau }{}\<[67]%
\>[67]{}\mathrel{=}{}\<[67E]%
\\
\>[B]{}\hsindent{3}{}\<[3]%
\>[3]{} \kcase_  \multiplicityfont{ \ottsym{1} }  \;\Varid{z}\;\mathbf{of}\;\{\mskip1.5mu (\underline{ \ottmv{z_{{\mathrm{1}}}} }, \ottmv{z_{{\mathrm{2}}}} )\to {}\<[E]%
\\
\>[3]{}\hsindent{2}{}\<[5]%
\>[5]{} \kcase_  \multiplicityfont{ \ottsym{1} }  \;(\dsterm{\ottmv{ \ottmv{z_{{\mathrm{1}}}} }}{  \constraintfont{ \ottnt{Q_{{\mathrm{1}}}} }  ; \Gamma_{{\mathrm{1}}}  \vdash \Varid{e}\mathbin{:}\Conid{T}\; \overline{\tau} })\;\mathbf{of}{}\<[E]%
\\
\>[5]{}\hsindent{2}{}\<[7]%
\>[7]{}\{\mskip1.5mu \overline{\Conid{K}\; \overline{\ottmv{x} }_{\ottmv{i}} \to \dsterm{\ottmv{ \ottmv{z_{{\mathrm{2}}}} }}{  \constraintfont{ \ottnt{Q_{{\mathrm{2}}}} }  ; \Gamma_{{\mathrm{2}}} ,\overline{ \ottmv{x_{\ottmv{i}}} \mathop{:_{(  \multiplicityfont{ \pi }  \mathbin{⋅}  \multiplicityfont{ \pi_{\ottmv{i}} }  )}} \upsilon_{\ottmv{i}} [ \overline{\tau} / \overline{\ottmv{a} } ]} \vdash  \ottnt{e_{\ottmv{i}}} \mathbin{:} \tau }}\mskip1.5mu\}\mskip1.5mu\}{}\<[E]%
\\
\>[B]{}\dsterm{\ottmv{\Varid{z}}}{  \constraintfont{ \ottnt{Q_{{\mathrm{1}}}} }   \qtensor   \constraintfont{ \ottnt{Q_{{\mathrm{2}}}} }  ; \Gamma_{{\mathrm{1}}} \mathbin{+} \Gamma_{{\mathrm{2}}}  \vdash  \kcase_  \multiplicityfont{ \omega }  \;\Varid{e}\;\mathbf{of}\;\{\mskip1.5mu \overline{\ottmv{K}_i\ \overline{\ottmv{x}_i } \to \ottnt{e}_i }\mskip1.5mu\}\mathbin{:} \tau }{}\<[53]%
\>[53]{}\mathrel{=}{}\<[53E]%
\\
\>[B]{}\hsindent{3}{}\<[3]%
\>[3]{} \kcase_  \multiplicityfont{ \ottsym{1} }  \;\Varid{z}\;\mathbf{of}\;\{\mskip1.5mu ( \ottmv{z_{{\mathrm{1}}}} , \ottmv{z_{{\mathrm{2}}}} )\to {}\<[E]%
\\
\>[3]{}\hsindent{2}{}\<[5]%
\>[5]{} \kcase_  \multiplicityfont{ \omega }  \;(\dsterm{\ottmv{ \ottmv{z_{{\mathrm{1}}}} }}{  \constraintfont{ \ottnt{Q_{{\mathrm{1}}}} }  ; \Gamma_{{\mathrm{1}}}  \vdash \Varid{e}\mathbin{:}\Conid{T}\; \overline{\tau} })\;\mathbf{of}{}\<[E]%
\\
\>[5]{}\hsindent{2}{}\<[7]%
\>[7]{}\{\mskip1.5mu \overline{\Conid{K}\; \overline{\ottmv{x} }_{\ottmv{i}} \to \dsterm{\ottmv{ \ottmv{z_{{\mathrm{2}}}} }}{  \constraintfont{ \ottnt{Q_{{\mathrm{2}}}} }  ; \Gamma_{{\mathrm{2}}} ,\overline{ \ottmv{x_{\ottmv{i}}} \mathop{:_{(  \multiplicityfont{ \pi }  \mathbin{⋅}  \multiplicityfont{ \pi_{\ottmv{i}} }  )}} \upsilon_{\ottmv{i}} [ \overline{\tau} / \overline{\ottmv{a} } ]} \vdash  \ottnt{e_{\ottmv{i}}} \mathbin{:} \tau }}\mskip1.5mu\}\mskip1.5mu\}{}\<[E]%
\ColumnHook
\end{hscode}\resethooks
\end{minipage}
\right.
$$

  \caption{Desugaring}
  \label{fig:full:desugaring}
\end{figure}

\section{Proofs}
\label{sec:appendix:proofs-lemmas}

\setcounter{subsection}{4}
\subsection{Lemmas on the qualified type system}
\label{sec:appendix:qual-type-syst}

\begin{proof}[Proof of \cref{lem:q:scaling}]
  Let us prove separately the cases $  \multiplicityfont{ \pi }  =  \multiplicityfont{ \ottsym{1} }  $ and
  $  \multiplicityfont{ \pi }  =  \multiplicityfont{ \omega }  $.
  \begin{itemize}
  \item When $  \multiplicityfont{ \pi }  =  \multiplicityfont{ \ottsym{1} }  $, then $  \constraintfont{   \multiplicityfont{ \pi }  \scale \ottnt{Q}  }  =  \constraintfont{ \ottnt{Q} }  $ for all
    $  \constraintfont{ \ottnt{Q} }  $, hence $ \constraintfont{ \ottnt{Q_{{\mathrm{1}}}} }   \Vdash   \constraintfont{ \ottnt{Q_{{\mathrm{2}}}} } $ implies $ \constraintfont{   \multiplicityfont{ \pi }  \scale \ottnt{Q_{{\mathrm{1}}}}  }   \Vdash   \constraintfont{   \multiplicityfont{ \pi }  \scale \ottnt{Q_{{\mathrm{2}}}}  } $.
  \item For the case $  \multiplicityfont{ \pi }  =  \multiplicityfont{ \omega }  $, let us consider a few
    properties. First note that, for any $  \constraintfont{ \ottnt{Q} }  $,
    $  \constraintfont{   \multiplicityfont{ \omega }  \scale \ottnt{Q}  }  =  \constraintfont{   \multiplicityfont{ \omega }  \scale \ottnt{Q}   \qtensor    \multiplicityfont{ \omega }  \scale \ottnt{Q}  }  $. From which it follows,
    using the laws of \cref{def:entailment-relation}, that
    $ \constraintfont{   \multiplicityfont{ \omega }  \scale \ottnt{Q}  }   \Vdash   \constraintfont{ \ottnt{Q_{{\mathrm{1}}}}  \qtensor  \ottnt{Q_{{\mathrm{2}}}} } $ if and only if $ \constraintfont{   \multiplicityfont{ \omega }  \scale \ottnt{Q}  }   \Vdash   \constraintfont{ \ottnt{Q_{{\mathrm{1}}}} } $ and
    $ \constraintfont{   \multiplicityfont{ \omega }  \scale \ottnt{Q}  }   \Vdash   \constraintfont{ \ottnt{Q_{{\mathrm{2}}}} } $.

    This means that to verify that
    $ \constraintfont{   \multiplicityfont{ \omega }  \scale \ottnt{Q_{{\mathrm{1}}}}  }   \Vdash   \constraintfont{   \multiplicityfont{ \omega }  \scale \ottnt{Q_{{\mathrm{2}}}}  } $, it is equivalent to prove that $ \constraintfont{   \multiplicityfont{ \omega }  \scale \ottnt{Q_{{\mathrm{1}}}}  }   \Vdash   \constraintfont{   \multiplicityfont{ \omega }  \scale \ottnt{q_{{\mathrm{2}}}}  } $ for each
    $  \constraintfont{ \ottnt{q_{{\mathrm{2}}}} }  ∈  \constraintfont{ \ottnt{U} }  $ (letting $  \constraintfont{   \multiplicityfont{ \omega }  \scale \ottnt{Q_{{\mathrm{2}}}}  }  =  \constraintfont{ \ottsym{(}  \ottnt{U}  \ottsym{,}  \emptyset  \ottsym{)} }  $). In turn, by \cref{def:entailment-relation} and
    observing that $  \constraintfont{   \multiplicityfont{ \omega }  \scale \ottsym{(}    \multiplicityfont{ \omega }  \scale \ottnt{Q_{{\mathrm{1}}}}   \ottsym{)}  }   =   \constraintfont{ \ottnt{Q_{{\mathrm{1}}}} }  $, this is
    equivalent to $ \constraintfont{   \multiplicityfont{ \omega }  \scale \ottnt{Q_{{\mathrm{1}}}}  }   \Vdash   \constraintfont{   \multiplicityfont{ \ottsym{1} }  \scale \ottnt{q_{{\mathrm{2}}}}  } $.

    This follows from the fact that $ \constraintfont{ \ottnt{Q_{{\mathrm{1}}}} }   \Vdash   \constraintfont{ \ottnt{Q_{{\mathrm{2}}}} } $ implies
    $ \constraintfont{   \multiplicityfont{ \omega }  \scale \ottnt{Q_{{\mathrm{1}}}}  }   \Vdash   \constraintfont{ \ottnt{Q_{{\mathrm{2}}}} } $ (\cref{def:entailment-relation}) and the
    property, shown above, that $ \constraintfont{   \multiplicityfont{ \omega }  \scale \ottnt{Q_{{\mathrm{1}}}}  }   \Vdash   \constraintfont{ \ottnt{Q_{{\mathrm{2}}}}  \qtensor  \ottnt{Q'_{{\mathrm{2}}}} } $ if and
    only if $ \constraintfont{   \multiplicityfont{ \omega }  \scale \ottnt{Q_{{\mathrm{1}}}}  }   \Vdash   \constraintfont{ \ottnt{Q_{{\mathrm{2}}}} } $ and $ \constraintfont{   \multiplicityfont{ \omega }  \scale \ottnt{Q_{{\mathrm{1}}}}  }   \Vdash   \constraintfont{ \ottnt{Q'_{{\mathrm{2}}}} } $.
  \end{itemize}
\end{proof}

\begin{proof}[Proof of \cref{lem:q:scaling-inversion}]
  Let us prove separately the cases $  \multiplicityfont{ \pi }  =  \multiplicityfont{ \ottsym{1} }  $ and
  $  \multiplicityfont{ \pi }  =  \multiplicityfont{ \omega }  $.
  \begin{itemize}
  \item When $  \multiplicityfont{ \pi }  =  \multiplicityfont{ \ottsym{1} }  $, then $  \constraintfont{   \multiplicityfont{ \pi }  \scale \ottnt{Q}  }  =  \constraintfont{ \ottnt{Q} }  $ for all
    $  \constraintfont{ \ottnt{Q} }  $, in particular $ \constraintfont{ \ottnt{Q_{{\mathrm{1}}}} }   \Vdash   \constraintfont{   \multiplicityfont{ \ottsym{1} }  \scale \ottnt{Q_{{\mathrm{2}}}}  } $ implies that
    $  \constraintfont{ \ottnt{Q_{{\mathrm{1}}}} }  =  \constraintfont{   \multiplicityfont{ \ottsym{1} }  \scale \ottnt{Q_{{\mathrm{1}}}}  }  $ with $ \constraintfont{ \ottnt{Q_{{\mathrm{1}}}} }   \Vdash   \constraintfont{ \ottnt{Q_{{\mathrm{2}}}} } $.
  \item When $  \multiplicityfont{ \pi }  =  \multiplicityfont{ \omega }  $, then let us first remark, letting
    $  \constraintfont{   \multiplicityfont{ \omega }  \scale \ottnt{Q_{{\mathrm{2}}}}  }  =  \constraintfont{ \ottsym{(}  \ottnt{U}  \ottsym{,}  \emptyset  \ottsym{)} }  $ that, by a straightforward
    induction on the cardinality of $  \constraintfont{ \ottnt{U} }  $ it is sufficient to
    prove that the result holds for atomic constraints.

    That is, we need to prove that if $ \constraintfont{ \ottnt{Q_{{\mathrm{1}}}} }   \Vdash   \constraintfont{   \multiplicityfont{ \omega }  \scale \ottnt{q_{{\mathrm{2}}}}  } $ then
    there exists $  \constraintfont{ \ottnt{Q'} }  $ such that $  \constraintfont{ \ottnt{Q_{{\mathrm{1}}}} }  =  \constraintfont{   \multiplicityfont{ \omega }  \scale \ottnt{Q'}  }  $ and
    $ \constraintfont{ \ottnt{Q'} }   \Vdash   \constraintfont{   \multiplicityfont{ \rho }  \scale \ottnt{q_{{\mathrm{2}}}}  } $ (for all $  \multiplicityfont{ \rho }  $).

    This result, in turns, holds by \cref{def:entailment-relation}.
  \end{itemize}
\end{proof}

\begin{lemma}\label{lem:simples:monoid-action}
  The following equality holds $  \constraintfont{   \multiplicityfont{ \pi }  \scale \ottsym{(}    \multiplicityfont{ \rho }  \scale \ottnt{Q}   \ottsym{)}  }  =  \constraintfont{   \multiplicityfont{ \ottsym{(}   \pi {⋅} \rho   \ottsym{)} }  \scale \ottnt{Q}  }  $
\end{lemma}
\begin{proof}
  Immediate by case analysis of $  \multiplicityfont{ \pi }  $ and $  \multiplicityfont{ \rho }  $.
\end{proof}

\setcounter{subsection}{5}
\subsection{Lemmas on constraint inference}
\label{sec:appendix:constraint-inference}

\begin{lemma}[$ \constraintfont{\mathcal{D} } $ discarding]
  \label{lem:dup-weakening}
  The two following, equivalent, properties hold
  \begin{itemize}
  \item if $  \constraintfont{ Q_{\mathcal{D} } }   \in  \constraintfont{\mathcal{D} } $, then $ \constraintfont{ Q_{\mathcal{D} } }   \Vdash   \constraintfont{  \mathbf{\varepsilon}  } $
  \item if $ \constraintfont{ \ottnt{Q_{{\mathrm{1}}}} }   \Vdash   \constraintfont{ \ottnt{Q_{{\mathrm{2}}}} } $ and $  \constraintfont{ Q_{\mathcal{D} } }   \in  \constraintfont{\mathcal{D} } $, then $ \constraintfont{ \ottnt{Q_{{\mathrm{1}}}}  \qtensor  Q_{\mathcal{D} } }   \Vdash   \constraintfont{ \ottnt{Q_{{\mathrm{2}}}} } $
  \end{itemize}
\end{lemma}
\begin{proof}
  The two properties are equivalent
  \begin{itemize}
  \item using $  \constraintfont{ \constraintfont{C_{{\mathrm{1}}}} }  =  \constraintfont{ \constraintfont{C_{{\mathrm{2}}}} }  =  \constraintfont{  \mathbf{\varepsilon}  }  $, the latter implies the
    former
  \item The former implies the latter by tensoring together
    $ \constraintfont{ \ottnt{Q_{{\mathrm{1}}}} }   \Vdash   \constraintfont{ \ottnt{Q_{{\mathrm{2}}}} } $ and $ \constraintfont{ Q_{\mathcal{D} } }   \Vdash   \constraintfont{  \mathbf{\varepsilon}  } $ following the rules
    of~\cref{fig:entailment-relation}.
  \end{itemize}
  Let $  \constraintfont{ Q_{\mathcal{D} } }   \in  \constraintfont{\mathcal{D} } $, then for each $  \constraintfont{   \multiplicityfont{ \ottsym{1} }  \scale \ottnt{q}  }   \in   \constraintfont{ Q_{\mathcal{D} } }  $,
  $ \constraintfont{   \multiplicityfont{ \ottsym{1} }  \scale \ottnt{q}  }   \Vdash   \constraintfont{  \mathbf{\varepsilon}  } $ (per~\cref{fig:entailment-relation}), tensoring
  each of these entailments together and with the $  \constraintfont{   \multiplicityfont{ \omega }  \scale \ottnt{q}  }   \in   \constraintfont{ Q_{\mathcal{D} } }  $, we get $ \constraintfont{ Q_{\mathcal{D} } }   \Vdash   \constraintfont{  \mathbf{\varepsilon}  } $.
\end{proof}

\begin{lemma}[$ \constraintfont{\mathcal{D} } $ duplication]
  \label{lem:dup-contraction}
  The two following, equivalent, properties hold
  \begin{itemize}
  \item if $  \constraintfont{ Q_{\mathcal{D} } }   \in  \constraintfont{\mathcal{D} } $, then $ \constraintfont{ Q_{\mathcal{D} } }   \Vdash   \constraintfont{ Q_{\mathcal{D} }  \qtensor  Q_{\mathcal{D} } } $
  \item if $ \constraintfont{ \ottnt{Q_{{\mathrm{1}}}}  \qtensor  Q_{\mathcal{D} } }   \Vdash   \constraintfont{ \ottnt{Q_{{\mathrm{2}}}} } $, $ \constraintfont{ Q_{\mathcal{D} }  \qtensor  \ottnt{Q'_{{\mathrm{1}}}} }   \Vdash   \constraintfont{ \ottnt{Q'_{{\mathrm{2}}}} } $, and
    $  \constraintfont{ Q_{\mathcal{D} } }   \in  \constraintfont{\mathcal{D} } $, then $ \constraintfont{ \ottnt{Q_{{\mathrm{1}}}}  \qtensor  Q_{\mathcal{D} }  \qtensor  \ottnt{Q'_{{\mathrm{1}}}} }   \Vdash   \constraintfont{ \ottnt{Q_{{\mathrm{2}}}}  \qtensor  \ottnt{Q'_{{\mathrm{2}}}} } $
  \end{itemize}
\end{lemma}
\begin{proof}
  The proof is similar to that of~\cref{lem:dup-weakening}
\end{proof}

\begin{lemma}[Transitive tensor decomposition]
  \label{lem:transitive-tensor-decomposition}
  if $  \constraintfont{ Q_{\mathcal{D} } }   \in  \constraintfont{\mathcal{D} } $ and $ \constraintfont{ \ottnt{Q} }   \Vdash   \constraintfont{ \ottnt{Q_{{\mathrm{1}}}}  \qtensor  Q_{\mathcal{D} }  \qtensor  \ottnt{Q_{{\mathrm{2}}}} } $, then there
  exists $\ottnt{Q'_{{\mathrm{1}}}}$, $Q_{\mathcal{D} }'$, $\ottnt{Q'_{{\mathrm{2}}}}$, such that
  \begin{itemize}
  \item $ \constraintfont{ \ottnt{Q'_{{\mathrm{1}}}}  \qtensor  Q_{\mathcal{D} }' }   \Vdash   \constraintfont{ \ottnt{Q_{{\mathrm{1}}}} } $
  \item $ \constraintfont{ Q_{\mathcal{D} }'  \qtensor  \ottnt{Q'_{{\mathrm{2}}}} }   \Vdash   \constraintfont{ \ottnt{Q_{{\mathrm{2}}}} } $
  \item $ \constraintfont{ Q_{\mathcal{D} }' }   \Vdash   \constraintfont{ Q_{\mathcal{D} } } $
  \end{itemize}
\end{lemma}
\begin{proof}
  By the inversion-of-tensor rule from~\cref{fig:entailment-relation},
  we get that ther exists $  \constraintfont{ \ottnt{Q'} }  $, $  \constraintfont{ Q_{\mathcal{D} }' }  $, and $  \constraintfont{ \ottnt{Q'_{{\mathrm{2}}}} }  $,
  such that
  \begin{itemize}
  \item $  \constraintfont{ \ottnt{Q} }   =   \constraintfont{ \ottnt{Q'}  \qtensor  Q_{\mathcal{D} }'  \qtensor  \ottnt{Q'_{{\mathrm{2}}}} }  $
  \item $ \constraintfont{ Q_{\mathcal{D} }'  \qtensor  \ottnt{Q'_{{\mathrm{2}}}} }   \Vdash   \constraintfont{ \ottnt{Q_{{\mathrm{2}}}} } $
  \item $ \constraintfont{ \ottnt{Q'}  \qtensor  Q_{\mathcal{D} }' }   \Vdash   \constraintfont{ \ottnt{Q_{{\mathrm{1}}}}  \qtensor  Q_{\mathcal{D} } } $. The inversion-of-tensor rule
    applies further to this case: there exists $  \constraintfont{ \ottnt{Q'_{{\mathrm{1}}}} }  $,
    $  \constraintfont{ Q_{\mathcal{D} }'' }  $, and $  \constraintfont{ \ottnt{Q''} }  $ such that
    \begin{itemize}
    \item $  \constraintfont{ \ottnt{Q'}  \qtensor  Q_{\mathcal{D} }' }   =   \constraintfont{ \ottnt{Q'_{{\mathrm{1}}}}  \qtensor  Q_{\mathcal{D} }''  \qtensor  \ottnt{Q''} }  $
    \item $ \constraintfont{ \ottnt{Q'_{{\mathrm{1}}}}  \qtensor  Q_{\mathcal{D} }'' }   \Vdash   \constraintfont{ \ottnt{Q'_{{\mathrm{1}}}} } $
    \item $ \constraintfont{ Q_{\mathcal{D} }''  \qtensor  \ottnt{Q''} }   \Vdash   \constraintfont{ Q_{\mathcal{D} }' } $
    \end{itemize}
  \end{itemize}
  Observe the following
  \begin{itemize}
  \item $  \constraintfont{ Q_{\mathcal{D} }''  \qtensor  \ottnt{Q''} }   \in  \constraintfont{\mathcal{D} } $ (because of the requirements
    of~\cref{fig:entailment-relation}). Let's write
    $  \constraintfont{ Q_{\mathcal{D} }''' }   =   \constraintfont{ Q_{\mathcal{D} }''  \qtensor  \ottnt{Q''} }  $.
  \item $  \constraintfont{ \ottnt{Q} }   =   \constraintfont{ \ottnt{Q'_{{\mathrm{1}}}}  \qtensor  Q_{\mathcal{D} }'''  \qtensor  \ottnt{Q'_{{\mathrm{2}}}} }  $
  \item We have
    \begin{itemize}
    \item $ \constraintfont{ \ottnt{Q'_{{\mathrm{1}}}}  \qtensor  Q_{\mathcal{D} }''' }   \Vdash   \constraintfont{ \ottnt{Q_{{\mathrm{1}}}} } $. Because
      \begin{itemize}
      \item $  \constraintfont{ \ottnt{Q'_{{\mathrm{1}}}}  \qtensor  Q_{\mathcal{D} }''' }   =   \constraintfont{ \ottnt{Q'}  \qtensor  Q_{\mathcal{D} }' }  $
      \item therefore $ \constraintfont{ \ottnt{Q'_{{\mathrm{1}}}}  \qtensor  Q_{\mathcal{D} }''' }   \Vdash   \constraintfont{ \ottnt{Q_{{\mathrm{1}}}}  \qtensor  Q_{\mathcal{D} } } $
      \item by~\cref{lem:dup-weakening}, and by transitivity of the
        entailment relation (per~\cref{fig:entailment-relation}), we
        can drop $  \constraintfont{ Q_{\mathcal{D} } }  $ from the conclusion.
      \end{itemize}
    \item $ \constraintfont{ Q_{\mathcal{D} }''' }   \Vdash   \constraintfont{ Q_{\mathcal{D} }' } $ (by definition of $  \constraintfont{ Q_{\mathcal{D} }''' }  $)
    \item $ \constraintfont{ Q_{\mathcal{D} }'''  \qtensor  \ottnt{Q'_{{\mathrm{2}}}} }   \Vdash   \constraintfont{ \ottnt{Q_{{\mathrm{2}}}} } $
      \begin{itemize}
      \item By tensoring together, per~\cref{fig:entailment-relation},
        $ \constraintfont{ Q_{\mathcal{D} }''' }   \Vdash   \constraintfont{ Q_{\mathcal{D} }' } $ and $ \constraintfont{ \ottnt{Q'_{{\mathrm{2}}}} }   \Vdash   \constraintfont{ \ottnt{Q'_{{\mathrm{2}}}} } $, we get
        $ \constraintfont{ Q_{\mathcal{D} }'''  \qtensor  \ottnt{Q'_{{\mathrm{2}}}} }   \Vdash   \constraintfont{ Q_{\mathcal{D} }'  \qtensor  \ottnt{Q'_{{\mathrm{2}}}} } $
      \item then we get the desired result by transitivity of the
        entailment relation.
      \end{itemize}
    \end{itemize}
  \end{itemize}

  This concludes the proof
\end{proof}

\begin{proof}[Proof of \cref{lem:inversion}]
  The cases $ \constraintfont{ \ottnt{Q} }   \vdash   \constraintfont{ \constraintfont{C_{{\mathrm{1}}}}  \aand  \constraintfont{C_{{\mathrm{2}}}} } $ and $ \constraintfont{ \ottnt{Q} }   \vdash   \constraintfont{   \multiplicityfont{ \pi }  \scale( \ottnt{Q_{{\mathrm{2}}}}  \Lolly  \constraintfont{C} )  } $ are
  straightforward by induction, so let us prove them first
  \begin{itemize}
  \item Suppose $ \constraintfont{ \ottnt{Q} }   \vdash   \constraintfont{ \constraintfont{C_{{\mathrm{1}}}}  \aand  \constraintfont{C_{{\mathrm{2}}}} } $, then there are two cases
    \begin{itemize}
    \item either it is the conclusion of a \rref*{C-With} rule,
      and the result is immediate.
    \item or it is the result of a \rref*{C-Dom} rule, then, there
      exists $  \constraintfont{ \ottnt{Q'} }  $, such that $ \constraintfont{ \ottnt{Q} }   \Vdash   \constraintfont{ \ottnt{Q'} } $ and
      $ \constraintfont{ \ottnt{Q'} }   \vdash   \constraintfont{ \constraintfont{C_{{\mathrm{1}}}}  \aand  \constraintfont{C_{{\mathrm{2}}}} } $.

      By induction $ \constraintfont{ \ottnt{Q'} }   \vdash   \constraintfont{ \constraintfont{C_{{\mathrm{2}}}} } $ and $ \constraintfont{ \ottnt{Q'} }   \vdash   \constraintfont{ \constraintfont{C_{{\mathrm{2}}}} } $, applying
      \rref*{C-Dom} to both gives $ \constraintfont{ \ottnt{Q} }   \vdash   \constraintfont{ \constraintfont{C_{{\mathrm{2}}}} } $ and $ \constraintfont{ \ottnt{Q} }   \vdash   \constraintfont{ \constraintfont{C_{{\mathrm{2}}}} } $ as
      required.
    \end{itemize}

  \item Suppose $ \constraintfont{ \ottnt{Q} }   \vdash   \constraintfont{   \multiplicityfont{ \pi }  \scale( \ottnt{Q_{{\mathrm{2}}}}  \Lolly  \constraintfont{C} )  } $, then there are two cases
    \begin{itemize}
    \item either it is the conclusion of a \rref*{C-Impl} rule,
      and the result is immediate.
    \item or it is the result of a \rref*{C-Dom} rule, then, there
      exists $  \constraintfont{ \ottnt{Q'} }  $, such that $ \constraintfont{ \ottnt{Q} }   \Vdash   \constraintfont{ \ottnt{Q'} } $ and
      $ \constraintfont{ \ottnt{Q'} }   \vdash   \constraintfont{   \multiplicityfont{ \pi }  \scale( \ottnt{Q_{{\mathrm{2}}}}  \Lolly  \constraintfont{C} )  } $.

      By induction, there exists $  \constraintfont{ \ottnt{Q'_{{\mathrm{1}}}} }  $ such that $ \constraintfont{ \ottnt{Q'_{{\mathrm{1}}}}  \qtensor  \ottnt{Q_{{\mathrm{2}}}} }   \vdash   \constraintfont{ \constraintfont{C} } $ and $  \constraintfont{ \ottnt{Q'} }   =   \constraintfont{   \multiplicityfont{ \pi }  \scale \ottnt{Q'_{{\mathrm{1}}}}  }  $,
      by~\cref{fig:entailment-relation}, there exists $  \constraintfont{ \ottnt{Q_{{\mathrm{1}}}} }  $ such
      that $  \constraintfont{ \ottnt{Q} }   =   \constraintfont{   \multiplicityfont{ \pi }  \scale \ottnt{Q_{{\mathrm{1}}}}  }  $ and $ \constraintfont{ \ottnt{Q_{{\mathrm{1}}}} }   \Vdash   \constraintfont{ \ottnt{Q'_{{\mathrm{1}}}} } $. Hence $ \constraintfont{ \ottnt{Q_{{\mathrm{1}}}}  \qtensor  \ottnt{Q_{{\mathrm{2}}}} }   \Vdash   \constraintfont{ \ottnt{Q'_{{\mathrm{1}}}}  \qtensor  \ottnt{Q_{{\mathrm{2}}}} } $, which lets us conclude with \rref*{C-Dom}.
    \end{itemize}
  \end{itemize}

  For $ \constraintfont{ \ottnt{Q} }   \vdash   \constraintfont{ \constraintfont{C_{{\mathrm{1}}}}  \qtensor  \constraintfont{C_{{\mathrm{2}}}} } $ we have the following cases:
  \begin{itemize}
  \item either it is the conclusion of a \rref*{C-Tensor} rule, and
    the result is immediate
  \item or it is the result of a \rref*{C-Id} rule, in which case
    $  \constraintfont{ \ottnt{Q} }   =   \constraintfont{ \constraintfont{C_{{\mathrm{1}}}}  \qtensor  \constraintfont{C_{{\mathrm{2}}}} }  $, which proves the result
  \item or it is the result of a \rref*{C-Dom} rule, in which case
    there is $  \constraintfont{ \ottnt{Q'} }  $ such that $ \constraintfont{ \ottnt{Q} }   \Vdash   \constraintfont{ \ottnt{Q'} } $ and $ \constraintfont{ \ottnt{Q'} }   \vdash   \constraintfont{ \constraintfont{C_{{\mathrm{1}}}}  \qtensor  \constraintfont{C_{{\mathrm{2}}}} } $.

    By induction, there exist $  \constraintfont{ \ottnt{Q'_{{\mathrm{1}}}} }  $, $  \constraintfont{ Q_{\mathcal{D} }' }  $, and
    $  \constraintfont{ \ottnt{Q'_{{\mathrm{2}}}} }  $, such that $  \constraintfont{ Q_{\mathcal{D} }' }   \in  \constraintfont{\mathcal{D} } $, $ \constraintfont{ \ottnt{Q'_{{\mathrm{1}}}}  \qtensor  Q_{\mathcal{D} }' }   \vdash   \constraintfont{ \constraintfont{C_{{\mathrm{1}}}} } $,
    $ \constraintfont{ Q_{\mathcal{D} }'  \qtensor  \ottnt{Q'_{{\mathrm{2}}}} }   \vdash   \constraintfont{ \constraintfont{C_{{\mathrm{2}}}} } $, and $  \constraintfont{ \ottnt{Q'} }   =   \constraintfont{ \ottnt{Q'_{{\mathrm{1}}}}  \qtensor  Q_{\mathcal{D} }'  \qtensor  \ottnt{Q'_{{\mathrm{2}}}} }  $.

    Then~\cref{lem:transitive-tensor-decomposition}, gives us
    $  \constraintfont{ \ottnt{Q_{{\mathrm{1}}}} }  $, $  \constraintfont{ Q_{\mathcal{D} } }  $, and $  \constraintfont{ \ottnt{Q_{{\mathrm{2}}}} }  $ such that
    \begin{itemize}
    \item $  \constraintfont{ \ottnt{Q} }   =   \constraintfont{ \ottnt{Q_{{\mathrm{1}}}}  \qtensor  Q_{\mathcal{D} }  \qtensor  \ottnt{Q_{{\mathrm{2}}}} }  $
    \item $  \constraintfont{ Q_{\mathcal{D} } }   \in  \constraintfont{\mathcal{D} } $
    \item $ \constraintfont{ \ottnt{Q_{{\mathrm{1}}}}  \qtensor  Q_{\mathcal{D} } }   \Vdash   \constraintfont{ \ottnt{Q'_{{\mathrm{1}}}} } $
    \item $ \constraintfont{ Q_{\mathcal{D} }  \qtensor  \ottnt{Q_{{\mathrm{2}}}} }   \Vdash   \constraintfont{ \ottnt{Q'_{{\mathrm{2}}}} } $
    \item $ \constraintfont{ Q_{\mathcal{D} } }   \Vdash   \constraintfont{ Q_{\mathcal{D} }' } $
    \end{itemize}
    By~\ref{lem:dup-contraction}, we can further deduce that
    \begin{itemize}
    \item $ \constraintfont{ \ottnt{Q_{{\mathrm{1}}}}  \qtensor  Q_{\mathcal{D} } }   \Vdash   \constraintfont{ \ottnt{Q'_{{\mathrm{1}}}}  \qtensor  Q_{\mathcal{D} }' } $
    \item $ \constraintfont{ Q_{\mathcal{D} }  \qtensor  \ottnt{Q_{{\mathrm{2}}}} }   \Vdash   \constraintfont{ Q_{\mathcal{D} }'  \qtensor  \ottnt{Q'_{{\mathrm{2}}}} } $
    \end{itemize}
    Which concludes the proof, by the \rref*{C-Dom} rule
  \end{itemize}
\end{proof}

\begin{proof}[Proof of \cref{lem:wanted:promote}]
  By induction on the syntax of $  \constraintfont{ \constraintfont{C} }  $
  \begin{itemize}
  \item If $  \constraintfont{ \constraintfont{C} }  =  \constraintfont{ \ottnt{Q'} }  $, then the result follows
    from \cref{lem:q:scaling}
  \item If $  \constraintfont{ \constraintfont{C} }  =  \constraintfont{ \constraintfont{C_{{\mathrm{1}}}}  \qtensor  \constraintfont{C_{{\mathrm{2}}}} }  $, then we can prove the result like we
    proved the corresponding case in \cref{lem:q:scaling},
    using \cref{lem:inversion}.
  \item If $  \constraintfont{ \constraintfont{C} }  =  \constraintfont{ \constraintfont{C_{{\mathrm{1}}}}  \aand  \constraintfont{C_{{\mathrm{2}}}} }  $, then we the case where $  \multiplicityfont{ \pi }  =  \multiplicityfont{ \ottsym{1} }  $ is
    immediate, so we can assume without loss of generality that
    $  \multiplicityfont{ \pi }  =  \multiplicityfont{ \omega }  $, and, therefore, that $  \constraintfont{   \multiplicityfont{ \pi }  \scale \constraintfont{C}  }   =   \multiplicityfont{ \pi }  \scale \constraintfont{C_{{\mathrm{1}}}}   \qtensor    \multiplicityfont{ \pi }  \scale \constraintfont{C_{{\mathrm{2}}}} $.
    By \cref{lem:inversion}, we have that $ \constraintfont{ \ottnt{Q} }   \vdash   \constraintfont{ \constraintfont{C_{{\mathrm{1}}}} } $ and
    $ \constraintfont{ \ottnt{Q} }   \vdash   \constraintfont{ \constraintfont{C_{{\mathrm{2}}}} } $; hence, by induction, $ \constraintfont{   \multiplicityfont{ \omega }  \scale \ottnt{Q}  }   \vdash   \constraintfont{   \multiplicityfont{ \omega }  \scale \constraintfont{C_{{\mathrm{1}}}}  } $ and
    $ \constraintfont{   \multiplicityfont{ \omega }  \scale \ottnt{Q}  }   \vdash   \constraintfont{   \multiplicityfont{ \omega }  \scale \constraintfont{C_{{\mathrm{1}}}}  } $.
    Then, by definition of the entailment relation, we have $ \constraintfont{   \multiplicityfont{ \omega }  \scale \ottnt{Q}   \qtensor    \multiplicityfont{ \omega }  \scale \ottnt{Q}  }   \vdash   \constraintfont{   \multiplicityfont{ \omega }  \scale \constraintfont{C_{{\mathrm{1}}}}   \qtensor    \multiplicityfont{ \omega }  \scale \constraintfont{C_{{\mathrm{2}}}}  } $, which concludes,
    since $  \constraintfont{   \multiplicityfont{ \omega }  \scale \ottnt{Q}  }   =   \constraintfont{   \multiplicityfont{ \omega }  \scale \ottnt{Q}   \qtensor    \multiplicityfont{ \omega }  \scale \ottnt{Q}  }  $.
  \item If $  \constraintfont{ \constraintfont{C} }  =  \constraintfont{   \multiplicityfont{ \rho }  \scale( \ottnt{Q_{{\mathrm{1}}}}  \Lolly  \constraintfont{C'} )  }  $, then by
    \cref{lem:inversion}, there is a $  \constraintfont{ \ottnt{Q'} }  $ such that
    $  \constraintfont{ \ottnt{Q} }  =  \constraintfont{   \multiplicityfont{ \pi }  \scale \ottnt{Q'}  }  $ and $ \constraintfont{ \ottnt{Q'}  \qtensor  \ottnt{Q_{{\mathrm{1}}}} }   \vdash   \constraintfont{ \constraintfont{C'} } $. Applying
    rule~\rref*{C-Impl} with $  \multiplicityfont{  \pi {⋅} \rho  }  $, we get
    $ \constraintfont{   \multiplicityfont{ \ottsym{(}   \pi {⋅} \rho   \ottsym{)} }  \scale \ottnt{Q'}  }   \vdash   \constraintfont{   \multiplicityfont{ \ottsym{(}   \pi {⋅} \rho   \ottsym{)} }  \scale( \ottnt{Q_{{\mathrm{1}}}}  \Lolly  \constraintfont{C'} )  } $.

    In other words: $ \constraintfont{   \multiplicityfont{ \pi }  \scale \ottnt{Q}  }   \vdash   \constraintfont{   \multiplicityfont{ \pi }  \scale \ottsym{(}    \multiplicityfont{ \rho }  \scale( \ottnt{Q}  \Lolly  \constraintfont{C} )   \ottsym{)}  } $ as expected.
  \end{itemize}
\end{proof}

\begin{proof}[Proof of \cref{lem:wanted:demote}]
  By induction on the syntax of $  \constraintfont{ \constraintfont{C} }  $
  \begin{itemize}
  \item If $  \constraintfont{ \constraintfont{C} }  =  \constraintfont{ \ottnt{Q'} }  $, then the result follows from
    \cref{lem:q:scaling-inversion}
  \item If $  \constraintfont{ \constraintfont{C} }  =  \constraintfont{ \constraintfont{C_{{\mathrm{1}}}}  \qtensor  \constraintfont{C_{{\mathrm{2}}}} }  $, then we can prove the result like we
    proved the corresponding case in
    \cref{lem:q:scaling-inversion} using
    \cref{lem:inversion}.
  \item If $  \constraintfont{ \constraintfont{C} }  =  \constraintfont{ \constraintfont{C_{{\mathrm{1}}}}  \aand  \constraintfont{C_{{\mathrm{2}}}} }  $, then we the case where $  \multiplicityfont{ \pi }  =  \multiplicityfont{ \ottsym{1} }  $ is
    immediate, so we can assume without loss of generality that
    $  \multiplicityfont{ \pi }  =  \multiplicityfont{ \omega }  $, and, therefore, that
    $  \constraintfont{   \multiplicityfont{ \pi }  \scale \constraintfont{C}  }   =   \constraintfont{   \multiplicityfont{ \pi }  \scale \constraintfont{C_{{\mathrm{1}}}}   \qtensor    \multiplicityfont{ \pi }  \scale \constraintfont{C_{{\mathrm{2}}}}  }  $. By \cref{lem:inversion},
    there exist $  \constraintfont{ \ottnt{Q_{{\mathrm{1}}}} }  $ and $  \constraintfont{ \ottnt{Q_{{\mathrm{2}}}} }  $ such that $ \constraintfont{ \ottnt{Q_{{\mathrm{1}}}} }   \vdash   \constraintfont{   \multiplicityfont{ \omega }  \scale \constraintfont{C_{{\mathrm{1}}}}  } $,
    $ \constraintfont{ \ottnt{Q_{{\mathrm{2}}}} }   \vdash   \constraintfont{   \multiplicityfont{ \omega }  \scale \constraintfont{C_{{\mathrm{2}}}}  } $ and $  \constraintfont{ \ottnt{Q} }  =  \constraintfont{ \ottnt{Q_{{\mathrm{1}}}}  \qtensor  \ottnt{Q_{{\mathrm{2}}}} }  $. By induction
    hypothesis, we get $  \constraintfont{ \ottnt{Q_{{\mathrm{1}}}} }   =   \constraintfont{   \multiplicityfont{ \omega }  \scale \ottnt{Q'_{{\mathrm{1}}}}  }  $ and $  \constraintfont{ \ottnt{Q_{{\mathrm{2}}}} }   =   \constraintfont{   \multiplicityfont{ \omega }  \scale \ottnt{Q'_{{\mathrm{2}}}}  }  $
    such that $ \constraintfont{ \ottnt{Q'_{{\mathrm{1}}}} }   \vdash   \constraintfont{ \constraintfont{C_{{\mathrm{1}}}} } $ and $ \constraintfont{ \ottnt{Q'_{{\mathrm{2}}}} }   \vdash   \constraintfont{ \constraintfont{C_{{\mathrm{2}}}} } $. From which it
    follows that $ \constraintfont{   \multiplicityfont{ \omega }  \scale \ottnt{Q'_{{\mathrm{1}}}}   \qtensor    \multiplicityfont{ \omega }  \scale \ottnt{Q'_{{\mathrm{2}}}}  }   \vdash   \constraintfont{ \constraintfont{C_{{\mathrm{1}}}} } $ and
    $ \constraintfont{   \multiplicityfont{ \omega }  \scale \ottnt{Q'_{{\mathrm{1}}}}   \qtensor    \multiplicityfont{ \omega }  \scale \ottnt{Q'_{{\mathrm{2}}}}  }   \vdash   \constraintfont{ \constraintfont{C_{{\mathrm{1}}}} } $ (by
    \cref{lem:wanteds:weakening}) and, finally,
    $  \constraintfont{ \ottnt{Q} }  =  \constraintfont{   \multiplicityfont{ \omega }  \scale \ottnt{Q}  }  $ (by \cref{lem:wanteds:module-action})
    and $ \constraintfont{ \ottnt{Q} }   \vdash   \constraintfont{ \constraintfont{C_{{\mathrm{1}}}}  \aand  \constraintfont{C_{{\mathrm{2}}}} } $.
  \item If $  \constraintfont{ \constraintfont{C} }  =  \constraintfont{   \multiplicityfont{ \rho }  \scale( \ottnt{Q_{{\mathrm{1}}}}  \Lolly  \constraintfont{C'} )  }  $, then
    $  \constraintfont{   \multiplicityfont{ \pi }  \scale \constraintfont{C}  }   =   \constraintfont{   \multiplicityfont{ \ottsym{(}   \pi {⋅} \rho   \ottsym{)} }  \scale( \ottnt{Q_{{\mathrm{1}}}}  \Lolly  \constraintfont{C'} )  }  $. The result follows
    immediately by \cref{lem:inversion}.
  \end{itemize}
\end{proof}

\begin{proof}[Proof of \cref{lem:generation-soundness}]
  By induction on $\Gamma  \vdashi  \ottnt{e}  \ottsym{:}  \tau  \leadsto   \constraintfont{ \constraintfont{C} } $
  \begin{description}
  \item[\rref*{G-Var}] We have
    \begin{itemize}
    \item $\Gamma_{{\mathrm{1}}}  \ottsym{=}   \ottmv{x} {:}_{  \multiplicityfont{ \ottsym{1} }  }  \forall   \overline{\ottmv{a} } .   \constraintfont{ \ottnt{Q} }   \Lolly  \upsilon  $
    \item $\Gamma_{{\mathrm{1}}}  \ottsym{+}    \multiplicityfont{ \omega }  \scale \Gamma_{{\mathrm{2}}}   \vdashi  \ottmv{x}  \ottsym{:}  \upsilon  \ottsym{[}  \overline{\tau}  \ottsym{/}  \overline{\ottmv{a} }  \ottsym{]}  \leadsto   \constraintfont{ \ottnt{Q}  \ottsym{[}  \overline{\tau}  \ottsym{/}  \overline{\ottmv{a} }  \ottsym{]} } $
    \item $ \constraintfont{ \ottnt{Q_{\ottmv{g}}} }   \vdash   \constraintfont{ \ottnt{Q}  \ottsym{[}  \overline{\tau}  \ottsym{/}  \overline{\ottmv{a} }  \ottsym{]} } $
    \end{itemize}
    Therefore, by rules~\rref*{E-Var} and~\rref*{E-Sub}, it follows
    immediately that $ \constraintfont{ \ottnt{Q_{\ottmv{g}}} }   \ottsym{;}  \Gamma_{{\mathrm{1}}}  \ottsym{+}    \multiplicityfont{ \omega }  \scale \Gamma_{{\mathrm{2}}}   \vdash  \ottmv{x}  \ottsym{:}  \upsilon  \ottsym{[}  \overline{\tau}  \ottsym{/}  \overline{\ottmv{a} }  \ottsym{]}$
  \item[\rref*{G-Abs}] We have
    \begin{itemize}
    \item $\Gamma  \vdashi  \lambda  \ottmv{x}  \ottsym{.}  \ottnt{e}  \ottsym{:}   \tau_{{\mathrm{0}}}  \to_{  \multiplicityfont{ \pi }  }  \tau   \leadsto   \constraintfont{ \constraintfont{C} } $
    \item $ \constraintfont{ \ottnt{Q_{\ottmv{g}}} }   \vdash   \constraintfont{ \constraintfont{C} } $
    \item $\Gamma  \ottsym{,}   \ottmv{x} {:}_{  \multiplicityfont{ \pi }  } \tau_{{\mathrm{0}}}   \vdashi  \ottnt{e}  \ottsym{:}  \tau  \leadsto   \constraintfont{ \constraintfont{C} } $
    \end{itemize}
    By induction hypothesis we have
    \begin{itemize}
    \item $ \constraintfont{ \ottnt{Q_{\ottmv{g}}} }   \ottsym{;}  \Gamma  \ottsym{,}   \ottmv{x} {:}_{  \multiplicityfont{ \pi }  } \tau_{{\mathrm{0}}}   \vdash  \ottnt{e}  \ottsym{:}  \tau$
    \end{itemize}
    From which follows that $ \constraintfont{ \ottnt{Q_{\ottmv{g}}} }   \ottsym{;}  \Gamma  \vdash  \lambda  \ottmv{x}  \ottsym{.}  \ottnt{e}  \ottsym{:}   \tau_{{\mathrm{0}}}  \to_{  \multiplicityfont{ \pi }  }  \tau $.
  \item[\rref*{G-Let}] We have
    \begin{itemize}
    \item $  \multiplicityfont{ \pi }  \scale \Gamma_{{\mathrm{1}}}   \ottsym{+}  \Gamma_{{\mathrm{2}}}  \vdashi   \klet_  \multiplicityfont{ \pi }   \, \ottmv{x}  \ottsym{=}  \ottnt{e_{{\mathrm{1}}}} \, \ottkw{in} \, \ottnt{e_{{\mathrm{2}}}}  \ottsym{:}  \tau  \leadsto   \constraintfont{   \multiplicityfont{ \pi }  \scale \constraintfont{C_{{\mathrm{1}}}}   \qtensor  \constraintfont{C_{{\mathrm{2}}}} } $
    \item $ \constraintfont{ \ottnt{Q_{\ottmv{g}}} }   \vdash   \constraintfont{   \multiplicityfont{ \pi }  \scale \constraintfont{C_{{\mathrm{1}}}}   \qtensor  \constraintfont{C_{{\mathrm{2}}}} } $
    \item $\Gamma_{{\mathrm{2}}}  \ottsym{,}   \ottmv{x} {:}_{  \multiplicityfont{ \pi }  } \tau_{{\mathrm{1}}}   \vdashi  \ottnt{e_{{\mathrm{2}}}}  \ottsym{:}  \tau  \leadsto   \constraintfont{ \constraintfont{C_{{\mathrm{2}}}} } $
    \item $\Gamma_{{\mathrm{1}}}  \vdashi  \ottnt{e_{{\mathrm{1}}}}  \ottsym{:}  \tau_{{\mathrm{1}}}  \leadsto   \constraintfont{ \constraintfont{C_{{\mathrm{1}}}} } $
    \end{itemize}
    By \cref{lem:inversion,lem:wanted:demote}, there exist $  \constraintfont{ \ottnt{Q_{{\mathrm{1}}}} }  $,
    $  \constraintfont{ Q_{\mathcal{D} } }  $ and $  \constraintfont{ \ottnt{Q_{{\mathrm{2}}}} }  $ such that\info{Here I used,
      implicitly (and without proof), the easy fact that if there
      exists $  \constraintfont{ \ottnt{Q'} }  $ such that $  \constraintfont{ \ottnt{Q} }  =  \constraintfont{   \multiplicityfont{ \pi }  \scale \ottnt{Q'}  }  $, then
      $  \constraintfont{ \ottnt{Q} }  =  \constraintfont{   \multiplicityfont{ \pi }  \scale \ottnt{Q}  }  $ (because $  \multiplicityfont{  \pi {⋅} \pi  }   =   \multiplicityfont{ \pi }  $) and $ \constraintfont{ \ottnt{Q} }   \Vdash   \constraintfont{ \ottnt{Q'} } $ (identity for $  \multiplicityfont{ \pi }  =  \multiplicityfont{ \ottsym{1} }  $ and dereliction for $  \multiplicityfont{ \pi }  =  \multiplicityfont{ \omega }  $).}
    \begin{itemize}
    \item $ \constraintfont{ \ottnt{Q_{{\mathrm{1}}}}  \qtensor  Q_{\mathcal{D} } }   \vdash   \constraintfont{ \constraintfont{C_{{\mathrm{1}}}} } $
    \item $ \constraintfont{ Q_{\mathcal{D} }  \qtensor  \ottnt{Q_{{\mathrm{2}}}} }   \vdash   \constraintfont{ \constraintfont{C_{{\mathrm{2}}}} } $
    \item $  \constraintfont{ \ottnt{Q_{\ottmv{g}}} }   =   \constraintfont{   \multiplicityfont{ \pi }  \scale \ottnt{Q_{{\mathrm{1}}}}   \qtensor  Q_{\mathcal{D} }  \qtensor  \ottnt{Q_{{\mathrm{2}}}} }  $
    \item $  \constraintfont{ Q_{\mathcal{D} } }   \in  \constraintfont{\mathcal{D} } $
    \item $  \constraintfont{   \multiplicityfont{ \pi }  \scale Q_{\mathcal{D} }  }   =   \constraintfont{ Q_{\mathcal{D} } }  $
    \end{itemize}
    By induction hypothesis we have
    \begin{itemize}
    \item $ \constraintfont{ \ottnt{Q_{{\mathrm{1}}}}  \qtensor  Q_{\mathcal{D} } }   \ottsym{;}  \Gamma_{{\mathrm{1}}}  \vdash  \ottnt{e_{{\mathrm{1}}}}  \ottsym{:}  \tau_{{\mathrm{1}}}$
    \item $ \constraintfont{ Q_{\mathcal{D} }  \qtensor  \ottnt{Q_{{\mathrm{2}}}} }   \ottsym{;}  \Gamma_{{\mathrm{2}}}  \ottsym{,}   \ottmv{x} {:}_{  \multiplicityfont{ \pi }  } \tau_{{\mathrm{1}}}   \vdash  \ottnt{e_{{\mathrm{1}}}}  \ottsym{:}  \tau_{{\mathrm{1}}}$
    \end{itemize}
    From which follows that $ \constraintfont{ \ottnt{Q_{\ottmv{g}}} }   \ottsym{;}    \multiplicityfont{ \pi }  \scale \Gamma_{{\mathrm{1}}}   \ottsym{+}  \Gamma_{{\mathrm{2}}}  \vdash   \klet_  \multiplicityfont{ \pi }   \, \ottmv{x}  \ottsym{=}  \ottnt{e_{{\mathrm{1}}}} \, \ottkw{in} \, \ottnt{e_{{\mathrm{2}}}}  \ottsym{:}  \tau$.
  \item[\rref*{G-LetSig}] We have
    \begin{itemize}
    \item $  \multiplicityfont{ \pi }  \scale \Gamma_{{\mathrm{1}}}   \ottsym{+}  \Gamma_{{\mathrm{2}}}  \vdashi   \klet_  \multiplicityfont{ \pi }   \, \ottmv{x}  \ottsym{:}   \forall   \overline{\ottmv{a} } .   \constraintfont{ \ottnt{Q} }   \Lolly  \tau_{{\mathrm{1}}}   \ottsym{=}  \ottnt{e_{{\mathrm{1}}}} \, \ottkw{in} \, \ottnt{e_{{\mathrm{2}}}}  \ottsym{:}  \tau  \leadsto   \constraintfont{ \constraintfont{C_{{\mathrm{2}}}}  \qtensor    \multiplicityfont{ \pi }  \scale( \ottnt{Q}  \Lolly  \constraintfont{C_{{\mathrm{1}}}} )  } $
    \item $ \constraintfont{ \ottnt{Q_{\ottmv{g}}} }   \vdash   \constraintfont{ \constraintfont{C_{{\mathrm{2}}}}  \qtensor    \multiplicityfont{ \pi }  \scale( \ottnt{Q}  \Lolly  \constraintfont{C_{{\mathrm{1}}}} )  } $
    \item $\Gamma_{{\mathrm{1}}}  \vdashi  \ottnt{e_{{\mathrm{1}}}}  \ottsym{:}  \tau_{{\mathrm{1}}}  \leadsto   \constraintfont{ \constraintfont{C_{{\mathrm{1}}}} } $
    \item $\Gamma_{{\mathrm{2}}}  \ottsym{,}   \ottmv{x} {:}_{  \multiplicityfont{ \pi }  }  \forall   \overline{\ottmv{a} } .   \constraintfont{ \ottnt{Q} }   \Lolly  \tau_{{\mathrm{1}}}    \vdashi  \ottnt{e_{{\mathrm{2}}}}  \ottsym{:}  \tau  \leadsto   \constraintfont{ \constraintfont{C_{{\mathrm{2}}}} } $
    \end{itemize}
    By \cref{lem:inversion,lem:wanted:demote}, there exist $  \constraintfont{ \ottnt{Q_{{\mathrm{1}}}} }  $,
    $Q_{\mathcal{D} }$, $  \constraintfont{ \ottnt{Q_{{\mathrm{2}}}} }  $ such
    that
    \begin{itemize}
    \item $ \constraintfont{ Q_{\mathcal{D} }  \qtensor  \ottnt{Q_{{\mathrm{2}}}} }   \vdash   \constraintfont{ \constraintfont{C_{{\mathrm{2}}}} } $
    \item $ \constraintfont{ \ottnt{Q_{{\mathrm{1}}}}  \qtensor  Q_{\mathcal{D} }  \qtensor  \ottnt{Q} }   \vdash   \constraintfont{ \constraintfont{C} } $
    \item $  \constraintfont{ \ottnt{Q_{\ottmv{g}}} }   =   \constraintfont{   \multiplicityfont{ \pi }  \scale \ottnt{Q_{{\mathrm{1}}}}   \qtensor  Q_{\mathcal{D} }  \qtensor  \ottnt{Q_{{\mathrm{2}}}} }  $
    \item $  \constraintfont{ Q_{\mathcal{D} } }   \in  \constraintfont{\mathcal{D} } $
    \item $  \constraintfont{   \multiplicityfont{ \pi }  \scale Q_{\mathcal{D} }  }   =   \constraintfont{ Q_{\mathcal{D} } }  $
    \end{itemize}
    By induction hypothesis
    \begin{itemize}
    \item $ \constraintfont{ \ottnt{Q_{{\mathrm{1}}}}  \qtensor  Q_{\mathcal{D} }  \qtensor  \ottnt{Q} }   \ottsym{;}  \Gamma_{{\mathrm{1}}}  \vdash  \ottnt{e_{{\mathrm{1}}}}  \ottsym{:}  \tau_{{\mathrm{1}}}$
    \item $ \constraintfont{ Q_{\mathcal{D} }  \qtensor  \ottnt{Q_{{\mathrm{2}}}} }   \ottsym{;}  \Gamma_{{\mathrm{2}}}  \ottsym{,}   \ottmv{x} {:}_{  \multiplicityfont{ \pi }  }  \forall   \overline{\ottmv{a} } .   \constraintfont{ \ottnt{Q} }   \Lolly  \tau_{{\mathrm{1}}}    \vdash  \ottnt{e_{{\mathrm{2}}}}  \ottsym{:}  \tau$
    \end{itemize}
    Hence $ \constraintfont{ \ottnt{Q_{\ottmv{g}}} }   \ottsym{;}    \multiplicityfont{ \pi }  \scale \Gamma_{{\mathrm{1}}}   \ottsym{+}  \Gamma_{{\mathrm{2}}}  \vdash   \klet_  \multiplicityfont{ \pi }   \, \ottmv{x}  \ottsym{:}   \forall   \overline{\ottmv{a} } .   \constraintfont{ \ottnt{Q} }   \Lolly  \tau_{{\mathrm{1}}}   \ottsym{=}  \ottnt{e_{{\mathrm{1}}}} \, \ottkw{in} \, \ottnt{e_{{\mathrm{2}}}}  \ottsym{:}  \tau$
  \item[\rref*{G-App}] \info{Most of the linearity problems are in the App
      rule. Unpack is also relevant.}
    We have
    \begin{itemize}
    \item $\Gamma_{{\mathrm{1}}}  \ottsym{+}    \multiplicityfont{ \pi }  \scale \Gamma_{{\mathrm{2}}}   \vdashi  \ottnt{e_{{\mathrm{1}}}} \, \ottnt{e_{{\mathrm{2}}}}  \ottsym{:}  \tau  \leadsto   \constraintfont{ \constraintfont{C_{{\mathrm{1}}}}  \qtensor    \multiplicityfont{ \pi }  \scale \constraintfont{C_{{\mathrm{2}}}}  } $
    \item $ \constraintfont{ \ottnt{Q_{\ottmv{g}}} }   \vdash   \constraintfont{ \constraintfont{C_{{\mathrm{1}}}}  \qtensor    \multiplicityfont{ \pi }  \scale \constraintfont{C_{{\mathrm{2}}}}  } $
    \item $\Gamma_{{\mathrm{1}}}  \vdashi  \ottnt{e_{{\mathrm{1}}}}  \ottsym{:}   \tau_{{\mathrm{2}}}  \to_{  \multiplicityfont{ \pi }  }  \tau   \leadsto   \constraintfont{ \constraintfont{C_{{\mathrm{1}}}} } $
    \item $\Gamma_{{\mathrm{2}}}  \vdashi  \ottnt{e_{{\mathrm{2}}}}  \ottsym{:}  \tau_{{\mathrm{2}}}  \leadsto   \constraintfont{ \constraintfont{C_{{\mathrm{2}}}} } $
    \end{itemize}
    By \cref{lem:inversion,lem:wanted:demote}, there exist
    $  \constraintfont{ \ottnt{Q_{{\mathrm{1}}}} }  $, $  \constraintfont{ Q_{\mathcal{D} } }  $, $  \constraintfont{ \ottnt{Q_{{\mathrm{2}}}} }  $ such that
    \begin{itemize}
    \item $ \constraintfont{ \ottnt{Q_{{\mathrm{1}}}}  \qtensor  Q_{\mathcal{D} } }   \vdash   \constraintfont{ \constraintfont{C_{{\mathrm{1}}}} } $
    \item $ \constraintfont{ Q_{\mathcal{D} }  \qtensor  \ottnt{Q_{{\mathrm{2}}}} }   \vdash   \constraintfont{ \constraintfont{C_{{\mathrm{2}}}} } $
    \item $  \constraintfont{ \ottnt{Q_{\ottmv{g}}} }   =   \constraintfont{ \ottnt{Q_{{\mathrm{1}}}}  \qtensor  Q_{\mathcal{D} }  \qtensor    \multiplicityfont{ \pi }  \scale \ottnt{Q_{{\mathrm{2}}}}  }  $
    \item $  \constraintfont{ Q_{\mathcal{D} } }   \in  \constraintfont{\mathcal{D} } $
    \item $  \constraintfont{   \multiplicityfont{ \pi }  \scale Q_{\mathcal{D} }  }   =   \constraintfont{ Q_{\mathcal{D} } }  $
    \end{itemize}
    By induction hypothesis
    \begin{itemize}
    \item $ \constraintfont{ \ottnt{Q_{{\mathrm{1}}}}  \qtensor  Q_{\mathcal{D} } }   \ottsym{;}  \Gamma_{{\mathrm{1}}}  \vdash  \ottnt{e_{{\mathrm{1}}}}  \ottsym{:}   \tau_{{\mathrm{2}}}  \to_{  \multiplicityfont{ \pi }  }  \tau $
    \item $ \constraintfont{ Q_{\mathcal{D} }  \qtensor  \ottnt{Q_{{\mathrm{2}}}} }   \ottsym{;}  \Gamma_{{\mathrm{2}}}  \vdash  \ottnt{e_{{\mathrm{2}}}}  \ottsym{:}  \tau_{{\mathrm{2}}}$
    \end{itemize}
    Hence $ \constraintfont{ \ottnt{Q_{\ottmv{g}}} }   \ottsym{;}  \Gamma_{{\mathrm{1}}}  \ottsym{+}    \multiplicityfont{ \pi }  \scale \Gamma_{{\mathrm{2}}}   \vdash  \ottnt{e_{{\mathrm{1}}}} \, \ottnt{e_{{\mathrm{2}}}}  \ottsym{:}  \tau$.
  \item[\rref*{G-Pack}] We have
    \begin{itemize}
    \item $\Gamma  \vdashi  \packbox \, \ottnt{e}  \ottsym{:}   \exists   \overline{\ottmv{a} } .  \tau  \RLolly   \constraintfont{ \ottnt{Q} }    \leadsto   \constraintfont{ \constraintfont{C}  \qtensor  \ottnt{Q}  \ottsym{[}  \overline{\upsilon}  \ottsym{/}  \overline{\ottmv{a} }  \ottsym{]} } $
    \item $ \constraintfont{ \ottnt{Q_{\ottmv{g}}} }   \vdash   \constraintfont{ \constraintfont{C}  \qtensor  \ottnt{Q}  \ottsym{[}  \overline{\upsilon}  \ottsym{/}  \overline{\ottmv{a} }  \ottsym{]} } $
    \item $\Gamma  \vdashi  \ottnt{e}  \ottsym{:}  \tau  \ottsym{[}  \overline{\upsilon}  \ottsym{/}  \overline{\ottmv{a} }  \ottsym{]}  \leadsto   \constraintfont{ \constraintfont{C} } $
    \end{itemize}
    By \cref{lem:inversion}, there exist $  \constraintfont{ \ottnt{Q_{{\mathrm{1}}}} }  $, $  \constraintfont{ Q_{\mathcal{D} } }  $, $  \constraintfont{ \ottnt{Q_{{\mathrm{2}}}} }  $
    such that
    \begin{itemize}
    \item $ \constraintfont{ \ottnt{Q_{{\mathrm{1}}}}  \qtensor  Q_{\mathcal{D} } }   \vdash   \constraintfont{ \constraintfont{C} } $
    \item $ \constraintfont{ Q_{\mathcal{D} }  \qtensor  \ottnt{Q_{{\mathrm{2}}}} }   \vdash   \constraintfont{ \ottnt{Q}  \ottsym{[}  \overline{\upsilon}  \ottsym{/}  \overline{\ottmv{a} }  \ottsym{]} } $
    \item $  \constraintfont{ \ottnt{Q_{\ottmv{g}}} }   =   \constraintfont{ \ottnt{Q_{{\mathrm{1}}}}  \qtensor  Q_{\mathcal{D} }  \qtensor  \ottnt{Q_{{\mathrm{2}}}} }  $
    \item $  \constraintfont{ Q_{\mathcal{D} } }   \in  \constraintfont{\mathcal{D} } $
    \end{itemize}
    By induction hypothesis
    \begin{itemize}
    \item $ \constraintfont{ \ottnt{Q_{{\mathrm{1}}}}  \qtensor  Q_{\mathcal{D} } }   \ottsym{;}  \Gamma  \vdash  \ottnt{e}  \ottsym{:}  \tau  \ottsym{[}  \overline{\upsilon}  \ottsym{/}  \overline{\ottmv{a} }  \ottsym{]}$
    \end{itemize}
    So we have $ \constraintfont{ \ottnt{Q_{{\mathrm{1}}}}  \qtensor  Q_{\mathcal{D} }  \qtensor  \ottnt{Q}  \ottsym{[}  \overline{\upsilon}  \ottsym{/}  \overline{\ottmv{a} }  \ottsym{]} }   \ottsym{;}  \Gamma  \vdash  \packbox \, \ottnt{e}  \ottsym{:}   \exists   \overline{\ottmv{a} } .  \tau  \RLolly   \constraintfont{ \ottnt{Q} }  $. By~\cref{lem:dup-contraction} rule~\rref*{E-Sub}, we conclude
    $ \constraintfont{ \ottnt{Q_{\ottmv{g}}} }   \ottsym{;}    \multiplicityfont{ \omega }  \scale \Gamma   \vdash  \packbox \, \ottnt{e}  \ottsym{:}   \exists   \overline{\ottmv{a} } .  \tau  \RLolly   \constraintfont{ \ottnt{Q} }  $.
  \item[\rref*{G-Unpack}] We have
    \begin{itemize}
    \item $\Gamma_{{\mathrm{1}}}  \ottsym{+}  \Gamma_{{\mathrm{2}}}  \vdashi   \klet\ \packbox  \ottmv{x}  =  \ottnt{e_{{\mathrm{1}}}}  \  \ottkw{in}  \  \ottnt{e_{{\mathrm{2}}}}   \ottsym{:}  \tau  \leadsto   \constraintfont{ \constraintfont{C_{{\mathrm{1}}}}  \qtensor    \multiplicityfont{ \ottsym{1} }  \scale( \ottnt{Q'}  \Lolly  \constraintfont{C_{{\mathrm{2}}}} )  } $
    \item $ \constraintfont{ \ottnt{Q_{\ottmv{g}}} }   \vdash   \constraintfont{ \constraintfont{C_{{\mathrm{1}}}}  \qtensor    \multiplicityfont{ \ottsym{1} }  \scale( \ottnt{Q'}  \Lolly  \constraintfont{C_{{\mathrm{2}}}} )  } $
    \item $\Gamma_{{\mathrm{1}}}  \vdashi  \ottnt{e_{{\mathrm{1}}}}  \ottsym{:}   \exists   \overline{\ottmv{a} } .  \tau_{{\mathrm{1}}}  \RLolly   \constraintfont{ \ottnt{Q'} }    \leadsto   \constraintfont{ \constraintfont{C_{{\mathrm{1}}}} } $
    \item $\Gamma_{{\mathrm{2}}}  \ottsym{,}   \ottmv{x} {:}_{  \multiplicityfont{ \pi }  } \tau_{{\mathrm{1}}}   \vdashi  \ottnt{e_{{\mathrm{2}}}}  \ottsym{:}  \tau  \leadsto   \constraintfont{ \constraintfont{C_{{\mathrm{2}}}} } $
    \end{itemize}
    By \cref{lem:inversion}, there exist $  \constraintfont{ \ottnt{Q_{{\mathrm{1}}}} }  $, $  \constraintfont{ Q_{\mathcal{D} } }  $, $  \constraintfont{ \ottnt{Q_{{\mathrm{2}}}} }  $
    such that
    \begin{itemize}
    \item $ \constraintfont{ \ottnt{Q_{{\mathrm{1}}}}  \qtensor  Q_{\mathcal{D} } }   \vdash   \constraintfont{ \constraintfont{C_{{\mathrm{1}}}} } $
    \item $ \constraintfont{ Q_{\mathcal{D} }  \qtensor  \ottnt{Q_{{\mathrm{2}}}}  \qtensor  \ottnt{Q'} }   \vdash   \constraintfont{ \constraintfont{C_{{\mathrm{2}}}} } $
    \item $  \constraintfont{ \ottnt{Q_{\ottmv{g}}} }   =   \constraintfont{ \ottnt{Q_{{\mathrm{1}}}}  \qtensor  Q_{\mathcal{D} }  \qtensor  \ottnt{Q_{{\mathrm{2}}}} }  $
    \item $  \constraintfont{ Q_{\mathcal{D} } }   \in  \constraintfont{\mathcal{D} } $
    \end{itemize}
    By induction hypothesis
    \begin{itemize}
    \item $ \constraintfont{ \ottnt{Q_{{\mathrm{1}}}}  \qtensor  Q_{\mathcal{D} } }   \ottsym{;}  \Gamma_{{\mathrm{1}}}  \vdash  \ottnt{e_{{\mathrm{1}}}}  \ottsym{:}   \exists   \overline{\ottmv{a} } .  \tau_{{\mathrm{1}}}  \RLolly   \constraintfont{ \ottnt{Q'} }  $
    \item $ \constraintfont{ Q_{\mathcal{D} }  \qtensor  \ottnt{Q_{{\mathrm{2}}}}  \qtensor  \ottnt{Q} }   \ottsym{;}  \Gamma_{{\mathrm{2}}}  \vdash  \ottnt{e_{{\mathrm{2}}}}  \ottsym{:}  \tau$
    \end{itemize}
    Therefore $ \constraintfont{ \ottnt{Q_{\ottmv{g}}} }   \ottsym{;}  \Gamma_{{\mathrm{1}}}  \ottsym{+}  \Gamma_{{\mathrm{2}}}  \vdash   \klet\ \packbox  \ottmv{x}  =  \ottnt{e_{{\mathrm{1}}}}  \  \ottkw{in}  \  \ottnt{e_{{\mathrm{2}}}}   \ottsym{:}  \tau$.
  \item[\rref*{G-Case}] We have
    \begin{itemize}
    \item $  \multiplicityfont{ \pi }  \scale \Gamma   \ottsym{+}  \Delta  \vdashi   \kcase_  \multiplicityfont{ \pi }   \, \ottnt{e} \, \ottkw{of} \, \ottsym{\{}  \overline{\ottmv{K}_i\ \overline{\ottmv{x}_i } \to \ottnt{e}_i }  \ottsym{\}}  \ottsym{:}  \tau  \leadsto   \constraintfont{   \multiplicityfont{ \pi }  \scale \constraintfont{C}   \qtensor  \bigaand  \constraintfont{C_{\ottmv{i}}} } $
    \item $ \constraintfont{ \ottnt{Q_{\ottmv{g}}} }   \vdash   \constraintfont{   \multiplicityfont{ \pi }  \scale \constraintfont{C}   \qtensor  \bigaand  \constraintfont{C_{\ottmv{i}}} } $
    \item $\Gamma  \vdashi  \ottnt{e}  \ottsym{:}  \ottmv{T} \, \overline{\sigma}  \leadsto   \constraintfont{ \constraintfont{C} } $
    \item For each $i$, $\Delta  \ottsym{,}   \overline{  \ottmv{x_{\ottmv{i}}} {:}_{  \multiplicityfont{ \ottsym{(}   \pi {⋅} \pi_{\ottmv{i}}   \ottsym{)} }  } \upsilon_{\ottmv{i}}  \ottsym{[}  \overline{\sigma}  \ottsym{/}  \overline{\ottmv{a} }  \ottsym{]}  }   \vdashi  \ottnt{e_{\ottmv{i}}}  \ottsym{:}  \tau  \leadsto   \constraintfont{ \constraintfont{C_{\ottmv{i}}} } $
    \end{itemize}
    By repeated uses of \cref{lem:inversion} as well as \cref{lem:wanted:demote}, there exist
    $  \constraintfont{ \ottnt{Q} }  $, $  \constraintfont{ Q_{\mathcal{D} } }  $, $  \constraintfont{ \ottnt{Q'} }  $ such that
    \begin{itemize}
    \item $ \constraintfont{ \ottnt{Q}  \qtensor  Q_{\mathcal{D} } }   \vdash   \constraintfont{ \constraintfont{C} } $
    \item For each $i$, $ \constraintfont{ Q_{\mathcal{D} }  \qtensor  \ottnt{Q'} }   \vdash   \constraintfont{ \constraintfont{C_{\ottmv{i}}} } $
    \item $  \constraintfont{ \ottnt{Q_{\ottmv{g}}} }   =   \constraintfont{   \multiplicityfont{ \pi }  \scale \ottnt{Q}   \qtensor  Q_{\mathcal{D} }  \qtensor  \ottnt{Q'} }  $
    \item $  \constraintfont{ Q_{\mathcal{D} } }   \in  \constraintfont{\mathcal{D} } $
    \item $  \constraintfont{   \multiplicityfont{ \pi }  \scale Q_{\mathcal{D} }  }   =   \constraintfont{ Q_{\mathcal{D} } }  $
    \end{itemize}
    By induction hypothesis
    \begin{itemize}
    \item $ \constraintfont{ \ottnt{Q}  \qtensor  Q_{\mathcal{D} } }   \ottsym{;}  \Gamma  \vdash  \ottnt{e}  \ottsym{:}  \ottmv{T} \, \overline{\sigma}$
    \item For each $i$, $ \constraintfont{ Q_{\mathcal{D} }  \qtensor  \ottnt{Q'} }   \ottsym{;}  \Delta  \ottsym{,}   \overline{  \ottmv{x_{\ottmv{i}}} {:}_{  \multiplicityfont{ \ottsym{(}   \pi {⋅} \pi_{\ottmv{i}}   \ottsym{)} }  } \upsilon_{\ottmv{i}}  \ottsym{[}  \overline{\sigma}  \ottsym{/}  \overline{\ottmv{a} }  \ottsym{]}  }   \vdash  \ottnt{e_{\ottmv{i}}}  \ottsym{:}  \tau$
    \end{itemize}
    Therefore $ \constraintfont{ \ottnt{Q_{\ottmv{g}}} }   \ottsym{;}    \multiplicityfont{ \pi }  \scale \Gamma   \ottsym{+}  \Delta  \vdash   \kcase_  \multiplicityfont{ \pi }   \, \ottnt{e} \, \ottkw{of} \, \ottsym{\{}  \overline{\ottmv{K}_i\ \overline{\ottmv{x}_i } \to \ottnt{e}_i }  \ottsym{\}}  \ottsym{:}  \tau$.
  \end{description}
\end{proof}

\begin{proof}[Proof of \cref{lem:solver-soundness}]
  By induction on $ \constraintfont{ \ottnt{U} }   \ottsym{;}   \constraintfont{ \ottnt{D} }   \ottsym{;}   \constraintfont{ \ottnt{L}_{\ottmv{i}} }   \vdashs   \constraintfont{ \constraintfont{C} }   \leadsto   \constraintfont{ \ottnt{L}_{\ottmv{o}} } $
  \begin{description}
  \item[\rref*{S-Atom}] We have
  \begin{itemize}
          \item $ \constraintfont{ \ottnt{U} }   \ottsym{;}   \constraintfont{ \ottnt{D} }   \ottsym{;}   \constraintfont{ \ottnt{L}_{\ottmv{i}} }   \vdashs   \constraintfont{   \multiplicityfont{ \pi }  \scale \ottnt{q}  }   \leadsto   \constraintfont{ \ottnt{L}_{\ottmv{o}} } $
          \item $  \constraintfont{ \ottnt{U} }  ;   \constraintfont{ \ottnt{D} }  ;   \constraintfont{ \ottnt{L}_{\ottmv{i}} }    \vdashsimp    \multiplicityfont{ \pi }   \scale   \constraintfont{ \ottnt{q} }    \leadsto    \constraintfont{ \ottnt{L}_{\ottmv{o}} }  $
  \end{itemize}
  By \cref{prop:atomic-solver-soundness} we have
  \begin{enumerate}
  \item $  \constraintfont{ \ottnt{L}_{\ottmv{o}} }  \subseteq   \constraintfont{ \ottnt{L}_{\ottmv{i}} }  $
  \item $ \constraintfont{ \ottsym{(}  \ottnt{U}  \ottsym{,}   \ottnt{D}  \uplus  \ottnt{L}_{\ottmv{i}}   \ottsym{)} }   \Vdash   \constraintfont{   \multiplicityfont{ \pi }  \scale \ottnt{q}   \qtensor  \ottsym{(}  \emptyset  \ottsym{,}  \ottnt{L}_{\ottmv{o}}  \ottsym{)} } $
  \end{enumerate}
  Then by \rref*{C-Dom} we have $ \constraintfont{ \ottsym{(}  \ottnt{U}  \ottsym{,}  \ottnt{L}_{\ottmv{i}}  \ottsym{)} }   \vdash   \constraintfont{   \multiplicityfont{ \pi }  \scale \ottnt{q}   \qtensor  \ottsym{(}  \emptyset  \ottsym{,}  \ottnt{L}_{\ottmv{o}}  \ottsym{)} } $.
  \item[\rref*{S-Add}] We have
  \begin{itemize}
          \item $ \constraintfont{ \ottnt{U} }   \ottsym{;}   \constraintfont{ \ottnt{D} }   \ottsym{;}   \constraintfont{ \ottnt{L}_{\ottmv{i}} }   \vdashs   \constraintfont{ \constraintfont{C_{{\mathrm{1}}}}  \aand  \constraintfont{C_{{\mathrm{2}}}} }   \leadsto   \constraintfont{ \ottnt{L}_{\ottmv{o}} } $
          \item $ \constraintfont{ \ottnt{U} }   \ottsym{;}   \constraintfont{ \ottnt{D} }   \ottsym{;}   \constraintfont{ \ottnt{L}_{\ottmv{i}} }   \vdashs   \constraintfont{ \constraintfont{C_{{\mathrm{1}}}} }   \leadsto   \constraintfont{ \ottnt{L}_{\ottmv{o}} } $
          \item $ \constraintfont{ \ottnt{U} }   \ottsym{;}   \constraintfont{ \ottnt{D} }   \ottsym{;}   \constraintfont{ \ottnt{L}_{\ottmv{i}} }   \vdashs   \constraintfont{ \constraintfont{C_{{\mathrm{2}}}} }   \leadsto   \constraintfont{ \ottnt{L}_{\ottmv{o}} } $
  \end{itemize}
  By induction hypothesis we have
  \begin{itemize}
          \item $  \constraintfont{ \ottnt{L}_{\ottmv{o}} }  \subseteq   \constraintfont{ \ottnt{L}_{\ottmv{i}} }  $
          \item $ \constraintfont{ \ottsym{(}  \ottnt{U}  \ottsym{,}   \ottnt{D}  \uplus  \ottnt{L}_{\ottmv{i}}   \ottsym{)} }   \vdash   \constraintfont{ \constraintfont{C_{{\mathrm{1}}}}  \qtensor  \ottsym{(}  \emptyset  \ottsym{,}  \ottnt{L}_{\ottmv{o}}  \ottsym{)} } $
          \item $ \constraintfont{ \ottsym{(}  \ottnt{U}  \ottsym{,}   \ottnt{D}  \uplus  \ottnt{L}_{\ottmv{i}}   \ottsym{)} }   \vdash   \constraintfont{ \constraintfont{C_{{\mathrm{2}}}}  \qtensor  \ottsym{(}  \emptyset  \ottsym{,}  \ottnt{L}_{\ottmv{o}}  \ottsym{)} } $
  \end{itemize}
  Then by \rref*{C-With} we have $ \constraintfont{ \ottsym{(}  \ottnt{U}  \ottsym{,}   \ottnt{D}  \uplus  \ottnt{L}_{\ottmv{i}}   \ottsym{)} }   \vdash   \constraintfont{ \constraintfont{C_{{\mathrm{1}}}}  \aand  \constraintfont{C_{{\mathrm{2}}}}  \qtensor  \ottsym{(}  \emptyset  \ottsym{,}  \ottnt{L}_{\ottmv{o}}  \ottsym{)} } $.
  \item[\rref*{S-Mult}] We have
  \begin{itemize}
          \item $ \constraintfont{ \ottnt{U} }   \ottsym{;}   \constraintfont{ \ottnt{D} }   \ottsym{;}   \constraintfont{ \ottnt{L}_{\ottmv{i}} }   \vdashs   \constraintfont{ \constraintfont{C_{{\mathrm{1}}}}  \qtensor  \constraintfont{C_{{\mathrm{2}}}} }   \leadsto   \constraintfont{ \ottnt{L}_{\ottmv{o}} } $
          \item $ \constraintfont{ \ottnt{U} }   \ottsym{;}   \constraintfont{ \ottnt{D} }   \ottsym{;}   \constraintfont{ \ottnt{L}_{\ottmv{i}} }   \vdashs   \constraintfont{ \constraintfont{C_{{\mathrm{1}}}} }   \leadsto   \constraintfont{ \ottnt{L}'_{\ottmv{o}} } $
          \item $ \constraintfont{ \ottnt{U} }   \ottsym{;}   \constraintfont{ \ottnt{D} }   \ottsym{;}   \constraintfont{ \ottnt{L}'_{\ottmv{o}} }   \vdashs   \constraintfont{ \constraintfont{C_{{\mathrm{2}}}} }   \leadsto   \constraintfont{ \ottnt{L}_{\ottmv{o}} } $
  \end{itemize}
  By induction hypothesis we have
  \begin{itemize}
          \item $  \constraintfont{ \ottnt{L}_{\ottmv{o}} }  \subseteq   \constraintfont{ \ottnt{L}'_{\ottmv{o}} }  $
          \item $  \constraintfont{ \ottnt{L}'_{\ottmv{o}} }  \subseteq   \constraintfont{ \ottnt{L}_{\ottmv{i}} }  $
          \item $ \constraintfont{ \ottsym{(}  \ottnt{U}  \ottsym{,}   \ottnt{D}  \uplus  \ottnt{L}_{\ottmv{i}}   \ottsym{)} }   \vdash   \constraintfont{ \constraintfont{C_{{\mathrm{1}}}}  \qtensor  \ottsym{(}  \emptyset  \ottsym{,}  \ottnt{L}'_{\ottmv{o}}  \ottsym{)} } $
          \item $ \constraintfont{ \ottsym{(}  \ottnt{U}  \ottsym{,}   \ottnt{D}  \uplus  \ottnt{L}'_{\ottmv{o}}   \ottsym{)} }   \vdash   \constraintfont{ \constraintfont{C_{{\mathrm{2}}}}  \qtensor  \ottsym{(}  \emptyset  \ottsym{,}  \ottnt{L}_{\ottmv{o}}  \ottsym{)} } $
  \end{itemize}
  Then
  \begin{itemize}
  \item by transitivity of $\subseteq$ we have
    $  \constraintfont{ \ottnt{L}_{\ottmv{o}} }  \subseteq   \constraintfont{ \ottnt{L}_{\ottmv{i}} }  $, and by \rref*{C-Tensor} we
    have
    $ \constraintfont{ \ottsym{(}  \ottnt{U}  \ottsym{,}   \ottnt{D}  \uplus  \ottnt{L}_{\ottmv{i}}   \ottsym{)}  \qtensor  \ottsym{(}  \ottnt{U}  \ottsym{,}   \ottnt{D}  \uplus  \ottnt{L}'_{\ottmv{o}}   \ottsym{)} }   \vdash   \constraintfont{ \constraintfont{C_{{\mathrm{1}}}}  \qtensor  \constraintfont{C_{{\mathrm{2}}}}  \qtensor  \ottsym{(}  \emptyset  \ottsym{,}  \ottnt{L}'_{\ottmv{o}}  \ottsym{)}  \qtensor  \ottsym{(}  \emptyset  \ottsym{,}  \ottnt{L}_{\ottmv{o}}  \ottsym{)} } $
  \item by~\ref{lem:dup-contraction} and the definition of tensor on
    unrestricted constraints, we have
    $ \constraintfont{ \ottsym{(}  \ottnt{U}  \ottsym{,}   \ottnt{D}  \uplus  \ottnt{L}_{\ottmv{i}}   \ottsym{)}  \qtensor  \ottsym{(}  \emptyset  \ottsym{,}  \ottnt{L}'_{\ottmv{o}}  \ottsym{)} }   \vdash   \constraintfont{ \constraintfont{C_{{\mathrm{1}}}}  \qtensor  \constraintfont{C_{{\mathrm{2}}}}  \qtensor  \ottsym{(}  \emptyset  \ottsym{,}  \ottnt{L}'_{\ottmv{o}}  \ottsym{)}  \qtensor  \ottsym{(}  \emptyset  \ottsym{,}  \ottnt{L}_{\ottmv{o}}  \ottsym{)} } $
  \item by~\cref{lem:inversion} we have
    $ \constraintfont{ \ottsym{(}  \ottnt{U}  \ottsym{,}  \ottnt{L}_{\ottmv{i}}  \ottsym{)} }   \vdash   \constraintfont{ \constraintfont{C_{{\mathrm{1}}}}  \qtensor  \constraintfont{C_{{\mathrm{2}}}}  \qtensor  \ottsym{(}  \emptyset  \ottsym{,}  \ottnt{L}_{\ottmv{o}}  \ottsym{)} } $.
  \end{itemize}
  \item[\rref*{S-ImplOne}] We have
  \begin{itemize}
          \item $ \constraintfont{ \ottnt{U} }   \ottsym{;}   \constraintfont{ \ottnt{D} }   \ottsym{;}   \constraintfont{ \ottnt{L}_{\ottmv{i}} }   \vdashs   \constraintfont{   \multiplicityfont{ \ottsym{1} }  \scale( \ottsym{(}  \ottnt{U}_{{\mathrm{0}}}  \ottsym{,}  \ottnt{L}_{{\mathrm{0}}}  \ottsym{)}  \Lolly  \constraintfont{C} )  }   \leadsto   \constraintfont{ \ottnt{L}_{\ottmv{o}} } $
          \item $ \constraintfont{  \ottnt{U}  \cup  \ottnt{U}_{{\mathrm{0}}}  }   \ottsym{;}   \constraintfont{  \ottnt{D}  \uplus  \ottsym{(}   \ottnt{L}_{{\mathrm{0}}}  \cap  \constraintfont{\mathcal{D} }   \ottsym{)}  }   \ottsym{;}   \constraintfont{  \ottnt{L}_{\ottmv{i}}  \uplus  \ottsym{(}   \ottnt{L}_{{\mathrm{0}}}  \setminus  \constraintfont{\mathcal{D} }   \ottsym{)}  }   \vdashs   \constraintfont{ \constraintfont{C} }   \leadsto   \constraintfont{ \ottnt{L}_{\ottmv{o}} } $
          \item $ \constraintfont{ \ottnt{L}_{\ottmv{o}} }   \subseteq   \constraintfont{ \ottnt{L}_{\ottmv{i}} } $
  \end{itemize}
  By induction hypothesis we have
  \begin{itemize}
          \item $ \constraintfont{ \ottsym{(}   \ottnt{U}  \cup  \ottnt{U}_{{\mathrm{0}}}   \ottsym{,}   \ottnt{D}  \uplus   \ottnt{L}_{\ottmv{i}}  \uplus  \ottnt{L}_{{\mathrm{0}}}    \ottsym{)} }   \vdash   \constraintfont{ \constraintfont{C}  \qtensor  \ottsym{(}  \emptyset  \ottsym{,}  \ottnt{L}_{\ottmv{o}}  \ottsym{)} } $
          \item $ \constraintfont{ \ottnt{L}_{\ottmv{o}} }   \subseteq   \constraintfont{  \ottnt{L}_{\ottmv{i}}  \uplus  \ottnt{L}_{{\mathrm{0}}}  } $
  \end{itemize}
  Then we know that
  $  \constraintfont{ \ottsym{(}  \emptyset  \ottsym{,}  \ottnt{L}_{\ottmv{i}}  \ottsym{)} }   =   \constraintfont{ \ottsym{(}  \emptyset  \ottsym{,}  \ottnt{L}_{\ottmv{o}}  \ottsym{)}  \qtensor  \ottsym{(}  \emptyset  \ottsym{,}  \ottnt{L}'_{\ottmv{i}}  \ottsym{)} }  $ for
  some $  \constraintfont{ \ottnt{L}'_{\ottmv{i}} }  $.
  Then by \cref{lem:inversion} we know that $ \constraintfont{ \ottsym{(}   \ottnt{U}  \cup  \ottnt{U}_{{\mathrm{0}}}   \ottsym{,}   \ottnt{D}  \uplus   \ottnt{L}'_{\ottmv{i}}  \uplus  \ottnt{L}_{{\mathrm{0}}}    \ottsym{)} }   \vdash   \constraintfont{ \constraintfont{C} } $ and by
  \rref*{C-Impl} we have $ \constraintfont{ \ottsym{(}  \ottnt{U}  \ottsym{,}  \ottnt{L}'_{\ottmv{i}}  \ottsym{)} }   \vdash   \constraintfont{   \multiplicityfont{ \ottsym{1} }  \scale( \ottsym{(}  \ottnt{U}_{{\mathrm{0}}}  \ottsym{,}  \ottnt{L}_{{\mathrm{0}}}  \ottsym{)}  \Lolly  \constraintfont{C} )  } $.
  Finally, by \rref*{C-Tensor} we conclude that $ \constraintfont{ \ottsym{(}  \ottnt{U}  \ottsym{,}  \ottnt{L}_{\ottmv{i}}  \ottsym{)} }   \vdash   \constraintfont{   \multiplicityfont{ \ottsym{1} }  \scale( \ottsym{(}  \ottnt{U}_{{\mathrm{0}}}  \ottsym{,}  \ottnt{L}_{{\mathrm{0}}}  \ottsym{)}  \Lolly  \constraintfont{C} )   \qtensor  \ottsym{(}  \emptyset  \ottsym{,}  \ottnt{L}_{\ottmv{o}}  \ottsym{)} } $
  \item[\rref*{S-ImplMany}] We have
  \begin{itemize}
          \item $ \constraintfont{ \ottnt{U} }   \ottsym{;}   \constraintfont{ \ottnt{D} }   \ottsym{;}   \constraintfont{ \ottnt{L}_{\ottmv{i}} }   \vdashs   \constraintfont{   \multiplicityfont{ \omega }  \scale( \ottsym{(}  \ottnt{U}_{{\mathrm{0}}}  \ottsym{,}  \ottnt{L}_{{\mathrm{0}}}  \ottsym{)}  \Lolly  \constraintfont{C} )  }   \leadsto   \constraintfont{ \ottnt{L}_{\ottmv{i}} } $
          \item $ \constraintfont{  \ottnt{U}  \cup  \ottnt{U}_{{\mathrm{0}}}  }   \ottsym{;}   \constraintfont{  \ottnt{L}_{{\mathrm{0}}}  \cap  \constraintfont{\mathcal{D} }  }   \ottsym{;}   \constraintfont{  \ottnt{L}_{{\mathrm{0}}}  \setminus  \constraintfont{\mathcal{D} }  }   \vdashs   \constraintfont{ \constraintfont{C} }   \leadsto   \constraintfont{ \emptyset } $
  \end{itemize}
  By induction hypothesis we have
  \begin{itemize}
          \item $ \constraintfont{ \ottsym{(}   \ottnt{U}  \cup  \ottnt{U}_{{\mathrm{0}}}   \ottsym{,}  \ottnt{L}_{{\mathrm{0}}}  \ottsym{)} }   \vdash   \constraintfont{ \constraintfont{C} } $
  \end{itemize}
  Then by
  \rref*{C-Impl} $ \constraintfont{ \ottsym{(}  \ottnt{U}  \ottsym{,}  \emptyset  \ottsym{)} }   \vdash   \constraintfont{   \multiplicityfont{ \omega }  \scale( \ottsym{(}  \ottnt{U}_{{\mathrm{0}}}  \ottsym{,}  \ottnt{L}_{{\mathrm{0}}}  \ottsym{)}  \Lolly  \constraintfont{C} )  } $ and finally by
  \rref{C-Tensor} we have $ \constraintfont{ \ottsym{(}  \ottnt{U}  \ottsym{,}  \ottnt{L}_{\ottmv{i}}  \ottsym{)} }   \vdash   \constraintfont{   \multiplicityfont{ \omega }  \scale( \ottsym{(}  \ottnt{U}_{{\mathrm{0}}}  \ottsym{,}  \ottnt{L}_{{\mathrm{0}}}  \ottsym{)}  \Lolly  \constraintfont{C} )   \qtensor  \ottsym{(}  \emptyset  \ottsym{,}  \ottnt{L}_{\ottmv{i}}  \ottsym{)} } $.
  $ \constraintfont{ \ottnt{L}_{\ottmv{i}} }   \subseteq   \constraintfont{ \ottnt{L}_{\ottmv{i}} } $ holds trivially.
  \end{description}
\end{proof}

\begin{lemma}[Weakening of wanteds]\label{lem:wanteds:weakening}
  If $ \constraintfont{ \ottnt{Q} }   \vdash   \constraintfont{ \constraintfont{C} } $, then $ \constraintfont{   \multiplicityfont{ \omega }  \scale \ottnt{Q'}   \qtensor  \ottnt{Q} }   \vdash   \constraintfont{ \constraintfont{C} } $
\end{lemma}
\begin{proof}
  This is proved by a straightforward induction on the derivation of
  $ \constraintfont{ \ottnt{Q} }   \vdash   \constraintfont{ \constraintfont{C} } $, using the corresponding property on the
  simple-constraint entailment relation from
  \cref{def:entailment-relation}, for the \rref*{C-Dom} case.
\end{proof}

\begin{lemma}\label{lem:wanteds:module-action}
  The following equality holds: $  \constraintfont{   \multiplicityfont{ \pi }  \scale \ottsym{(}    \multiplicityfont{ \rho }  \scale \constraintfont{C}   \ottsym{)}  }  =  \constraintfont{   \multiplicityfont{ \ottsym{(}   \pi {⋅} \rho   \ottsym{)} }  \scale \constraintfont{C}  }  $.
\end{lemma}
\begin{proof}
  This is proved by a straightforward induction on the structure of
  $  \constraintfont{ \constraintfont{C} }  $, using \cref{lem:simples:monoid-action} for the case
  $  \constraintfont{ \constraintfont{C} }  =  \constraintfont{ \ottnt{Q} }  $.
\end{proof}
\end{document}

%% file: newunicodedefs.tex
\newunicodechar{ }{{\,}}
\newunicodechar{¬}{\ensuremath{\neg}} %
\newunicodechar{²}{^2}
\newunicodechar{³}{^3}
\newunicodechar{·}{\ensuremath{\cdot}}
\newunicodechar{¹}{^1}
\newunicodechar{×}{\ensuremath{\times}} %
\newunicodechar{÷}{\ensuremath{\div}} %
\newunicodechar{ʲ}{^j}
\newunicodechar{ʳ}{^r}
\newunicodechar{ˡ}{^l}
\newunicodechar{̷}{\not} %
\newunicodechar{Γ}{\ensuremath{\Gamma}}   %
\newunicodechar{Δ}{\ensuremath{\Delta}} %
\newunicodechar{Η}{\ensuremath{\textrm{H}}} %
\newunicodechar{Θ}{\ensuremath{\Theta}} %
\newunicodechar{Λ}{\ensuremath{\Lambda}} %
\newunicodechar{Ξ}{\ensuremath{\Xi}} %
\newunicodechar{Π}{\ensuremath{\Pi}}   %
\newunicodechar{Σ}{\ensuremath{\Sigma}} %
\newunicodechar{Φ}{\ensuremath{\Phi}} %
\newunicodechar{Ψ}{\ensuremath{\Psi}} %
\newunicodechar{Ω}{\ensuremath{\Omega}} %
\newunicodechar{α}{\ensuremath{\mathnormal{\alpha}}}
\newunicodechar{β}{\ensuremath{\beta}} %
\newunicodechar{γ}{\ensuremath{\mathnormal{\gamma}}} %
\newunicodechar{δ}{\ensuremath{\mathnormal{\delta}}} %
\newunicodechar{ε}{\ensuremath{\mathnormal{\varepsilon}}} %
\newunicodechar{ζ}{\ensuremath{\mathnormal{\zeta}}} %
\newunicodechar{η}{\ensuremath{\mathnormal{\eta}}} %
\newunicodechar{θ}{\ensuremath{\mathnormal{\theta}}} %
\newunicodechar{ι}{\ensuremath{\mathnormal{\iota}}} %
\newunicodechar{κ}{\ensuremath{\mathnormal{\kappa}}} %
\newunicodechar{λ}{\ensuremath{\mathnormal{\lambda}}} %
\newunicodechar{μ}{\ensuremath{\mathnormal{\mu}}} %
\newunicodechar{ν}{\ensuremath{\mathnormal{\mu}}} %
\newunicodechar{ξ}{\ensuremath{\mathnormal{\xi}}} %
\newunicodechar{π}{\ensuremath{\mathnormal{\pi}}}
\newunicodechar{π}{\ensuremath{\mathnormal{\pi}}} %
\newunicodechar{ρ}{\ensuremath{\mathnormal{\rho}}} %
\newunicodechar{σ}{\ensuremath{\mathnormal{\sigma}}} %
\newunicodechar{τ}{\ensuremath{\mathnormal{\tau}}} %
\newunicodechar{φ}{\ensuremath{\mathnormal{\varphi}}} %
\newunicodechar{χ}{\ensuremath{\mathnormal{\chi}}} %
\newunicodechar{ψ}{\ensuremath{\mathnormal{\psi}}} %
\newunicodechar{ω}{\ensuremath{\mathnormal{\omega}}} %
\newunicodechar{ϕ}{\ensuremath{\mathnormal{\phi}}} %
\newunicodechar{ϵ}{\ensuremath{\mathnormal{\epsilon}}} %
\newunicodechar{ᵏ}{^k}
\newunicodechar{ᵢ}{\ensuremath{_i}} %
\newunicodechar{ }{\quad}
\newunicodechar{ }{{\,}}
\newunicodechar{′}{\ensuremath{^\prime}}  %
\newunicodechar{″}{\ensuremath{^\second}} %
\newunicodechar{‴}{\ensuremath{^\third}}  %
\newunicodechar{ⁱ}{^i}
\newunicodechar{⁵}{\ensuremath{^5}}
\newunicodechar{⁺}{\ensuremath{^+}} 
\newunicodechar{⁻}{\ensuremath{^-}} 
\newunicodechar{ⁿ}{^n}
\newunicodechar{₀}{\ensuremath{_0}} %
\newunicodechar{₁}{\ensuremath{_1}} 
\newunicodechar{₂}{\ensuremath{_2}} 
\newunicodechar{₃}{\ensuremath{_3}}
\newunicodechar{₊}{\ensuremath{_+}} 
\newunicodechar{₋}{\ensuremath{_-}} 
\newunicodechar{ₙ}{_n} %
\newunicodechar{ℂ}{\ensuremath{\mathbb{C}}} %
\newunicodechar{ℒ}{\ensuremath{\mathscr{L}}}
\newunicodechar{ℕ}{\mathbb{N}} %
\newunicodechar{ℚ}{\ensuremath{\mathbb{Q}}}
\newunicodechar{ℝ}{\ensuremath{\mathbb{R}}} %
\newunicodechar{ℤ}{\ensuremath{\mathbb{Z}}} %
\newunicodechar{ℳ}{\mathscr{M}}
\newunicodechar{⅋}{\ensuremath{\parr}} %
\newunicodechar{←}{\ensuremath{\leftarrow}} %
\newunicodechar{↑}{\ensuremath{\uparrow}} %
\newunicodechar{→}{\ensuremath{\rightarrow}} %
\newunicodechar{↔}{\ensuremath{\leftrightarrow}} %
\newunicodechar{↖}{\nwarrow} %
\newunicodechar{↗}{\nearrow} %
\newunicodechar{↝}{\ensuremath{\leadsto}}
\newunicodechar{↦}{\ensuremath{\mapsto}}
\newunicodechar{⇆}{\ensuremath{\leftrightarrows}} %
\newunicodechar{⇐}{\ensuremath{\Leftarrow}} %
\newunicodechar{⇒}{\ensuremath{\Rightarrow}} %
\newunicodechar{⇔}{\ensuremath{\Leftrightarrow}} %
\newunicodechar{∀}{\ensuremath{\forall}}   %
\newunicodechar{∂}{\ensuremath{\partial}}
\newunicodechar{∃}{\ensuremath{\exists}} %
\newunicodechar{∅}{\ensuremath{\varnothing}} %
\newunicodechar{∈}{\ensuremath{\in}}
\newunicodechar{∉}{\ensuremath{\not\in}} %
\newunicodechar{∋}{\ensuremath{\ni}}  %
\newunicodechar{∎}{\ensuremath{\qed}}
\newunicodechar{∏}{\prod}
\newunicodechar{∑}{\sum}
\newunicodechar{∗}{\ensuremath{\ast}} %
\newunicodechar{∘}{\ensuremath{\circ}} %
\newunicodechar{∙}{\ensuremath{\bullet}} 
\newunicodechar{∞}{\ensuremath{\infty}} %
\newunicodechar{∣}{\ensuremath{\mid}} %
\newunicodechar{∧}{\wedge}%
\newunicodechar{∨}{\vee}%
\newunicodechar{∩}{\ensuremath{\cap}} %
\newunicodechar{∪}{\ensuremath{\cup}} %
\newunicodechar{∫}{\int}
\newunicodechar{∷}{::} %
\newunicodechar{∼}{\ensuremath{\sim}} %
\newunicodechar{≃}{\ensuremath{\simeq}} %
\newunicodechar{≅}{\ensuremath{\cong}} %
\newunicodechar{≈}{\ensuremath{\approx}} %
\newunicodechar{≜}{\ensuremath{\stackrel{\scriptscriptstyle {\triangle}}{=}}} 
\newunicodechar{≝}{\ensuremath{\stackrel{\scriptscriptstyle {\text{def}}}{=}}}
\newunicodechar{≟}{\ensuremath{\stackrel {_\text{\textbf{?}}}{\text{\textbf{=}}\negthickspace\negthickspace\text{\textbf{=}}}}} 
\newunicodechar{≠}{\ensuremath{\neq}}%
\newunicodechar{≡}{\ensuremath{\equiv}}%
\newunicodechar{≤}{\ensuremath{\le}} %
\newunicodechar{≥}{\ensuremath{\ge}} %
\newunicodechar{⊂}{\ensuremath{\subset}} %
\newunicodechar{⊃}{\ensuremath{\supset}} %
\newunicodechar{⊆}{\ensuremath{\subseteq}} %
\newunicodechar{⊇}{\ensuremath{\supseteq}} %
\newunicodechar{⊎}{\ensuremath{\uplus}} %
\newunicodechar{⊑}{\ensuremath{\sqsubseteq}} %
\newunicodechar{⊒}{\ensuremath{\sqsupseteq}} %
\newunicodechar{⊓}{\ensuremath{\sqcap}} %
\newunicodechar{⊔}{\ensuremath{\sqcup}} %
\newunicodechar{⊕}{\ensuremath{\oplus}} %
\newunicodechar{⊗}{\ensuremath{\otimes}} %
\newunicodechar{⊛}{\ensuremath{\circledast}}
\newunicodechar{⊢}{\ensuremath{\vdash}} %
\newunicodechar{⊤}{\ensuremath{\top}}
\newunicodechar{⊥}{\ensuremath{\bot}} 
\newunicodechar{⊧}{\models} %
\newunicodechar{⊨}{\models} %
\newunicodechar{⊩}{\Vdash}
\newunicodechar{⊸}{\ensuremath{\multimap}} %
\newunicodechar{⋁}{\ensuremath{\bigvee}}
\newunicodechar{⋃}{\ensuremath{\bigcup}} %
\newunicodechar{⋄}{\ensuremath{\diamond}} %
\newunicodechar{⋅}{\ensuremath{\cdot}}
\newunicodechar{⋆}{\ensuremath{\star}} %
\newunicodechar{⋮}{\ensuremath{\vdots}} %
\newunicodechar{⋯}{\ensuremath{\cdots}} %
\newunicodechar{─}{---}
\newunicodechar{■}{\ensuremath{\blacksquare}} 
\newunicodechar{□}{\ensuremath{\square}} %
\newunicodechar{▵}{\ensuremath{\triangle}} %
\newunicodechar{▹}{\ensuremath{\rhd}} %
\newunicodechar{◃}{\triangleleft{}}
\newunicodechar{◅}{\ensuremath{\triangleleft{}}}
\newunicodechar{◇}{\ensuremath{\diamond}} %
\newunicodechar{◽}{\ensuremath{\square}}
\newunicodechar{★}{\ensuremath{\star}}   %
\newunicodechar{♭}{\ensuremath{\flat}} %
\newunicodechar{♯}{\ensuremath{\sharp}} %
\newunicodechar{✓}{\ensuremath{\checkmark}} %
\newunicodechar{⟂}{\ensuremath{^\bot}} 
\newunicodechar{⟦}{\ensuremath{\llbracket}}
\newunicodechar{⟧}{\ensuremath{\rrbracket}}
\newunicodechar{⟨}{\ensuremath{\langle}} %
\newunicodechar{⟩}{\ensuremath{\rangle}} %
\newunicodechar{⟶}{{\longrightarrow}}
\newunicodechar{⟷} {\ensuremath{\leftrightarrow}}
\newunicodechar{⟹}{\ensuremath{\Longrightarrow}} %
\newunicodechar{ⱼ}{_j} %
\newunicodechar{𝒟}{\ensuremath{\mathcal{D}}} %
\newunicodechar{𝒢}{\ensuremath{\mathcal{G}}}
\newunicodechar{𝒦}{\ensuremath{\mathcal{K}}} %
\newunicodechar{𝒫}{\ensuremath{\mathcal{P}}}
\newunicodechar{𝔸}{\ensuremath{\mathds{A}}} %
\newunicodechar{𝔹}{\ensuremath{\mathds{B}}} %
\newunicodechar{𝟙}{\ensuremath{\mathds{1}}} %

%% file: ott.tex
\newenvironment{ottdefnblock}[3][]{ \framebox{\mbox{#2}} \quad #3 \\[0pt]}{}

\newcommand{\ottnt}[1]{\mathit{#1}}
\newcommand{\ottmv}[1]{\mathit{#1}}
\newcommand{\ottkw}[1]{\mathbf{#1}}
\newcommand{\ottsym}[1]{#1}

\newcommand{\ottafterlastrule}{\\}

